\documentclass[11pt]{article}

\usepackage[margin=1in]{geometry}
\usepackage{amsmath,amsthm,amssymb}
\usepackage{environ, tabularx}
\usepackage{tikz-cd,graphicx}
\usepackage{float}
\usepackage{indentfirst}
\usepackage{mathtools}
\usepackage{hyperref}
\usepackage{amsmath}
\usepackage{dirtytalk}
\usepackage{pifont}
\usepackage{soul}
\usepackage{verbatim}
\usepackage{subcaption} 
\usepackage{apxproof}
\usepackage{thmtools}
\usepackage{thm-restate}
\usepackage{algorithm}
\usepackage{algpseudocode}
\usepackage{cleveref}
\usepackage{graphbox}

\newcommand{\R}{\mathbb{R}}
\newcommand{\Z}{\mathbb{Z}}

\newcommand{\K}{\mathcal{K}}

\newcommand{\U}{\mathcal{U}}
\newcommand{\J}{\mathcal{J}}
\newcommand{\RE}{\mathcal{R}}

\newcommand{\Jsf}{\mathsf{J}}
\newcommand{\C}{\mathcal{C}}
\renewcommand{\P}{\mathcal{P}}
\renewcommand{\H}{\mathcal{H}}

\newcommand{\opt}{\text{opt}}

\definecolor{BlackBlue}{rgb}{0.38, 0.44, 1.0}

\newcommand{\boundary}{\partial}
\newcommand{\coboundary}{\delta}
\newcommand{\tensor}{\otimes}
\newcommand{\oracle}{\mathcal{O}}

\newcommand{\EMPH}[1]{\textit{\textbf{#1}}}
\renewcommand{\L}{\mathcal{L}}

\newcommand{\prism}[1]{P{#1}}
\renewcommand{\sp}[1]{SP{#1}}
\renewcommand{\star}[1]{S{#1}}
\newcommand{\id}{\mathbb{I}}

\DeclareMathOperator{\spec}{spec}
\newcommand{\specnz}{\spec_{NZ}}

\newtheorem{theorem}{Theorem}[section]
\newtheorem{lemma}[theorem]{Lemma}
\newtheorem{corollary}[theorem]{Corollary}

\newtheorem{definition}{Definition}
\newtheorem{observation}[theorem]{Observation}
\newtheorem*{problem}{Problem}

\DeclareMathOperator{\im}{im}

\DeclareMathOperator{\supp}{supp}
\DeclareMathOperator{\spn}{span}

\DeclareMathOperator{\poly}{poly}
\DeclareMathOperator{\polylog}{polylog}
\DeclareMathOperator{\diag}{diag}

\DeclarePairedDelimiter\bra{\langle}{\rvert}
\DeclarePairedDelimiter\ket{\lvert}{\rangle}
\DeclarePairedDelimiterX\braket[2]{\langle}{\rangle}{#1 \delimsize\vert #2}

\author{Mitchell Black\thanks{Oregon State University: blackmit@oregonstate.edu} \and William Maxwell\thanks{Oregon State University} \and Amir Nayyeri\thanks{Oregon State University: nayyeria@oregonstate.edu}}

\title{An Incremental Span-Program-Based Algorithm and the Fine Print of Quantum Topological Data Analysis\footnote{A preliminary version of this paper appeared at ISAAC 2021 under the title ``Effective Resistance and Capacitance in Simplicial Complexes and a Quantum Algorithm.''}}
\date{}

\begin{document}

\maketitle

\begin{abstract}
        We introduce a new quantum algorithm for computing the Betti numbers of a simplicial complex. In contrast to previous quantum algorithms that work by estimating the eigenvalues of the combinatorial Laplacian, our algorithm is an instance of the generic Incremental Algorithm for computing Betti numbers that incrementally adds simplices to the simplicial complex and tests whether or not they create a cycle. In contrast to existing quantum algorithms for computing Betti numbers that work best when the complex has close to the maximal number of simplices, our algorithm works best for sparse complexes. 
        \par 
        To test whether a simplex creates a cycle, we introduce a quantum span-program algorithm. We show that the query complexity of our span program is parameterized by quantities called the effective resistance and effective capacitance of the boundary of the simplex. Unfortunately, we also prove upper and lower bounds on the effective resistance and capacitance, showing both quantities can be exponentially large with respect to the size of the complex, implying that our algorithm would have to run for exponential time to exactly compute Betti numbers.
        \par 
        However, as a corollary to these bounds, we show that the spectral gap of the combinatorial Laplacian can be exponentially small. As the runtime of all previous quantum algorithms for computing Betti numbers are parameterized by the inverse of the spectral gap, our bounds show that all quantum algorithms for computing Betti numbers must run for exponentially long to exactly compute Betti numbers.
        \par
        Finally, we prove some novel formulas for effective resistance and effective capacitance to give intuition for these quantities.
\end{abstract}

\section{Introduction.}

The past few years has seen the development of quantum algorithms with the potential to speed up computation of topological features of simplicial complexes called \textit{Betti numbers}. Betti numbers are important topological invariants of a space; indeed, there is an entire, rapidly-growing field called \textit{Topological Data Analysis} (TDA) that studies the application of topological invariants like Betti numbers (among other)~\cite{carlsson2021topological,dey2022computational, EdelsbrunnerHarer}. Accordingly, the study of quantum algorithms for computing Betti numbers has been deemed \textit{Quantum Topological Data Analysis} (QTDA). 
\par 
Betti numbers can be both time and space inefficient for classical computers to compute. For example, a simplicial complex on vertices can be exponentially large and it can take exponential time to compute Betti numbers in arbitrary dimensions. Quantum computers offer a potential solution to the shortcomings of the classical algorithm. For example, quantum computers can efficiently store a simplicial complex with $n$ vertices using only $O(\poly(n))$ qubits.
\par 
However, while these quantum algorithms have certain advantages over their classical counterparts like improved space complexity, QTDA algorithms only achieve significant advantage over classical TDA algorithms under certain circumstances. QTDA algorithms only achieve significant speed up over classical algorithms when the input complex is \textit{clique-dense}---it has close to the maximal number of simplices---and when the spectral gap of the combinatorial Laplacian of the complex is polynomially small. This second point is a particular problem as, before now, it was unknown how small the spectral gap of the combinatorial Laplacian could be. This makes the spectral gap of the combinatorial Laplacian an example of ``fine print''~\cite{Aaronson2015fineprint}: an unbounded parameter in the runtime of a celebrated quantum algorithm.

\subsection{Our Contributions.}

\begin{itemize}
    \item In Sections \ref{sec:incremental_algorithm} and \ref{sec:quantum_algorithm_null_homology_testing}, we provide a novel quantum algorithm for computing Betti numbers using the framework of span programs~\cite{Karchmer, Reichardt2012}. As opposed to existing QTDA algorithms that work by estimating the eigenvalues of the combinatorial Laplacian or singular values of the boundary matrices, our algorithm is more similar to classical matrix-reduction algorithms for computing Betti numbers as it works by incrementally adding simplices to the simplicial complex and testing if these simplices create or destroy a cycle. One advantage of our algorithm is that it avoids the step of creating a superpostion over the $k$-simplices, which is a bottleneck of existing QTDA algorithms that restricts their utility to the clique-dense regime. In \Cref{sec:witness_sizes} and \Cref{sec:time_complexity}, we show that the query and time complexity of our span program algorithm for QTDA is parameterized by the maximum effective resistance and capacitance of cycles in $\K$. In \Cref{sec:comparision_existing_algorithms}, we compare our algorithm with existing QTDA algorithms. The culmination of this section is the following theorem. 
    
    \begin{restatable}{theorem}{thmincrementalalgorithmruntime}
    \label{thm:incremental_algorithm_runtime}
        Let $\K$ be a simplicial complex. There is a quantum algorithm for computing the $d$th Betti number $\beta_d$ of $\K$ in time 
        $$
            \Tilde{O}\left( \left( \sqrt{\frac{\RE_{\max}\C_{\max}}{\tilde{\lambda}_{\min}}} n_0 + \sqrt{\RE_{\max}n_0}\right)(n_{d} + n_{d+1})\right),  
        $$
        where
        \begin{itemize}
            \item[$\bullet$] $n_i$ is the number of $i$-simplices of $\K$.
            \item[$\bullet$] $\RE_{\max}$ is the maximum finite effective resistance $\RE_{\boundary\sigma}(\L)$ of the boundary of any $d$- or $(d+1)$-simplex $\sigma\in\K$ in any subcomplex $\L\subset\K$.
            \item[$\bullet$] $\C_{\max}$ are the maximum finite effective capacitance $\C_{\boundary\sigma}(\L,\K)$ of the boundary of any $d$- or $(d+1)$-simplex $\sigma\in\K$ in any subcomplex $\L\subset\K$.
            \item[$\bullet$] $\tilde{\lambda}_{\min}$ is the minimum spectral gap of the normalized up Laplacians $\tilde{L}_{d-1}^{up}[\K]$ and $\tilde{L}_{d}^{up}[\K]$.
        \end{itemize}
    \end{restatable}
    
    \item In \Cref{sec:bounds}, we provide upper bounds on the maximum effective resistance by the size of the simplicial complex and the maximal rank of the torsion subgroup of the simplicial complex, as well as looser upper bounds purely in terms of the size of the complex. These upper bounds show that the effective resistance can be at most exponentially-large with respect to the size of the complex. We also provide similar upper bounds on effective capacitance for special cases. Finally, we provide families of simplicial complexes with cycles whose effective resistance or effective capacitance is exponentially large, thus matching the upper bound up to the base of the exponent. This implies that our algorithm for QTDA can take exponentially long in the worst case. However, in the next paragraph, we will see that our results imply all other QTDA algorithms must run for exponential time as well.
    
    \item In \Cref{sec:bounds_spectral_gap}, we show how the upper and lowers bounds for effective resistance provide lower and upper bounds for the spectral gap of the combinatorial Laplacian respectively; thus, the spectral gap is exponentially small in the worst case.  Moreover, we show there are clique-dense complexes that achieve worst-case spectral gap.  
    
    \begin{restatable}{theorem}{spectralgaplowerbound}
    \label{thm:spectral_gap_lower_bound}
        Let $K$ be a simplicial complex. Let $n_i$ be the number of $i$-simplices of $K$. Let $n = \max\{\min\{n_{d-1},n_d\},\min\{n_d,n_{d+1}\}\}$.
        Then the spectral gap $\lambda_{\min}(L_d[K])\in\Omega\left(\frac{1}{n^{2}d^{n}}\right)$. 
    \end{restatable}

    \begin{restatable}{theorem}{thmspectralgapupperbounddense}
    \label{thm:spectral_gap_upper_bound_dense}
        Let $d,n\geq 1$. There are constants $c_d,\kappa_d$ that depends only on $d$ and a $d$-dimensional simplicial complex $\mathcal{C}_d^{n}$ with $n_d=\Omega(\kappa_d\binom{n_{0}}{d})$ $d$-simplices such that the spectral gaps $\lambda_{\min}(L_{d-1}[\mathcal{C}_d^{n}])$, $\lambda_{\min}(L_d[\mathcal{C}_d^{n}])\in O(\frac{1}{c_d^{n_d}})$.
    \end{restatable}

    This answers one of the most important question in QTDA: how small can the spectral gap be? As all existing QTDA algorithms are parameterized by the inverse of the spectral gap of the combinatorial Laplacian $\frac{1}{\lambda_{\min}}$, this implies that all existing QTDA algorithms need exponential time to exactly estimate Betti numbers.\footnotemark Additionally, the space complexity of some QTDA algorithms is parameterized by $\log(\frac{1}{\lambda_{\min}})$, so these algorithms will need space proportional to the number of $(d-1)-$, $d-$, or $(d+1)$-simplices, rather than space proportional to the number of vertices. 
    
    \footnotetext{This exponential time complexity is not a result of the quantum nature of these algorithms. Some classical algorithms for computing Betti numbers are also parameterized by the inverse of the spectral gap, so these algorithms would need to run for exponential time as well~\cite{apers2022simple, Friedman1998BettiNumbers}.}
    
    \item We also prove some interesting formulas for effective resistance and capacitance that give intuition for these quantities. In \Cref{sec:formulas}, we provide series and parallel formulas for effective resistance akin to the formulas for effective resistance in graphs. We also show that effective resistance satisfies a Rayeligh monotonicity property akin to effective resistance in graphs. Finally, in \Cref{sec:duality}, we show that effective resistance is dual to effective capacitance for embedded simplicial complexes.
\end{itemize}

\subsection{Related Work}.

\paragraph{History of QTDA.} Lloyd, Garnerone, and Zanardi (LGZ) introduced the first quantum algorithm for computing Betti numbers up to a multiplicative error~\cite{Lloyd2016lgz}. Their algorithm works by estimating the eigenvalues of the combinatorial Laplacian, which is inspired by Friedman's classical algorithm for computing Betti numbers~\cite{Friedman1998BettiNumbers}. The LGZ algorithm has the advantage that its runtime is only polynomial with respect to the number of vertices, as opposed to the number of simplices like the matrix reduction algorithm. The trade-off is that this algorithm gains a dependence on the inverse of the spectral gap and the ratio of the number of simplices to the number of possible simplices. The LGZ algorithm performs best in the regime where the spectral gap of the combinatorial Laplacian is polynomially lower-bounded and the simplicial complex is clique-dense, meaning it has close to the maximal number of simplices. Subsequent works have improved the LGZ algorithm in different ways but maintain a runtime dependence on the inverse of the spectral gap and the clique density~\cite{gunn2019, mccardle2023streamlined, ubaru2021}.
\par 
Another line of QTDA research has been developing algorithms for \textit{persistent} Betti numbers. While the LGZ algorithm was initially claimed to be able to compute persistent Betti numbers, this was later disproved by Meijer~\cite{Meijer2019} and Neumann and den Breeijen~\cite{Neumann2019}. Hayakawa was the first to develop a quantum algorithm for computing persistent Betti numbers~\cite{Hayakawa2022quantumalgorithm}. McArdle, Gily\'{e}n, and Berta have also developed algorithms for computing persistent Betti numbers~\cite{mccardle2023streamlined}. 

\paragraph{Hardness of Computing Betti Numbers.} In addition to new algorithms for computing Betti numbers, there have also been a number of works arguing computing Betti numbers is hard in general. Adamaszek and Stacho~\cite{adamaszek2016complexity} show that determining if a simplicial complex has non-zero Betti number is NP-Hard when parameterized either by the number of vertices and the number of maximal simplices, or the number of vertices and number of minimal non-faces. Additionally, they show the problem is NP-Hard for clique complexes when parameterized by the number of vertices. Schmidhuber and Lloyd~\cite{alex2022complexitytheoretic} show that computing Betti numbers of a clique complex is \#P-Hard and estimating the Betti number up to a multiplicative constant is NP-Hard when parameterized by the number of vertices.
Moreover, the hardness results of Schmidhuber and Lloyd hold for clique-dense clique complexes. This is an important restriction as the runtime of LGZ and and other QTDA algorithms are lowest for clique-dense complexes. Here, the assumptions on the input are vital. Computing Betti numbers is in $P$ when parameterized by the number of all simplices in the complex. This does not contradict $P\neq NP$ though, as the number of simplices can be exponentially large with respect to the number of vertices.
\par 
There have also been a number of works showing that problems related to computing Betti numbers are hard for the quantum computing complexity class DCQ-1. Crichigno and Kohler~\cite{crichigno2022clique} showed that determining if the Betti number of a clique complex  is nonzero is QMA\textsubscript{1}-Hard when parameterized by the number of vertices, and computing the Betti number of a clique complex is \#BQP-Hard. Gyurik, Cade, and Dunjko~\cite{gyurik2022advantage} show that a generalization of Betti number estimation called \textit{low-lying spectral density estimation (LLSD)} is DCQ1-Complete, suggesting that LLSD may be classically intractable. Cade and Crichigno~\cite{cade2021complexity} showed that estimating Betti numbers for general chain complexes (not just those arising from simplicial complexes) is also DCQ1-complete. 

\paragraph{Lower Bounds on the Spectral Gap of the Combinatorial Laplacian.} 
All existing QTDA algorithms are parameterized by the inverse of the spectral gap of the combinatorial Laplacian. While we show the spectral gap can be exponentially small, there have also been a number of exact or expected lower bounds on the spectral gap of the combinatorial Laplacian for certain families of simplicial complexes~\cite{beit2020spectral, Friedman1998BettiNumbers, gundert2016eigenvalues,  lew2020spectralb, lew2020spectral, lew2023garland, STEENBERGEN201456, yamada2019spectrum}. However, these bounds place non-trivial assumptions on the simplicial complex so should not be taken to represent general simplicial complexes. 

\section{Preliminaries.}
\label{sec:prelim}

\paragraph{Algebraic Topology.} A \EMPH{simplicial complex} $\K$ on a set of vertices $V$ is a subset of the power set $\K \subseteq P(V)$ with the property that if $\sigma \in \K$ and $\tau \subset \sigma$ then $\tau \in \K$. An element of $\K$ is a \textit{\textbf{simplex}}. A simplex $\sigma\in\K$ of size $|\sigma|=d+1$ is a \EMPH{d-simplex}.
The set of all $d$-simplices of $\K$ is denoted $\K_d$, and the number of $d$-simplices is denoted $n_d = |\K_d|$. The \EMPH{d-skeleton} of $\K$, denoted $\K^d$, is the simplicial complex of all simplices of $\K$ of dimension at most $d$, i.e. $\K^{d}=\cup_{i=0}^{d}\K_i$. The \EMPH{dimension} of $\K$ is the largest $d$ such that $\K$ contains a $d$-simplex; a 1-dimensional simplicial complex is a \EMPH{graph}. 
\par
The \EMPH{d\textsuperscript{th} chain group} $C_d(\K)$ is the vector space over $\R$ with orthonormal basis $\K_d$. An element of $C_d(\K)$ is a \EMPH{d-chain}. Unless otherwise stated, all vectors and matrices will be in the basis $\K_d$. For a chain $f\in C_d(\K)$, we denote its $\sigma$ coordinate $f(\sigma)$. Finally, the \EMPH{support} of a chain $f$ is the set of simplices given a non-zero value by $f$ and is denoted $\supp(f) = \{\sigma_i \in \K_d \colon f(\sigma_i) \neq 0\}$.
\par 
 We assume there is a fixed but arbitrary order on the vertices $V = (v_1, \dots, v_n)$. Let $\sigma=\{v_{i_0},\ldots,v_{i_d}\}$  be a $d$-simplex in $\K$ with $v_{i_j}\leq v_{i_k}$ whenever $j\leq k$. The \EMPH{boundary} of $\sigma$ is the $(d-1)$-chain $\boundary\sigma=\sum_{j=0}^{d}(-1)^{j}\cdot(\sigma\setminus\{v_{i_j}\})$.  The \EMPH{d\textsuperscript{th} boundary map} is the linear map $\boundary_d:C_d(\K)\to C_{d-1}(\K)$ defined $\boundary_d f=\sum_{\sigma\in \K_d} f(\sigma)\boundary\sigma$ where $f(\sigma)$ denotes the component of $f$ indexed by the simplex $\sigma$. An element in $\ker\boundary_d$ is a \EMPH{cycle}, and an element in $\im\boundary_d$ is a \EMPH{boundary} or a \EMPH{null-homologous cycle}. See \Cref{fig:null_homologous_and_not_null_homologous_cycles}.
The boundary maps have the property that $\boundary_d\circ\boundary_{d+1}=0$, so $\im\boundary_{d+1}\subset\ker\boundary_{d}$. The \EMPH{d\textsuperscript{th} homology group} is the quotient group $H_d(\K)=\ker(\boundary_{d})/\im(\boundary_{d+1})$. The \EMPH{d\textsuperscript{th} Betti number} $\beta_d$ is the dimension of $H_d(\K)$.
 The \EMPH{d\textsuperscript{th} coboundary map} is the map $\coboundary_d:=\boundary_{d+1}^T:C_{d}(\K)\to C_{d+1}(\K)$. An element of $\ker\coboundary_d$ is a \EMPH{cocycle}, and an element in $\im\coboundary_{d-1}$ is a \EMPH{coboundary}.
We will use the notation $\partial[\K]$ and $\delta[\K]$ when we want to specify the complex associated with the (co)boundary operator.
\par 
While our algorithms calculate homology with real coefficients, for some of our topological results, we will need to consider homology with integer coefficients. The \textit{\textbf{integral chain group}} $C_d(\K, \Z)$ is the free abelian group generated by the set of $d$-simplices $\K_d$ whose elements are formal sums $\sum_{\sigma_i \in \K_d} \alpha_i \sigma_i$ with coefficients $\alpha_i \in \Z$. The integer homology groups are constructed in the same way as the real homology groups. We define boundary maps $\boundary_d:C_d(\K;\Z)\to C_{d-1}(\K;\Z)$ the same way as for the real chain groups, except now the boundary maps are group homomorphisms rather than linear maps. The \textit{\textbf{integral homology groups}} are the quotient groups $H_d(\K; \Z)=\ker\boundary_{d}/\im\boundary_{d+1}$.
\par
\paragraph{Laplacians.} The \EMPH{d\textsuperscript{th} up Laplacian} is $L^{up}_{d}=\boundary_{d+1}\coboundary_{d}$, the \EMPH{d\textsuperscript{th} down Laplacian} is $L^{down}_d=\coboundary_{d-1}\boundary_d$, and the \EMPH{d\textsuperscript{th} (combinatorial) Laplacian} is $L_d = L_d^{up} + L_d^{down}$. The Laplacians define the following orthogonal decomposition of the $d$\textsuperscript{th} chain group $C_d(\K)$ called the \textit{\textbf{Hodge Decomposition}}.
\begin{align*}
    C_d(\K) &= \im L_d^{up} \oplus \im L_d^{down} \oplus \ker L_d \\
    &= \im\boundary_d \oplus \im\coboundary_{d-1} \oplus \ker L_d 
\end{align*}
where the second equality follows from the fact that $\im AA^{T} = \im A$ for any matrix $A$. We call the subspaces $\im\boundary_d$, $\im\coboundary_{d-1}$, and $\ker L_d$ the \textit{\textbf{boundary}}, \textit{\textbf{coboundary}}, and \textit{\textbf{harmonic spaces}}. Arguably the fundamental theorem of the combinatorial Laplacian is the \textit{\textbf{Hodge Theorem}}.

\begin{theorem}[Hodge Theorem, Eckmann~\cite{Eckmann1944}]
\label{thm:hodge}
     The $d$\textsuperscript{th} harmonic space is isomorphic to the $d$\textsuperscript{th} homology group, i.e. $\ker L_d\cong H_d(\K)$.
\end{theorem}
Therefore, the $d$\textsuperscript{th} Betti number can equivalently be computed by computing the rank of $L_d$, a fact used by many existing QTDA algorithms. 
\par 
The following lemma gives several properties of the spectrum of the combinatorial Laplacian. 

\begin{lemma}
\label{lem:properties_spectrum_laplacian}
    Let $\specnz(A)$ denote the multiset of the non-zero eigenvalues of a linear operator $A$. Let $\K$ be a simplicial complex. Let $d>0$ be a positive integer. Then
    \begin{enumerate}
        \item (Goldberg \cite[Lemma 4.1.8]{Goldberg2002}) $\specnz(L^{up}_{d})=\specnz(L^{down}_{d+1})$
        \item (Goldberg \cite[Lemma 4.1.7]{Goldberg2002}) $\specnz(L_{d}) = \specnz(L^{up}_{d})\cup \specnz(L^{down}_{d})$
        \item (Goldberg \cite[Lemma 4.2.3]{Goldberg2002}) If $\K$ has connected components $\K_1,\ldots,\K_m$, then $$\specnz(L_{d}[\K]) = \specnz(L_{d}[\K_1])\cup\cdots\cup \specnz(L_{d}[\K_m])$$
    \end{enumerate}
    where all unions are multiset unions.
\end{lemma}

 The up, down, and combinatorial Laplacian are all \textit{\textbf{positive-semidefinite}}, meaning their eigenvalues are all non-negative~\cite{Goldberg2002}. The \textit{\textbf{spectral gap}} $\lambda_{\min}(L_d)$ is the smallest non-zero eigenvalue of the combinatorial Laplacian. \Cref{lem:properties_spectrum_laplacian} Part 2 implies the following theorem about the spectral gap of the combinatorial Laplacian.

\begin{corollary}
\label{lem:spectral_gap_combinatorial_up_down}
    Let $\K$ be a simplicial complex. Let $d$ be a positive integer. Then
    $$
     \lambda_{\min}(L_{d}[\K]) = \min\{\lambda_{\min}(L^{down}_{d}[\K]),\, \lambda_{\min}(L^{up}_{d}[\K])\}
    $$
\end{corollary}

In \Cref{sec:bounds_spectral_gap}, we discuss upper and lower bounds on the spectral gap. There are also known upper and lower bounds on the \textit{largest} eigenvalue of the combinatorial Laplacian.

\begin{theorem}
\label{thm:upper_bound_maximum_eigenvalue_laplacian}
    Let $\K$ be a simplicial complex with $n_0$ vertices. Let $d$ be a natural number. Then the maximal eigenvalue of the combinatorial Laplacian $\lambda_{\max}(L_d)\leq n_0$.
\end{theorem}
\begin{proof}
    Let $\Delta_{n_0}$ be the complete complex on $n_0$ vertices. The maximum eigenvalue of the $d$\textsuperscript{th} up Laplacian is $\lambda_{\max}(L_{d}^{up}[\Delta_{n_0}])=n_0$ for any dimension $d$~\cite[Lemma 2.6]{GundertThesis2013}. Moreover, by the interlacing theorem of eigenvalues of the up Laplacian, $\lambda_{\max}(L_{d}^{up}[\K])\leq \lambda_{\max}(L_{d}^{up}[\Delta_{n_0}])$ for any subcomplex $\K\subset \Delta_{n_0}$~\cite[Theorem 1.1]{horak2013interlacing}. The theorem follows as $\lambda_{\max}(L_d[\K]) = \max\{\lambda_{\max}(L_d^{up}[\K]),$ $\lambda_{\max}(L_{d-1}^{up}[\K])\}\leq n_0$
\end{proof}

We also consider two variants of the Laplacian variants of the up-Laplacian: the weighted up Laplacian and the normalized up Laplacian. Let $w:\K_{d+1} \to \R^+$ be a weight function on the $(d{+}1)$-simplices. Let $W:C_{d+1}(\K)\to C_{d+1}(\K)$ be the diagonal matrix with $W_{\tau,\tau}=w(\tau)$. The \EMPH{d\textsuperscript{th} weighted up Laplacian} is $L^{up,\,W}_d=\boundary_{d+1} W \coboundary_d$. The \EMPH{degree} of a $d$-simplex $\sigma$ is $\deg(\sigma) = \sum_{\tau\in\K_{d+1}\colon \sigma\subset\tau}w(\tau)$. Let $D:C_d(\K)\to C_d(\K)$ be the diagonal matrix with $D_{\sigma,\sigma}=\deg(\sigma)$. The \EMPH{d\textsuperscript{th} normalized up Laplacian} is $\Tilde{L}^{up}_d = D^{-1/2} \boundary_{d+1} W \coboundary_d D^{-1/2}$. The following theorem relates the spectral gap of the normalized and unnormalized Laplacians. A proof can be found in \Cref{apx:normalized_vs_unnoramlized}

\begin{restatable}{lemma}{lemnormalizedvsunnormalizedspectralgap}
\label{lem:normalized_vs_unnormalized_spectral_gap}
    Let $\K$ be a simplicial complex. Let $d_{\min}$ and $d_{\max}$ be the minimum and maximum degrees of any $d$-simplex in $\K$. Suppose that $d_{\min}>0$. The normalized and unnormalized spectral gap are related as follows:
    $$
         \frac{1}{d_{\max}} \lambda_{\min}(L_d^{up}) \leq \lambda_{\min}(\tilde{L}_d^{up}) \leq \frac{1}{d_{\min}} \lambda_{\min}(L_d^{up})
    $$
\end{restatable}

\paragraph{Pseudoinverse of a Linear Map.}

Let $A:\R^{n}\to\R^{m}$ be a rank $k$ linear operator with singular value decomposition $A = \sum_{i=1}^{k} \sigma_{i}u_iv_i^{T}$. The \textit{\textbf{pseudoinverse}} of $A$ is the linear operator $A^+:\R^{m}\to\R^{n}$ defined $ A^{+}= \sum_{i=1}^{k} \sigma^{-1}_{i}u_iv_i^{T}$. While this is in the most compact definition of the pseudoinverse, it is not the most informative. Equivalently, the \textit{\textbf{pseudoinverse}} of $A:\R^{m}\to\R^{n}$ is the unique linear operator with the following properties: (1) $A^{+}$ maps each vector $x\in\im A$ to the unique vector $y\in\im A^{T}$ such that $Ay=x$ and (2) $A^{+}$ maps each vector in $(\im A)^{\perp}$ to 0. The following are well-known properties of the pseudoinverse that follow from these definitions

\begin{lemma}
\label{lem:properties_pseudoinverse}
    Let $A:\R^{m}\to\R^{n}$ be a linear map.
    \begin{enumerate}
        \item $(AA^{T})^{+} = (A^{T})^{+}A^{+}$.
        \item For $x\in\im A$, $A^{+}x = \arg\min\{\|y\| : Ay=x\}$
    \end{enumerate}
\end{lemma}

\paragraph{Bra-Ket Notation.} When discussing quantum algorithms, we will use \EMPH{bra-ket notation} for vectors. As this paper may also be of interest to topologists who may be unfamiliar with this notation, we introduce bra-ket notation now. Assuming a fixed basis for a finite vector space, a \EMPH{bra} is a row vector represented by the notation $\bra{v}$. A \EMPH{ket} is a column vector represented by the notation $\ket{v}$. Using bras and kets, we can represent an inner product as $\braket{u}{v}$, an outer product as $\ket{u}\bra{v}$, or a tensor product as $\ket{u}\ket{v}$. 

\section{The Incremental Algorithm for Computing Betti Numbers.}
\label{sec:incremental_algorithm}

\begin{figure}[t]
    \centering
    \begin{subfigure}{0.45\textwidth}
        \centering
        \includegraphics[height=1.5in]{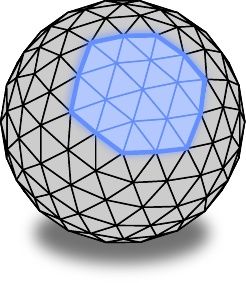}    
        \vspace*{0.05in}
    \end{subfigure}
    \begin{subfigure}{0.45\textwidth}
        \centering
        \includegraphics[height=1.5in]{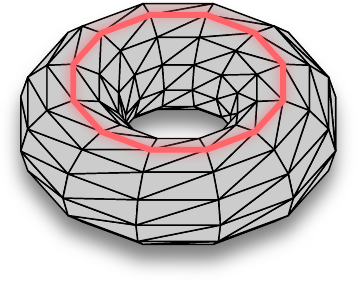}    
    \end{subfigure}
    \caption{Left: The 1-cycle is null-homologous as it is the boundary of the pictured 2-chain. Right: The 1-cycle is not null-homologous as it is not the boundary of any 2-chain. Coefficients on colored simplices are $\pm 1$. Orientations on the simplices have been omitted for simplicity.}
   \label{fig:null_homologous_and_not_null_homologous_cycles}
\end{figure}

In this section, we review the incremental algorithm for computing Betti numbers introduced by Delfinado and Edelsbrunner~\cite{Delfinado1993}. The incremental algorithm is a generic framework for computing Betti numbers based around the primitive of \textit{\textbf{null-homology testing}}. See \Cref{fig:null_homologous_and_not_null_homologous_cycles}
\begin{problem}[Null-Homology Testing]
    Given a simplicial complex $\K$ and a cycle $\gamma$ in $\K$, determine if $\gamma$ is null-homologous.
\end{problem}
\par 
The \textit{\textbf{Incremental Algorithm for Computing Betti Numbers}} computes $\beta_d$ by testing if the boundary of simplices are null-homologous. Specifically, the incremental algorithm incrementally adds $d$-simplices to the simplicial complex and and then performs a null-homology test on their boundaries to see how the dimension of the spaces $\ker\boundary_d$ and $\im\boundary_{d+1}$ change. 
\par 
To see how this works, fix an order on the $d$-simplices $\K_d = \{\sigma_1,\ldots,\sigma_d\}$, and then iteratively add each simplex $\sigma_i$ in increasing order of $i$. Adding the simplex $\sigma_i$ will either increase the dimension of $\ker\boundary_d$ or $\im\boundary_{d-1}$ by 1. If $\boundary\sigma_i$ is null-homologous in $\K^{d-1}\cup\{\sigma_1,\ldots,\sigma_{i-1}\}$, then adding $\sigma_i$ will increase the dimension of $\ker\boundary_d$ by 1, which also increases $\beta_d$ by 1. If not, then adding $\sigma_i$ will decrease $\im\boundary_{d-1}$ by 1, which also decreases $\beta_{d-1}$ by 1.\footnote{In their original paper on the incremental algorithm~\cite{Delfinado1993}, Delfinado and Edelsbrunner phrase this slightly differently as testing ``whether $\sigma_i$ is in [the support of] a cycle.'' It is straightforward to verify that $\sigma_i$ is in the support of a cycle if and only if $\boundary\sigma_i$ is null-homologous.} The incremental algorithm is summarized in~\Cref{alg:incremental_algorithm}. 

\begin{algorithm}[ht]
\caption{The incremental algorithm for computing the $d$\textsuperscript{th} Betti Number $\beta_d$~\cite{Delfinado1993}}
\label{alg:incremental_algorithm}
\begin{algorithmic}
    \State $\beta_d \gets 0$
    \For{$\sigma_i\in\K_{d}=\{\sigma_1,\ldots,\sigma_{n_d}\}$}
        \If{$\boundary\sigma_i$ is null-homologous in $\K^{d-1}\cup\{\sigma_1,\ldots,\sigma_{i-1}\}$}
            \State $\beta_d\gets \beta_d+1$ 
        \EndIf
    \EndFor
    \For{$\tau_j\in\K_{d+1}=\{\tau_1,\ldots,\tau_{n_{d+1}}\}$}
        \If{$\boundary\tau_j$ is not null-homologous in $\K^{d}\cup\{\tau_1,\ldots,\tau_{j-1}\}$}
            \State $\beta_d\gets \beta_d-1$
        \EndIf
    \EndFor
    \State Return $\beta_d$
\end{algorithmic}
\end{algorithm}

In the classical matrix reduction algorithm for computing Betti numbers, testing whether $\boundary\sigma_i$ is null-homologous is done by reducing the column corresponding to $\sigma_i$ in the boundary matrix, which takes $O(n_{d-1} n_{d})$ time in the worst case. However, there are special cases where null-homology testing can be performed much more quickly. For example, when a simplicial complex is embedded in $\R^{3}$ or the 3-sphere, null-homology testing can be performed in nearly-linear time using the union-find algorithm~\cite{Delfinado1993}. In the next section, we give a quantum algorithm for null-homology testing.

\section{A Quantum Algorithm for Null-Homology Testing.}
\label{sec:quantum_algorithm_null_homology_testing}

In this section, we provide a quantum algorithm based on the \textit{span program} model to decide whether or not a cycle $\gamma$ is null-homologous in a simplicial complex $\K$.
\par
Our algorithm is a generalization of the quantum algorithm developed by Belovs and Reichardt to decide $st$-connectivity in a graph~\cite{Belovs2012}. Their algorithm is parameterized by the effective resistance and capacitance between the vertices $s$ and $t$. The query complexity of our algorithm is parameterized by higher-dimensional analogues of effective resistance and capacitance of $\gamma$ that we introduce in \Cref{sec:res_and_cap}.
\par 
Upper bounds on the effective resistance and capacitance in graphs imply a query complexity of $O(n^{3/2})$ for $st$-connectivity, where $n$ is the number of vertices~\cite{jarret_et_al:LIPIcs:2018:9512}. In \Cref{sec:bounds}, we provide upper bounds on the effective resistance and capacitance. Our upper bounds on effective resistance and capacitance imply that the query complexity is polynomial in both the number of $d$-simplices as well as the cardinality of the largest torsion subgroup of a relative homology group of $\K$. In the case that $\K$ is a graph, our analysis of the witness sizes matches the $O(n^{3/2})$ upper bounds of previous analyses. Specifically, under the assumptions that $\K$ is relative torsion free and that $\gamma$ is the boundary of a $d$-simplex (which may or may not be included in the complex), we match the $O(n^{3/2})$ upper bound. These assumptions are always true for $st$-connectivity in graphs, which is why we match the query complexity for this problem. However, in \Cref{sec:lower_bounds}, we provide examples of simplicial complexes where the effective resistance or capacitance of $\gamma$ is exponentially large.

\subsection{A Brief Introduction to Span Programs.}

Span programs were first defined by Karchmer and Wigderson \cite{Karchmer} and were first used for quantum algorithms by Reichardt and \v{S}palek \cite{Reichardt2012}. Intuitively, a span program is a model of computation which encodes a boolean function $f \colon \{0, 1\}^n\rightarrow \{0, 1\}$ into the geometry of two vector spaces and a linear operator between them. Encoding $f$ into a span program implies the existence of a quantum query algorithm evaluating $f$ (Theorem~\ref{thm:reichardt}.)

\begin{definition}
A \textbf{span program} $\P = (\H, \U, \ket{\tau}, A)$ over the set of strings $\{0, 1\}^n$ is a 4-tuple consisting of:
\begin{enumerate}
\item A finite dimensional Hilbert space $\H = \H_1 \oplus \dots \oplus \H_n$ where $\H_i = \H_{i, 0} \oplus \H_{i, 1}$,
\item a vector space $\U$,
\item a non-zero vector $\ket{\tau} \in \U$, called the \textit{\textbf{target vector}}
\item a linear operator $A \colon \H \rightarrow \U$.
\end{enumerate}
For every string $x = (x_1,\dots,x_n) \in \{0, 1\}^n$ we associate the Hilbert space $\H(x) = \H_{1, x_1} \oplus \dots \oplus \H_{N, x_n}$ and the linear operator $A(x) = A \Pi_{\H(x)} \colon \H \rightarrow \U$ where $\Pi_{\H(x)}$ is the projection of $\H$ onto $\H(x)$. A string $x\in\{0,1\}^n$ is a \textit{\textbf{positive instance}} if $\ket{\tau}\in\im A(x)$ and is a \textit{\textbf{negative witness}} otherwise.
\end{definition}

A span program $\P$ \textit{\textbf{decides}} the function $f \colon \{0,1\}^n\to\{0,1\}$ if $f(x)=1$ when $x$ is a positive instance and $f(x)=0$ when $x$ is a negative instance. A span program can also evaluate a partial boolean function $g \colon D\to\{0,1\}$ where $D\subset\{0,1\}^n$ by the same criteria.
\par
Span programs are a popular method in quantum computing because there are upper bounds on the complexity of evaluating span programs in the \textit{\textbf{query model}}. The query model evaluates the complexity of a quantum algorithm by its \EMPH{query complexity}, the number of times it queries an input oracle. In our case, the input oracle returns the bits of the binary string $x$. The \textit{\textbf{input oracle}} $\oracle_x$ takes $\oracle_x:\ket{i}\ket{b}\to\ket{i}\ket{b\oplus x_i}$ where $i\in[n]$. Observe that the states $\ket{i}$ can be stored on $\lceil\log n\rceil$ qubits.
Reichardt~\cite{Reichardt2009} showed that the query complexity of a span program is a function of the positive and negative witness sizes of the program, which we now define.

\begin{definition}
Let $\P$ be a span program and let $x\in\{0,1\}^n$. A \textit{\textbf{positive witness}} for $x$ is a vector $\ket{w}\in \H(x)$ such that $A\ket{w}=\ket{\tau}$. The \textit{\textbf{positive witness size}} of $x$ is
$$
  w_{+}(x,\P)=\min\{\|\ket{w}\|^2:\ket{w}\in \H(x),\, A\ket{w}=\ket{\tau}\}.
$$
If no positive witness exists for $x$, then $w_{+}(x,\P)=\infty$.
\par
A \textit{\textbf{negative witness}} for $x$ is a linear map $\bra{w}:\U\to\R$ such that $\bra{w}A\Pi_{\H(x)}=0$ and $\braket{w}{\tau}=1$. The \textit{\textbf{negative witness size}} of $x$ is
$$
w_{-}(x,\P)=\min\{\|\bra{w}A\|^2:\bra{w}:\U\to\R,\, \bra{w}A\Pi_{\H(x)} = 0,\, \braket{w}{\tau}=1\}.
$$
If no negative witness exists for $x$, then $w_{-}(x,\P)=\infty$.
\end{definition}

\begin{theorem}[Reichardt~\cite{Reichardt2009}]\label{thm:reichardt}
  Let $D\subset \{0,1\}^n$ and $f:D\to\{0,1\}$. Let $\P$ be a span program that decides $f$. Let $W_+(f,\P)=\max_{x\in f^{-1}(1)}w_{+}(x,\P)$ and $W(f,\P)_{-}=\max_{x\in f^{-1}(0)}w_{-}(x,\P)$. There is a bounded error quantum algorithm that decides $f$ with query complexity $O\left(\sqrt{W_+(f,\P)W_-(f,\P)}\right)$.
\end{theorem}
A caveat to the query complexity model is that in general the time complexity of an algorithm can be much larger than its query complexity.

\subsection{A Span Program for Null-Homology Testing.}
\label{sec:span_homology}

In this section, we present a span program for testing if a cycle is null-homologous in a simplicial complex. This span program is a generalization of the span program for $st$-connectivity defined in \cite{Karchmer} and used to develop quantum algorithms in \cite{Belovs2012,cade,jarret_et_al:LIPIcs:2018:9512,Jeffery2017}.
\par 
Let $\K$ be a $d$-dimensional simplicial complex. Let $\gamma\in C_{d-1}(\K)$ be a $(d-1)$-cycle. Let $n_d$ be the number of $d$-simplices in $\K$. Order the $d$-simplices $\{\sigma_1,\ldots,\sigma_{n_d}\}$. Let $w \colon \K_d \rightarrow \R$ be a weight function on the $d$-simplices, and let $W:C_d(\K)\to C_d(\K)$ be the diagonal weight matrix. We define a span program over the strings $\{0,1\}^{n_d}$ as follows.
\begin{enumerate}
  \item $\H=C_d(\K)$, with $\H_{i,1}=\spn\{\ket{\sigma_i}\}$ and $\H_{i,0}=\{0\}$.
  \item $\U=C_{d-1}(\K)$
  \item $ A= \boundary_d\sqrt{W} \colon C_d(\K) \rightarrow C_{d-1}(\K)$
  \item $\ket{\tau}=\gamma$
\end{enumerate}

We denote the above span program by $\P_\K$. Let $x\in\{0,1\}^{n_d}$ be a binary string. We define the subcomplex $\K(x) \coloneqq \K^{d-1}\cup\{\sigma_i:x_i=1\}$; that is, $\K(x)$ contains the $d$-simplices $\sigma_i$ such that $x_i = 1$. There exists a solution $v$ to the linear system $\partial_d \sqrt{W} \Pi_{\K(x)} v = \gamma$ if and only if the cycle $\gamma$ is null-homologous in $\K(x)$ if and only if $x$ is a positive instance of $\P_\K$. The span program $\P_\K$ decides the boolean function $f:\{0,1\}^{n_d}\to\{0,1\}$ where $f(x) = 1$ if and only if $\gamma$ is a null-homologous cycle in the subcomplex $\K(x)$.
\par 
\Cref{thm:reichardt} allows us to bound the query complexity of our span program by the size of positive and negative witness. In the next section, we provide bounds on the positive and negative witness size of our span program.

\subsection{Witness Sizes of the Null-Homology Testing Span Program.}
\label{sec:witness_sizes}

In this section, we bound the positive and negative witness sizes of our span program for null-homology testing. We will show that they are equal to the quantities called the effective resistance and effective capacitance of the cycle. We first introduce these quantities and show some of their properties. Then, in~\Cref{sec:res_and_cap_and_witness_sizes}, we show that these quantities are the witness sizes of our span program.

\subsubsection{Background: Effective Resistance and Effective Capacitance.}
\label{sec:res_and_cap}

Let $\gamma\in C_{d-1}(\K)$ be a cycle in a simplicial complex. We associate two quantities with $\gamma$: its \textit{effective resistance} and \textit{effective capacitance}. The effective resistance is finite if and only if $\gamma$ is null-homologous, and the effective capacitance is finite if and only if $\gamma$ is not null-homologous.
We begin with the definition of effective resistance.

\begin{definition}
\label{def:effective_resistance}
  Let $\K$ be a simplicial complex with weight function $w:\K\to\R^+$. Let $\gamma$ be a $(d{-}1)$-cycle in $\K$. The \textbf{effective resistance} of $\gamma$ is
  $$
    \RE_\gamma(\K, W)= \begin{cases}
        \gamma^T \left( L^{up,\,W}_{d-1} \right)^+ \gamma & \text{if $\gamma$ is null-homologous} \\
        \infty & \text{otherwise} \\
    \end{cases}
  $$
  When obvious or when $\K$ is unweighted, we drop the weights from the notation and write $\RE_\gamma(\K)$.
\end{definition}

This definition of effective resistance is consistent with effective resistance in graphs (see \cite{Spielman2019}) and other definitions of effective resistance in simplicial complexes \cite{Kook2018, osting2020sparsification, Hansen2019}.\footnotemark However, this definition gives little intuition about effective resistance. We now prove there is an alternative definition of effective resistance in terms of chains with boundary $\gamma$. We begin with two definitions.
\footnotetext{Effective resistance can be defined even more generally using the combinatorial Laplacian. For simplicity, consider the unweighted case. For a null-homologous $(d-1)$-cycle $\gamma$, the effective resistance can be defined $\RE_{\gamma}(\K) = \gamma^{T}L_{d-1}^{+}\gamma$. This equals the formula in \Cref{def:effective_resistance} because (1) $L_{d-1}^{+} = (L_{d-1}^{up})^{+} + (L_{d-1}^{down})^{+}$ and (2) $\gamma^{T}(L_{d-1}^{down})^{+}\gamma=0$ as $\gamma\in\ker L_{d-1}^{down} = \ker (L_{d-1}^{down})^{+}$.}

\begin{definition}
Given a $d$-dimensional simplicial complex $\K$ and a $(d-1)$-dimensional null-homologous cycle $\gamma$, a \textbf{unit \boldmath$\gamma$-flow} is a $d$-chain $f \in C_d(\K)$ such that $\partial f = \gamma$.
\end{definition}

In the case of graphs, a unit $st$-flow is a flow sending 1 unit of flow from $s$ to $t$.
\begin{definition}
Given a $d$-dimensional simplicial complex $\K$ with weight function $w:C_d(\K) \to \R^{+}$ and a unit $\gamma$-flow $f$, the \textbf{flow energy} of $f$ on $\K$ is \[ \mathsf{J}(f) = \sum_{\sigma \in \K^{(d)}} \frac{f(\sigma)^2}{w(\sigma)} = f^{T} W^{-1} f \] where $W$ is the $n_d\times n_d$ diagonal matrix whose entries are the weights of the $d$-simplices.
\end{definition}

We will now relate unit $\gamma$-flows and their energy to effective resistance. This generalizes a formula for effective resistance in graphs~\cite[Chapter IX Corollary 6]{bollobas1998modern}.

\begin{lemma}
\label{lem:effective_resistance_flows}
  Let $\K$ be a simplicial complex and let $\gamma$ be a null-homologous $d$-cycle. The effective resistance of $\gamma$ is the minimum flow energy over all unit $\gamma$-flows, i.e.
  $$
    \mathcal{R}_\gamma(\K) = \min\{ \mathsf{J}(f) \mid \boundary f = \gamma \}
  $$
\end{lemma}
\begin{proof}
  Our first observation is that we can factor the weighted Laplacian as
  \begin{align*}
    L_d^{up,\,W} &= \boundary_{d+1} W \coboundary_d \\
          &= \boundary_{d+1} W^{1/2} W^{1/2} \coboundary_d \\
          &= (\boundary_{d+1} W^{1/2}) (\boundary_{d+1} W^{1/2})^T
  \end{align*}
  By \Cref{lem:properties_pseudoinverse} Part 1, $(L_d^{up,\,W})^+ = ((\boundary_{d+1} W^{1/2})^T)^+ (\boundary_{d+1} W^{1/2})^+$. Therefore,
  $$
  \mathcal{R}_\gamma(\K)= \gamma^{T}((\boundary_{d+1} W^{1/2})^T)^+(\boundary_{d+1} W^{1/2})^+\gamma = \| (\boundary W^{1/2})^+\gamma \|^2.
  $$
   By \Cref{lem:properties_pseudoinverse} Part 2, $\mathcal{R}_\gamma(\K)$ is the minimum squared-norm of a vector that $\boundary_{d+1}W^{1/2}$ maps to $\gamma$. Let $f=(\boundary W^{1/2})^+\gamma$; the vector $f$ is the unit $\gamma$-flow of minimum flow energy, which we now prove.
   \par
    A vector $v$ is mapped to $\gamma$ by $\boundary W^{1/2}$ iff $W^{1/2}v$ is mapped to $\gamma$ by $\boundary$ as $W^{1/2}$ is a bijection; that is all to say, $W^{1/2}v$ is a unit $\gamma$-flow. Moreover, the flow energy of $W^{1/2}v$ is
   \begin{align*}
    \mathsf{J}(W^{1/2}v) &=  (W^{1/2}v)^{T} W^{-1} W^{1/2}v \\
    &= v^{T} W^{1/2}W^{-1}W^{1/2} v \\
    &= v^{T}v \\
    &= \| v \|^2
   \end{align*}
   Therefore, the minimum flow energy of a unit $\gamma$-flow is the minimum squared-norm of a vector that $\boundary W^{1/2}$ maps to $\gamma$, which we previously saw was $\mathcal{R}_\gamma(\K)$.
\end{proof}

\noindent
We call $(\boundary W^{1/2})^+\gamma$ the \EMPH{minimum-energy} unit $\gamma$-flow\footnote{The minimum-energy unit $\gamma$-flow is unlike the optimal bounded chain \cite{Dunfield2011}, another minimum spanning object of a null-homologous cycle. The optimal bounded chain of $\gamma$ is (often) defined with $\Z_2$ coefficients and is the smallest set of simplices with boundary $\gamma$; it is analogous to an $s,t$-path in a graph. The minimum-energy $\gamma$-flow is defined with real coefficients and can have fractional values on simplices; it is analogous to an $s,t$-flow in a graph. While the weight of the optimal bounded chain is always minimized when fewer simplices are used, the minimum minimum-energy $\gamma$-flow will push a fraction of the flow on many chains spanning $\gamma$. Indeed, in Section \ref{sec:formulas}, we prove the minimum-energy $\gamma$-flow pushes a non-zero amount of flow on each chain spanning $\gamma.$}.
\par 
Some of the key properties of effective resistance in graphs are the series and parallel formulas and Rayleigh Monotonicity. In~\Cref{sec:formulas}, we prove analogous results for higher-dimensional effective resistance.
\par
While effective resistance has previously been generalized from graphs to simplicial complexes~\cite{Hansen2019, Kook2018, osting2020sparsification}, to our knowledge, we are the first to generalize effective capacitance from graphs to simplicial complexes. Unfortunately, effective capacitance is more opaque than effective resistance, both in graphs and simplicial complexes. The definition of effective capacitance is less intuitive than the definition for effective resistance, and there are fewer results about effective capacitance in graphs than effective resistance.
\par
Before defining effective capacitance in simplicial complexes, we review the definition of effective capacitance in graphs, which can be found in~\cite{jarret_et_al:LIPIcs:2018:9512}.
Let $G$ be a graph such that $s$ and $t$ are connected in $G$, and let $H \subseteq G$ be a subgraph such that $s$ and $t$ are not connected in $H$. A \textit{unit $st$-potential} is a function $p \colon V(G) \rightarrow \R$ such that $p(t) = 1$, $p(s)=0$, and $p(u) = p(v)$ for any two vertices $u,v$ in the same connected component. The \textit{potential energy} of $p$ is $\sum_{\{u,v\}\in E(G)} (p(u)-p(v))^{2}$. The \textit{effective capacitance} of $s$ and $t$ is the minimum potential energy of any $st$-potential.
\par 
Our definition of effective capacitance in simplicial complexes will be analogous to the defintion in graphs; namely, the effective capacitance of a cycle $\gamma$ will be the minimum energy of a unit $\gamma$-potential.

\begin{figure}
    \centering
    \begin{subfigure}{0.3\textwidth}
        \centering
        \vspace{0.25in}
        \includegraphics[height=1in]{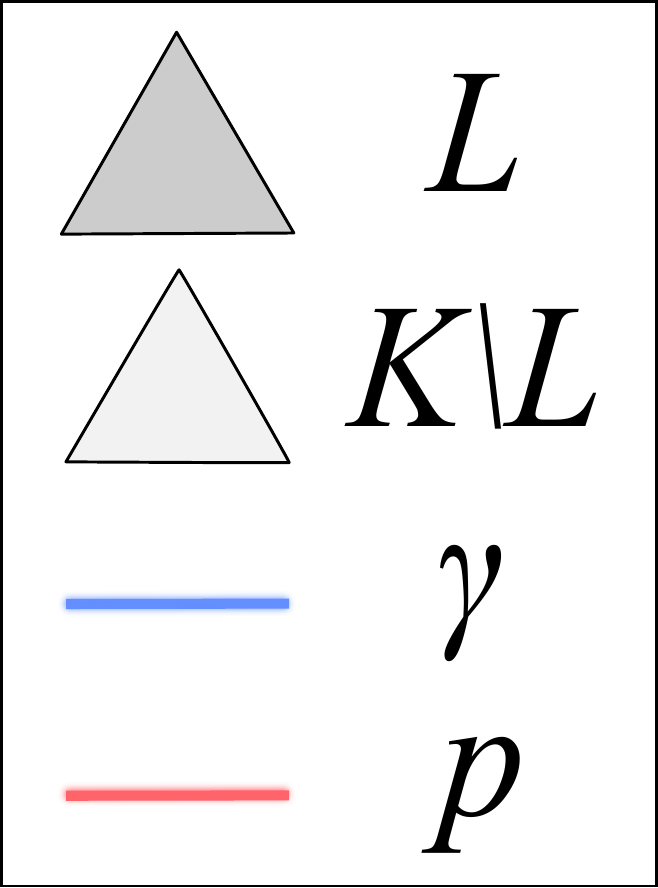}
        \vspace{0.25in}
    \end{subfigure}
    \begin{subfigure}{0.3\textwidth}
        \centering
        \includegraphics[height=1.5in]{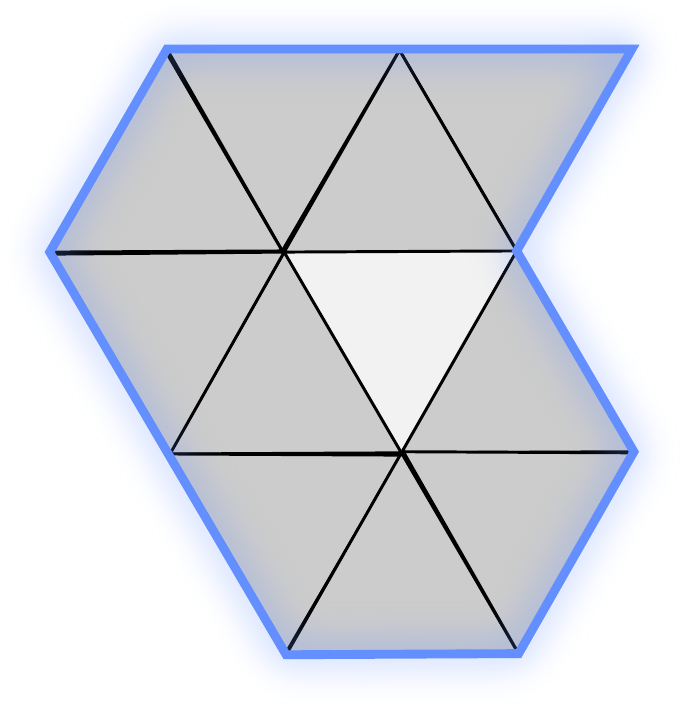}
    \end{subfigure}
    \begin{subfigure}{0.3\textwidth}
        \centering
        \includegraphics[height=1.5in]{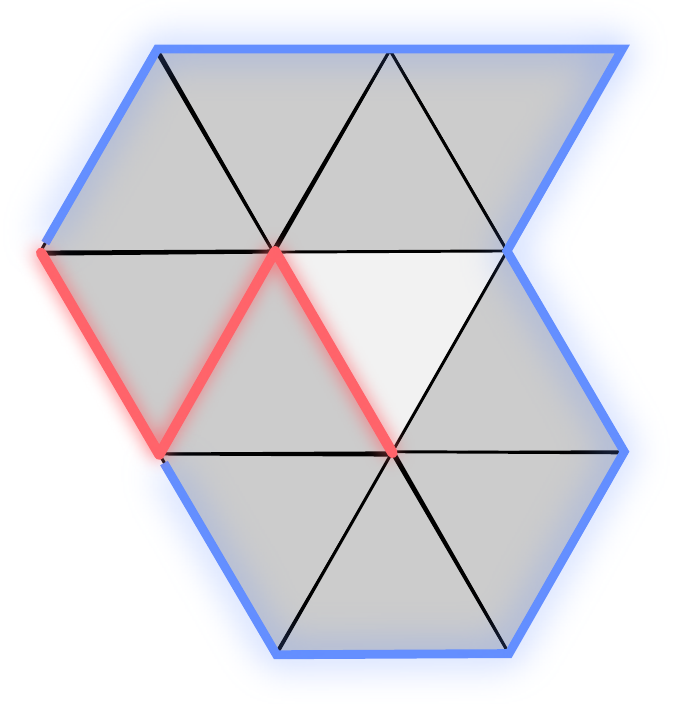}
    \end{subfigure}
    \caption{Left: A 1-cycle $\gamma$ with $\pm 1$ coefficients on the blue edges. Right: A unit $\gamma$-potential $p$ with $\pm 1$ coefficients on the red edges. If this complex is unweighted, then the potential energy of $p$ is $1$. It can be proved that $p$ is a minimal energy unit $\gamma$-potential \protect\footnotemark, so $\C_{\gamma}(\L, \K)=1$ }
    \label{fig:potential}
\end{figure}
\footnotetext{The proof that $p$ is the minimal energy unit $\gamma$-potential is analogous to the proof of \Cref{thm:capacitance_lower_bound}. The trick to proving this is to realize that because there is a single 2-simplex in $\K\setminus\L$, all unit $\gamma$-potentials have the same potential energy.}

\begin{definition}
Let $\L$ be a simplicial complex, and let $\gamma\in C_{d-1}(\L)$ be a $(d{-}1)$-cycle that is not null-homologous in $\L$. A \textbf{unit \boldmath$\gamma$-potential} in $\mathcal{L}$ is a $(d{-}1)$-chain $p$ such that $\coboundary_{d-1}[\L]p=0$ and $p^{T}\gamma=1$.
\end{definition}

\noindent
Figure \ref{fig:potential} shows a $\gamma$-potential in a simplicial complex.

\begin{definition}
Given simplicial complexes $\L \subset \K$ with weight function $w:C_d(\K) \to \R$ and a
$\gamma$-potential $p$ in $\L$, the \textbf{potential energy} of $p$ on $\K$
is \[\mathcal{J}(p) = \sum_{\sigma \in \K_d} \delta[\K]p(\sigma)^{2}
w(\sigma) = (\coboundary[\K]p)^{T} W (\coboundary[\K]p). \]
\end{definition}

\begin{definition}
Let $\L\subset\K$ be simplicial complexes, and let $\gamma\in C_{d-1}(\L)$ be a $(d-1)$-cycle that is null-homologous in $\K$. If $\gamma$ is not null-homologous in $\L$, the \textbf{effective capacitance} of $\gamma$ in $\L$ and $\K$ is 

$$
\C_\gamma(\L, \K) = 
\begin{cases}
\underset{p:\text{$p$ unit $\gamma$-potential}}{\min}\: \mathcal{J}(p) & \text{$\gamma$ is not null-homologous} \\
\infty & \text{$\gamma$ is null-homologous}
\end{cases}
$$ 
\end{definition}

Our definition of effective capacitance in simplicial complexes matches the definition of effective capacitance in graphs; however, this may not be obvious at first glance, as our definition of $st$-potential is more general. A function $p:V(G)\to\R$ is equivalent to a 0-chain $p\in C_0(H)$, and the requirement that $p(u) = p(v)$ for any two vertices $u,v$ in the same connected component is equivalent to saying $\coboundary_{0}[H]=0$; however, not all chains $p$ such that $p^{T}(s-t)=1$ satisfy $p(s)=1$ and $p(t)=0$. (For example, it could be the case that $p(s)=\frac{1}{2}$ and $p(t)=-\frac{1}{2}$.) This difference in the definition ends up not mattering though. This is because the all-1s vectors $1\in\ker\coboundary_{0}$ for any graph. Using this fact, we can see that for any $st$-potential $p$ under our definition, there is an $st$-potential $p'$ under the previous definition with the same potential energy, namely the potential $p'=p-p(t)1$. 
\par 
There is one small detail left to show. It is not obvious from the definition that a unit $\gamma$-potential will even exist for $\gamma$. We prove this in the following lemma.

\begin{lemma}
  Let $\L$ be a simplicial complex, and let $\gamma\in C_{d-1}(\L)$ be a cycle. Then there exists a unit $\gamma$-potential in $\L$ if and only if $\gamma$ is not null-homologous in $\L$.
\end{lemma}
\begin{proof}
  Observe that $\ker\coboundary_{d-1}[\L] = (\im\boundary_d[\L])^\perp$ as $\coboundary_{d-1}[\L]=\boundary_{d}[\L]^T$. Assume there is a $\gamma$-potential $p$ in $\L$. As $\coboundary[\L]p=0$, then $p\in\ker\coboundary_{d-1}[\L] = (\im\boundary_d[\L])^\perp$. As $\gamma^{T}p=1$ we see that $\gamma$ has a non-zero component in $(\im\boundary_d[\L])^\perp$, so $\gamma\not\in\im\boundary_{d}[\L]$.
  \par
  Alternatively, suppose that $\gamma$ is not null-homologous in $\L$. Then $\gamma$ has a non-zero component in $(\im\boundary_d[\L])^\perp=\ker\coboundary_{d-1}[\L]$. Let $q=\Pi_{\ker\coboundary[\L]}\gamma$, where $\Pi_{\ker\coboundary[\L]}$ is the projection operator onto $\ker\coboundary_{d-1}[\L]$. Then $\gamma^{T}q\neq 0$ and $\coboundary_{d-1}[\L]q=0$. The vector $q$ is not necessarily a unit $\gamma$-potential as it is not necessarily the case that $\gamma^{T}q=1$, but the scaled vector $p=\frac{1}{\gamma^{T}q}q$ is a unit $\gamma$-potential.
\end{proof}

One interesting property of effective resistance and capacitance in graphs is that, in planar graphs, the effective resistance between certain pair of nodes in the dual graph equals the effective capacitance between certain pairs of nodes in the primal graph. In \Cref{sec:duality}, we show that an analogous property holds for higher-dimensional embedded simplicial complexes. 

\subsubsection{Effective Resistance and the Spectral Gap.}

In this section, we give a characterization of the spectral gap of the combinatorial Laplacian in terms of the effective resistance of a cycle. While the proof of this lemma follows from some simple linear algebra, the advantage of this theorem comes down to the fact that effective resistance is easier to work with than eigenvectors of the Laplacian (in our opinion). We first relate effective resistance to the spectral gap of the up Laplacian. We then show how this relates effective resistance to the spectral gap of the combinatorial Laplacian. We prove this relationship for unweighted simplicial complexes; however, the theorems also hold for weighted simplicial complexes.

\begin{lemma}
\label{lem:spectral_gap_effective_resistance}
    The spectral gaps of the up Laplacian $L^{up}_{d-1}$ and down Laplacian $L^{down}_d$ are 
    $$
        \lambda_{\min}(L^{down}_d) = \lambda_{\min}(L^{up}_{d-1}) = \min\{ \mathcal{R}^{-1}_{\gamma}(\K) : \gamma\in\im\boundary_d,\,\|\gamma\|=1 \}.
    $$
\end{lemma}
\begin{proof}
    We first prove this is the case for the spectral gap of the up Laplacian $L^{up}_d$. We then show the equivalence of $\lambda_{\min}(L^{down}_d)$ and $\lambda_{\min}(L^{up}_{d-1})$.
    \par 
    The lemma follows from some standard facts about symmetric matrices. First, because $L^{up}_d$ is symmetric, a vector $x$ is an eigenvector of $L^{up}_d$  with non-zero eigenvalue $\lambda$ if and only if $x$ is an eigenvector of $(L_d^{up})^{+}$ with non-zero eigenvalue $\lambda^{-1}$. This follows from the fact that the singular values and vectors of a symmetric matrix are also its eigenvalues and eigenvectors. Therefore, the smallest non-zero eigenvalue of $L^{up}_d$ is the inverse of the largest non-zero eigenvalue of $(L^{up}_d)^{+}$, or $\lambda_{\min}(L_d^{up}) = \lambda^{-1}_{\max}((L_d^{up})^{+})$ for short.  
    \par 
    Next, we can characterize the eigenvalues of the symmetric matrix $(L_d^{up})^{+}$ with the \textit{\textbf{Courant-Fischer Theorem}}. We use a special case of the theorem, which says that $\lambda_{\max}((L_d^{up})^{+}) = \max\{x^{T}(L_d^{up})^{+}x : \|x\|=1\}$ and $x_{\max} = \arg\max\{x^{T}(L_d^{up})^{+}x : \|x\|=1\}$, where $x_{\max}$ is an eigenvector corresponding to $\lambda_{\max}((L_d^{up})^{+})$. The lemma follows from the fact that $x_{\max}\in\im L_d^{up} = \im\boundary_d$, which is the case because $x_{\max}$ is the eigenvector of a non-zero eigenvalue of $L_d^{up}$. 
    \par 
    Finally, $\lambda_{\min}(L^{up}_d) = \lambda_{\min}(L^{down}_{d-1})$ by \Cref{lem:properties_spectrum_laplacian} Part 1.
\end{proof}

\subsubsection{Effective Capacitance and the Spectral Gap.}

In the previous section, we saw that the effective resistance of a unit-length cycle is always bounded above by the inverse of the spectral gap of the combinatorial Laplacian. While we don't know such a bound for the effective capacitance of arbitrary cycles, we can prove such a bound for the effective capacitance for the boundaries of simplices. This is sufficient for our analysis of the incremental algorithm as the only cycles we consider are the boundaries of simplices.
\par 
 Before proving our upper bound on the effective capacitance of a cycle, we need to prove an upper bound on the largest singular value of the coboundary matrix.

\begin{lemma}
\label{lem:upper_bound_largest_singular_value_coboundary}
Let $\K$ be a simplicial complex with $n_0$ vertices. For any $d\geq 1$, the largest singular value of the coboundary matrix $\delta_{d-1}[\K]$ is $\sigma_\text{max}(\delta_{d-1}) = O(\sqrt{n_0})$.
\end{lemma}
\begin{proof}
This follows as the squared singular values of $\coboundary_{d-1}$ are the eigenvalues of the up Laplacian $\coboundary_{d-1}^T\coboundary_{d-1}=L_{d-1}^{up}$. (This is true for any matrix of the form $A^{T}A$.) The maximum eigenvalue of $L_{d-1}^{up}$ is known to be at most $n_0$ by \Cref{thm:upper_bound_maximum_eigenvalue_laplacian}. 
\end{proof}

\begin{theorem}\label{thm:cap_spectral_gap}
    Let $\L\subset\K$ be $d$-dimensional simplicial complexes. Let $\gamma\in C_{d-1}(\L)$ be a $(d-1)$-cycle that is null-homologous in $\K$ but not in $\L$. Assume that $\gamma=\boundary\sigma$ for a $d$-simplex $\sigma\notin \L$.\footnotemark The effective capacitance of $\gamma$ in $\K$ is bounded above by $\C_\gamma(\L,\K) = O \left( n_{0}\cdot\lambda_{\min}^{-1}(L_{d-1}[\L\cup\{\sigma\}])\right)$. 
\end{theorem}
\footnotetext{The theorem holds whether or not $\sigma\in\K$.}
\begin{proof}
    We can express the constraints of a $\gamma$-potential $p$ in the following set of linear equations: 
    $$
      \begin{bmatrix} \phantom{0} \\ \coboundary[\L] \\ \phantom{\vdots} \\ \gamma^T \end{bmatrix} p = \begin{bmatrix} 0 \\ 0 \vphantom{\coboundary[\L]} \\ \vdots  \\ 1 \end{bmatrix}
    $$
    To simplify notation, let $C=\begin{bmatrix}  \coboundary[\L]^T & \gamma \end{bmatrix}^T$ and $b=\begin{bmatrix} 0 & 0 & \cdots & 1 \end{bmatrix}^T$
    \par 
    We consider the smallest vector $p$ which satisfies these equations, which is $p=C^{+}b$. Because $\gamma=\boundary\sigma$, we can see that $C=\coboundary[\L\cup\{\sigma\}]$. Therefore, $\| p \| = O(\|C^{+}b\|) = O(\sigma^{-1}_{\min}(\coboundary[\L\cup\{\sigma\}])$, where $\sigma_{\min}(\coboundary[\L\cup\{\sigma\}])$ is the smallest non-zero singular value of $\coboundary[\L\cup\{\sigma\}]$. However, we know that $\sigma_{\min} = \sqrt{\lambda_{\min}(L_{d-1}^{up}[\L\cup\{\sigma\}])} \in \Omega(\sqrt{\lambda_{\min}(L_{d-1}[\L\cup\{\sigma\}])})$. Therefore, $\|p\|\in O\left(\sqrt{\lambda_{\min}^{-1}(L_{d-1}[\L\cup\{\sigma\}])}\right)$. 
    \par 
    We now want to bound the potential energy of $p$. Using \Cref{lem:upper_bound_largest_singular_value_coboundary}, we can bound $\|\coboundary[\K]p\|^{2} \in O(n_0\cdot \lambda_{\min}^{-1}(L_{d-1}[\L\cup\{\sigma\}]))$.
\end{proof}

\subsubsection{Connecting Effective Resistance and Capacitance to Witness Sizes.}
\label{sec:res_and_cap_and_witness_sizes}

Given a string $x \in \{0,1\}^{n_d}$, we show in the following two lemmas that $w_+(x, \P_\K) = \mathcal{R}_\gamma(\K(x))$ and $w_-(x, \P_\K) = \C_\gamma(\K(x),\K)$.
The proofs are simple calculations following from the definitions of effective resistance and capacitance.

\begin{lemma}\label{lem:pos_inst}
Let $x\in\{0,1\}^{n_d}$ be a positive instance. There is a bijection between positive witnesses $\ket{w}$ for $x$ and unit $\gamma$-flows $f$ in $\K(x)$. Moreover, the positive witness size  is equal to the effective resistance of $\gamma$ in $\K(x)$; that is, $w_+(x, \P_\K) = \mathcal{R}_\gamma(\K(x))$.
\end{lemma}
\begin{proof}
Let $\ket{w_+} \in C_d(\K)$ be a positive witness for $x$, so $\partial_d \sqrt{W} \ket{w_+} = \gamma$.
We construct a unit $\gamma$-flow $f$ in $\K(x)$ by $f = \sqrt{W}\ket{w_+}$; $f$ is indeed a unit $\gamma$-flow as $\partial_d f = \partial_d \sqrt{W}\ket{w_+} = \gamma$. Moreover, $\ket{w_+}=W^{-1/2}\ket{f}$.
The flow energy of $\gamma$ is
\begin{align*}
\Jsf(f) &= \bra{f} W^{-1} \ket{f}\\
&= \braket{W^{-1/2}f}{W^{-1/2}f}\\
&= \braket{w_+}{w_+}\\
&= \|\ket{w_+}\|^2.
\end{align*}
Hence, the flow energy of $f$ equals the witness size of $x$.

Conversely, let $f$ be a unit $\gamma$-flow in $\K(x)$ and define the positive witness for $x$ as $\ket{w_+} = W^{-1/2}\ket{f}$. The same computation in the above paragraph shows that the flow energy of $f$ equals the positive witness size of $x$.
\end{proof}

\begin{lemma}\label{lem:neg_inst}
Let $x\in\{0,1\}^{n_d}$ be a negative instance. There is a bijection between negative witnesses $\bra{w_-}$ for $x$ and unit $\gamma$-potentials $p$ in $\K(x)$. Moreover, the negative witness size is equal to the effective capacitance of $\gamma$ in $\K(x)$; that is, $w_-(x, \P_\K) = C_\gamma(\K(x))$.
\end{lemma}
\begin{proof}
Let $\bra{w_-}$ be a negative witness for $x$. As $\bra{w}$ is a linear function from $C_{d-1}(\K)$ to $\R$ we may view it as a $(d-1)$-chain $p^{T} = \bra{w}$.
Since $\braket{w_{-}}{\gamma} = 1$, then $p^{T}\gamma = 1$.
To show that $p$ is a unit $\gamma$-potential we must show that the coboundary of $p$ is zero in $\K(x)$.
By the definition of a negative witness we have
\begin{align*}
0 &= \bra{w_-} \partial_d \sqrt{W} \Pi_{\K(x)} \\
&= \bra{p} \partial_d \sqrt{W} \Pi_{\K(x)} \\
&= \bra{\delta_d(p)} \sqrt{W} \Pi_{\K(x)}.
\end{align*}
Since $\sqrt{W}$ is a diagonal matrix and $\Pi_{\K(x)}$ restricts the coboundary to the subcomplex $\K(x)$ we see that $\braket{\delta_d(p)}{\sigma} = 0$ for any $\sigma \in \K(x)_d$.
To show that the witness size of $\bra{w_-}$ is equal to the potential energy of $p$ we have
\begin{align*}
\| \bra{w_-} \partial_d \sqrt{C}\| ^2 &= \braket{ p \partial_d \sqrt{W}}{p \partial_d \sqrt{W}}\\
&= \braket{\sqrt{W}\delta_d(p) }{\sqrt{W}\delta_d(p) } \\
&= \sum_{\sigma \in \K_d} \braket{\delta_d(p)}{\sigma}^2 w(\sigma) \\
&= \J(p).
\end{align*}

Conversely, let $p$ be a unit $\gamma$-potential for $\K(x)$ we construct a negative witness for $x$ by setting $\bra{w_-} \coloneqq p^{T}$. Since the coboundary of $p$ is zero in $\K(x)$ we have $\braket{\delta_p(p)}{\sigma} = 0$ for each $\sigma \in \K(x)_d$ which implies $\bra{w_-}\partial_d \sqrt{W} \Pi_{\K(x)} = 0$ by the reasoning in the previous paragraph.
Also by the previous paragraph we have that the potential energy of $p$ is equal to the negative witness size of $\bra{w_-}$ which concludes the proof.
\end{proof}

From these two lemmas we obtain the main theorem of the section, the quantum query complexity of $\gamma$.

\begin{theorem}\label{thm:querycomplexity}
    Given a $d$-dimensional simplicial complex $\K$, a $(d-1)$-dimensional cycle $\gamma$ that is null-homologous in $\K$, and a $d$-dimensional subcomplex $\K(x) \subseteq \K$, there exists a quantum algorithm deciding whether or not $\gamma$ is null-homologous in $\K(x)$ with quantum query complexity $O \left( \sqrt{\RE_{\max}(\gamma) \C_{\max}(\gamma)} \right)$, where $\RE_{\max}(\gamma)$ is the maximum finite effective resistance $\RE_{\gamma}(\L)$ in any subcomplex $\L\subset\K$, and $\C_{\max}(\gamma)$ is the maximum finite effective capacitance $\C_{\gamma}(\L,\K)$ in any subcomplex $\L\subset\K$.
\end{theorem}
\begin{proof}
By Theorem~\ref{thm:reichardt}, the span program $\P_\K$ can be converted into a quantum algorithm whose query complexity is $O \left( \sqrt{W_+(f, \P_\K) W_-(f, \P_\K)} \right)$ where $W_+(f,\P_\K)=\max_{x \in f^{-1}(1)} \mathcal{R}_\gamma(\K(x))=\RE_{\max}(\gamma)$ and $W_-(f, \P_\K)=\max_{x \in f^{-1}(0)} \C_\gamma(\K(x),\K)=\C_{\max}(\gamma)$.
\end{proof}

\subsection{Time Efficient Implementations of the Span Program.}
\label{sec:time_complexity}

We have given bounds on the query complexity of null-homology testing; however, this does not imply a bound on the time complexity of evaluating this span program, as the query complexity does not account for the work outside of the oracle calls. In \Cref{sec:evaluate}, we describe the details of an implementation of this algorithm. For certain special cases, we are able to analyse the time complexity of the algorithm. We describe this special case below.
\par 
There are two obstacles to a time-efficient implementation of the span program: the weights and the input cycle $\gamma$. The weights on the $d$-simplices make it difficult to implement the matrix $\boundary\sqrt{W}$, as the weights on the simplices can be arbitrary real numbers. The input cycle $\gamma$ is difficult to create on a quantum computer as the entries of $\gamma$ can also be arbitrary real numbers.
\par
Accordingly, we can give a quantum algorithm of bounded time complexity in one particular instance: when $\K$ is unweighted and $\gamma$ is the boundary of a $d$-simplex.   (We do not require the $d$-simplex to actually appear in the complex.) While this is only a special case of the generic null-homology testing algorithm, this is the only case we need for the incremental algorithm for computing Betti numbers (\Cref{alg:incremental_algorithm}). The time complexity of this case is given in the following theorem.

\begin{restatable}{theorem}{thmtimecomplexity}\label{thm:null_homology_time_complexity}
  Let $\K$ be an unweighted simplicial complex with $n_0$ vertices, let $\gamma\in C_{d-1}(\K)$ a null-homologous cycle in $\K$, and $\K(x) \subset \K$ be a simplicial complex. Furthermore, assume that $\gamma$ is the boundary of a $d$-simplex and the complex is unweighted. There is a quantum algorithm for deciding if $\gamma$ is null-homologous in $\K(x)$ that runs in time
  $$
    \Tilde{O}\left( \sqrt{\frac{\RE_{\max}(\gamma)\C_{\max}(\gamma)}{\tilde{\lambda}_{\min}}} n_0 + \sqrt{\RE_{\gamma}(\K)n_0}\right),
  $$
  where $\RE_{\max}(\gamma)$ is the maximum finite effective resistance $\RE_{\gamma}(\L)$ of $\gamma$ in any subcomplex $\L\subset\K$, $\C_{\max}(\gamma)$ is the maximum finite effective capacitance $\C_{\gamma}(\L, \L)$ in any subcomplex $\K(x)$, and $\tilde{\lambda}_{\min}$ is the spectral gap of the normalized up-Laplacian $\tilde{L}_{d-1}^{up}[\K]$.
\end{restatable}

\subsection{Runtime of the Quantum Incremental Algorithm}

In the previous section, we saw an implementation of an algorithm for testing if the boundary of a $d$-simplex was null-homologous. Combined with the framework of the Incremental Algorithm (\Cref{alg:incremental_algorithm}), this allows us to compute the $d$-Betti number.

\thmincrementalalgorithmruntime*

\begin{proof}
    The Incremental Algorithm (\Cref{alg:incremental_algorithm}) incrementally adds each $d$ and $(d+1)$-simplex $\sigma$ to the simplicial complex and checks if the cycle $\boundary\sigma$ is null-homologous. We can use the span-program algorithm of \Cref{thm:null_homology_time_complexity} to check if $\boundary\sigma$ is null-homologous. The theorem follows by using this algorithm for each of the $(n_d+n_{d+1})$ $d$- and $(d+1)$-simplices.
\end{proof}

\subsection{Comparison with Existing Algorithms.}
\label{sec:comparision_existing_algorithms}

In this section, we compare our algorithm to existing algorithms for QTDA. This presentation specifically compares our algorithm to the LGZ algorithm~\cite{Lloyd2016lgz}, but most of these ideas also hold for other existing QTDA algorithms.

\paragraph{Input.} 

Our algorithm makes different assumptions about how the simplicial complex is stored compared to previous algorithms. We assume we have a list of simplices in the simplicial complex; this is required for the incremental algorithm as we must iteratively add the simplices and test if their boundaries are null-homologous. Compare this to existing quantum algorithms, which assume we have a way of checking if a simplex is included in the simplicial complex. 
\par 
Our algorithm assumes we have a \textit{\textbf{list oracle}} that can return the simplices in the simplicial complex:
$$
    \mathcal{O}_{list}:\ket{i}\ket{0}\to\ket{i}\ket{\sigma_i},
$$
where $\sigma_i$ is the $i$th $d$-simplex of our simplicial complex. 
\par 
Compare this to the \textbf{\textit{membership oracle}} used in other QTDA algorithms that can check whether a simplex is in the simplicial complex: 
$$
    \mathcal{O}_{memb}:\ket{\sigma_i}\ket{j}\to\ket{\sigma_i}\ket{j\oplus b_i},
$$ 
where $b_i$ is a bit indicating if $\sigma_i\in\K_d$.
\par 
These oracles come with different trade-offs. The oracle $\mathcal{O}_{memb}$ does not require computing the set of simplices in advanced, while $\mathcal{O}_{list}$ does. However, algorithms that use the membership oracle $\mathcal{O}_{memb}$ pay for this in the time it takes to compute a uniform superposition of the $d$-simplices, a costly operation leading to a factor of $\zeta_d = n_d/\binom{n_{0}}{d+1}$ in the runtime. Thus, our algorithm is better suited for \textit{sparse} simplicial complexes---complexes where $n_d << \binom{n}{d+1}$ and where the list of simplices can be computed efficiently---a family of complexes where existing QTDA algorithms perform poorly; see the section ``Runtime'' below for more discussion. 

\paragraph{Output.}

The LGZ algorithm estimates the $d$th Betti number up to an additive factor by returning a value $\chi_{d}$ such that $\left| \chi_{d} - \frac{\beta_d}{\dim C_d(\K)} \right| \leq \epsilon$; the problem of computing $\chi_{d}$ has been deemed \textit{\textbf{Betti number estimation}}. Our algorithm instead returns the Betti number $\beta_k$.

\paragraph{Runtime.}

To compare our algorithm to existing quantum algorithms, we bound the runtime of our algorithm with respect to the spectral gap of the combinatorial Laplacian. Note that while we can bound the runtime of our algorithm by the inverse of the spectral gap, this bound is not necessarily tight.
\begin{corollary}
    Let $\K$ be a simplicial complex with $n_i$ $i$-simplices. There is a quantum algorithm for computing the $d$th Betti number $\beta_d$ in time 
    $$
        \Tilde{O}\left(\Lambda_{\min}^{-3/2} n_0^{5/2}\cdot(n_{d} + n_{d+1})\right)
    $$        
    where $\Lambda_{\min}$ is the minimum spectral gap of $L_{d}[\L]$ over all subcomplexes $\L\subset\K$.
\end{corollary}
\begin{proof}
    This follows from \Cref{thm:incremental_algorithm_runtime} by applying the bounds of \Cref{lem:spectral_gap_effective_resistance}, \Cref{thm:cap_spectral_gap}, and \Cref{lem:normalized_vs_unnormalized_spectral_gap} to bound $\RE_{\max}$, $\C_{\max}$, and $\frac{1}{\tilde{\lambda}_{\min}}$ respectively. The bounds on the effective resistance of \Cref{lem:spectral_gap_effective_resistance} only apply to unit vectors, so one factor of $n_0$ is because of the fact that $\|\boundary\sigma\|=\sqrt{d}$, so $\RE_{\max}(\boundary\sigma)\leq d\Lambda_{\min}^{-1}$.
\end{proof}
  
Compare this to the LGZ algorithm, which runs in time 
\[
    O \left( \epsilon^{-2}\lambda_{\min}^{-1}n_{0}^4\sqrt{\zeta_d^{-1}} \right) 
\] 
where $n_{0}$ is the number of vertices, $\epsilon$ is the error term, $\lambda_{\min}$ is the spectral gap of the combinatorial Laplacian, and $\zeta_d$ is a density term given by 
\[
    \zeta_d = \frac{n_d}{\binom{n_{0}}{d+1}}.
\] 
The density term is the ratio of the number $k$-simplices in the $\K$ to number of $k$-simplices in the complete complex on $n_{0}$ vertices, which may be exponentially small. For example, when $\K$ is sparse, meaning that the number of $d$-simplices is polynomial in the number of vertices, the density may be exponentially small with respect to $d$. Specifically, if $\zeta_d = \Omega \left( n_{0}^{O(1)}/n_0^{d+1} \right)= \Omega \left( 1/n_0^{O(d+1)} \right)$, then the runtime of LGZ is 
\[
    O \left( \epsilon^{-2}\lambda_{\min}^{-1} n_0^{O(d+1)}\right).
\] 
Compare this to our algorithm, which in this case has runtime
$$
    \Tilde{O}\left(\Lambda_{\min}^{-3/2} n_0^{O(1)}\right).
$$        
In this case, our algorithm has a better asymptotic dependence on the size of the complex as it avoids the factor of $\binom{n_{0}}{d+1}$. This factor of $\binom{n_{0}}{d+1}$ shows up in many of the alternatives to the LGZ algorithm, so our algorithm has a more favorable dependence on $n_0$ compared to these algorithms as well. Additionally, we note that the term $\Lambda_{\min}$ in our algorithm and $\lambda_{\min}$ in the LGZ algorithm are similar but not directly comparable. See the following section.

\paragraph{Effective Resistance and Capacitance vs.~Spectral Gap.} Our algorithm is parameterized by the maximum effective resistance and capacitance of all subcomplexes of a simplicial complex and the square root of the inverse of spectral gap of the simplicial complex, whereas previous QTDA algorithm are only parameterized by the inverse of the spectral gap of the simplicial complex. Although for a \textit{fixed} complex, effective resistance and capacitance are upper bounded by the inverse of the spectral gap, the fact that our algorithm is parameterized by the maximum effective resistance and capacitance over \textit{all} subcomplexes means that the runtime of our algorithm is not entirely comparable to the runtime of existing QTDA algorithms. It is possible are complexes where the effective resistance and capacitance in subcomplexes are significantly lower than the spectral gap of the entire complex, and complexes where the effective resistance or capacitance of a cycle in a subcomplex is larger than the spectral gap of the entire complex. The complete complex is an example of the second case, as it has the maximal possible spectral gap of $n_0$.

\paragraph{Randomized Order for the Incremental Algorithm.} Building on the previous point, while our algorithm is parameterized by the maximum effective resistance and capacitance in various subcomplexes, our algorithm can also incrementally add the simplices \textit{in any order}. This is potentially beneficial as a simplex may have smaller effective resistance or capacitance in one order than another. Of course, we likely will not know in advance whether or not a particular order of the simplices results in less or greater resistance and capacitance. However, we still may able to use this fact to our advantage, as we could run our algorithm multiple times with different orders to gain confidence that our Betti number computations are accurate, which is not the case with previous QTDA algorithms.

\section{Bounds on Effective Resistance and Capacitance.}
\label{sec:bounds}

In this section, we provide upper bounds on the resistance and capacitance of a cycle $\gamma$ in an simplicial complex $\K$. Throughout this section, all simplicial complexes are \textbf{unweighted}.

Our upper bounds are polynomial in the number of $d$-simplices and the cardinality of the torsion subgroup of the relative homology groups. In particular, our bounds on resistance and capacitance are dependent on the maximum cardinality of the torsion subgroup of the relative homology group $H_{d-1}(\mathcal{L}, \mathcal{L}_0, \Z)$, where $\mathcal{L}\subset\K$ is a $d$-dimensional subcomplex and $\mathcal{L}_0 \subset \L$ is a $(d{-}1)$-dimensional subcomplex. In the worst case, our upper bounds are exponential.
\par
In Theorem~\ref{thm:resistance_lower_bound} we provide an example of a simplicial complex containing a cycle $\gamma$ whose effective resistance is exponential in the number of simplices in the complex.
It is important to reiterate that our bounds are in terms of the torsion of the relative homology groups, not the torsion of the (non-relative) homology groups. The simplicial complex we provide has no torsion in its homology groups, only torsion in its relative homology groups.

\subsection{Upper Bounds}

Our upper bounds rely on a change of basis on the boundary matrix called the Smith normal form which reveals information about the torsion subgroup of $H_{d-1}(\K, \Z)$. We state the normal form theorem below.

\begin{theorem}[Munkres, Chapter 1 Section 11~\cite{Munkres1984}]\label{thm:normalform}
There are bases for $C_d(\K)$ and $C_{d-1}(\K)$ such that the matrix for the boundary operator $\partial_d \colon C_d(\K, \Z) \rightarrow C_{d-1}(\K, \Z)$ is in \EMPH{Smith normal form}, i.e. \[\tilde{\partial}_d = \begin{bmatrix} D  & 0  \\ 0  & 0 \end{bmatrix} \]
where $D$ is a diagonal matrix with positive integer entries $d_1,\dots,d_m$ such that each $d_i$ divides $d_{i+1}$ and each $0$ is a zero matrix of appropriate dimensionality. The normal form of $\partial_d$ satisfies the following properties:
\begin{enumerate}
\item The entries $d_1,\dots,d_m$ correspond to the torsion coefficients of $H_{d-1}(\K, \Z) \cong \Z^{\beta_d} \oplus \Z_{d_1} \oplus \dots \oplus \Z_{d_m}$ $($where $\Z_1=0)$,

\item The number of zero columns is equal to the dimension of $\ker(\partial_d)$.
\end{enumerate}
Moreover, the boundary matrix $\boundary$ in the standard basis can be transformed to $\tilde{\boundary}$ by elementary row and column operations. If $\boundary$ is square, these operations multiply $\det\boundary$ by $\pm 1$.
\end{theorem}

Using Theorem~\ref{thm:normalform}, we obtain an upper bound on the determinants of the square submatrices of the boundary matrix $\partial_d[\K]$ in terms of the \textit{relative homology groups} of $\K$.
Let $\mathcal{L}$ be $d$-dimensional subcomplex of $\K$, and let $\L_0$ be a $(d{-}1)$-dimensional subcomplex of $\K$. The \EMPH{relative boundary matrix} $\partial_d[\mathcal{L}, \mathcal{L}_0]$ is the submatrix of $\partial_d$ obtained by including the columns of the $d$-simplices in $\mathcal{L}$ and excluding the rows of the $(d{-}1)$-simplices in $\mathcal{L}_0$.
With the relative boundary matrices, one can define the \EMPH{relative homology groups} as $H_d(\mathcal{L}, \mathcal{L}_0, \Z) = \ker\boundary_d[\L,\L_0] / \im\boundary_{d+1}[\L,\L_0]$.
More information on the relative boundary matrix can be found in~\cite{Dey2011}.
We denote the cardinality of the torsion subgroup of the relative homology group $H_{d-1}(\mathcal{L}, \mathcal{L}_0, \Z)$ by $\mathcal{T}(\mathcal{L}, \mathcal{L}_0)$.
Similarly, we denote the maximum $\mathcal{T}(\mathcal{L}, \mathcal{L}_0)$ over all relative homology groups as $\mathcal{T}_{\max}(\K)$.

\begin{lemma}\label{lem:detbound}
Let $\partial_d[\mathcal{L}, \mathcal{L}_0]$ be a $k \times k$ square submatrix of $\partial_d$ constructed by including columns for the $d$-simplices in $\mathcal{L}$ and excluding rows for the $(d-1)$-simplices in $\mathcal{L}_0$.
The magnitude of the determinant of $\partial_d[\mathcal{L}, \mathcal{L}_0]$ is bounded above by the cardinality of the torsion subgroup of $H_{d-1}(\mathcal{L}, \mathcal{L}_0, \Z)$, i.e 
\[
    | \det \left( \partial_d[\mathcal{L}, \mathcal{L}_0] \right)| \leq \mathcal{T}(\mathcal{L}, \mathcal{L}_0).
\]
\end{lemma}
\begin{proof}
Without loss of generality, we assume that $\det(\partial_d[\mathcal{L}, \mathcal{L}_0]) \neq 0$; if $\det(\partial_d[\mathcal{L}, \mathcal{L}_0]) = 0$, the bound is trivial. Since $\partial_d[\mathcal{L}, \mathcal{L}_0]$ is a non-singular square matrix, its normal form $\tilde{\partial}_d[\mathcal{L}, \mathcal{L}_0]$ is a diagonal matrix $D = \diag(d_1,\dots,d_k)$.
The determinant is equal to $\pm\prod_{i=1}^{k} d_i$ and by Theorem~\ref{thm:normalform} the torsion subgroup of $H_{d-1}(\mathcal{L}, \mathcal{L}_0)$ is $\Z_{d_1} \oplus \dots \oplus \Z_{d_k}$ which has cardinality $\mathcal{T}(\mathcal{L}, \mathcal{L}_0) = \prod_{i=1}^k d_i$.
\end{proof}

\subsubsection{Upper Bounds on Effective Resistance} We are now ready to upper bound the effective resistance of a cycle in a simplicial complex.

\begin{theorem}\label{thm:resistance_upper_bound_torsion}
Let $\K$ be a $d$-dimensional simplicial complex and $\gamma$ a unit-length null-homologous $(d{-}1)$-cycle in $\K$. Let $n=\min\{n_{d-1},n_d\}$. The effective resistance of $\gamma$ is bounded above by $\RE_\gamma(\K) \in O \left( n^2 \cdot \mathcal{T}_{\max}(\K)^2 \right)$.
\end{theorem}
\begin{proof}
First, we remove $d$-simplices from $\K$ to create a new complex $\L$ such that $\ker(\partial_d[\L]) = 0$ and $\im\boundary_d[\K]=\im\boundary_d[\L]$.
Theorem \ref{thm:resistance_monoticity} proves that removing $d$-simplices only increases the effective resistance, so $\mathcal{R}_\gamma(\K) \leq \mathcal{R}_\gamma(\L)$. As $\ker(\boundary_d[\L])=0$, there is a unique unit $\gamma$-flow $f\in C_{d-1}(\L)$ which implies $\RE_\gamma(\L)=\|f\|^2$. Let $n\leq n_d$ denote the number of $d$-simplices in $\L$.
\par
The matrix $\partial_d[\L]$ has full column rank, so we can find a non-singular $n \times n$ square submatrix of $\partial_d[\L]$; call this submatrix $B$. 
Let $\mathcal{L}_0$ be the $(d-1)$-dimensional subcomplex that contains the $(d-1)$-simplices corresponding to rows excluded from $B$; $B$ is the relative boundary matrix $\boundary_d[\L,\L_0]$.
We have that $B f = c$, where $c$ is the restriction of $\gamma$ to the rows of $B$.
Observe that $\|c\|\leq \|\gamma\|=1$
 \par
We will apply Cramer's rule to upper bound the size of $f$. By Cramer's rule we have the equality
\[f(\sigma) = \frac{\det(B_{\sigma, c})}{\det(B)}\]
where $B_{\sigma, c}$ is the matrix obtained by replacing the column of $B$ indexed by $\sigma$ with the vector $c$. Since $\det(B)$ is integral, $|\det(B)|\geq 1$, so we drop the denominator and focus on the inequality $|f(\sigma)| \leq |\det(B_{\sigma, c})|$.
We bound $|\det(B_{\sigma, c})|$ by its cofactor expansion,
\begin{align*}
 |\det(B_{\sigma, c})| &= \left|\sum_{i=1}^{n_d} (-1)^i \cdot c_i \cdot \det(B_{\sigma, c}^{c, i})\right| \\
&\leq \sum_{i=1}^{n_d} |c_i| \cdot \mathcal{T}_{\max}(\K)\\
&= \|c\|_1 \cdot \mathcal{T}_{\max}(\K)\\
&= O \left( \sqrt{n} \cdot \mathcal{T}_{\max}(\K) \right)
\end{align*}
where $B_{\sigma, c}^{c, i}$ denotes the submatrix obtained by removing the column $c$ and removing the $i$th row and $c_i$ denotes the $i$th component of $c$.
The first inequality comes from Lemma~\ref{lem:detbound}, as $B_{\sigma, c}^{c, i}$ is the relative boundary matrix $\boundary_{d}[\L\setminus\{\sigma\}, \L_0\cup \sigma_i]$, where $\sigma_i$ is the $(d{-}1)$-simplex corresponding to the $i$th row of $B$. The factor of $\sqrt{n}$ comes from the fact that $\|c\|_1 \leq \sqrt{n} \|c\|_2$ and $\|c\|_2 \leq 1$.
Finally, we compute the flow energy of $f$ as
\begin{align*}
\Jsf(f) &= \sum_{\sigma \in \L_d} f(\sigma)^2\\
&\leq \sum_{i=1}^{n} n \cdot \mathcal{T}_{\max}(\K)^2\\
&= O\left( n^2 \cdot \mathcal{T}_{\max}(\K)^2 \right).
\end{align*}
The effective resistance of $\gamma$ is the flow energy of $f$, so the result follows.
\end{proof}

If $\L\subset\K$, then the boundary matrix $\boundary_d[\L]$ is a submatrix of $\boundary_d[\K]$. In particular, $\mathcal{T}_{\max}(\L)\leq\mathcal{T}_{\max}(\K)$. Therefore, the proof of \Cref{thm:resistance_upper_bound_torsion} gives an upper bound on the effective resistance for any subcomplex $\L\subset\K$.

\begin{corollary}\label{cor:subcomplex_res_bound}
Let $\L\subset\K$ be a $d$-dimensional simplicial complex and $\gamma$ a null-homologous $(d{-}1)$-cycle in $\L$. Let $n=\min\{n_{d-1}[\L],n_d[\L]\}$. The effective resistance of $\gamma$ in $\L$ is bounded above by $\mathcal{R}_\gamma(\L) = O \left( n^2 \cdot \mathcal{T}_{\max}(\K)^2 \right)$.
\end{corollary}

In \Cref{sec:upper_bound_relative_torsion}, we give an upper bound on relative torsion, which implies an upper bound on the effective resistance purely in terms of the size of the complex.

\subsubsection{Upper Bounds on Capacitance.}

We now provide an upper bound for the effective capacitance of a cycle. While upper bounds on the effective resistance only depended on the norm of the cycle, upper bounds on the capacitance of the cycle are not. Therefore, we consider the special case where $\gamma$ is the boundary of a $d$-simplex. This is a natural assumption as these are exactly the type of cycles considered in the incremental algorithm for computing Betti numbers (\Cref{alg:incremental_algorithm}). While we only prove this special case, we note that our proof could be adapted to bound the effective capacitance of a cycle whose entries have constant upper and lower bounds. 

\begin{theorem}\label{thm:cap_bound}
    Let $\L\subset\K$ be $d$-dimensional simplicial complexes. Let $\gamma\in C_{d-1}(\L)$ be a $(d-1)$-cycle that is null-homologous in $\K$ but not in $\L$. Let $n=\min\{n_{d-1},n_d\}$. Assume that $\gamma=\boundary\sigma$ for a $d$-simplex $\sigma\notin \L$.
    The effective capacitance of $\gamma$ in $\K$ is bounded above by $\C_\gamma(\L,\K) \in O \left( n \cdot n_0 \cdot \mathcal{T}_{\max}(\K)^2 \right)$.
\end{theorem}
\begin{proof}
Let $p$ be a $\gamma$-potential. We upper bound the potential energy of $p$. By definition, $\coboundary[\L]p=0$ and $\gamma^{T}p=1$. We can express these constraints as the linear system
$$
  \begin{bmatrix} \phantom{0} \\ \coboundary[\L] \\ \phantom{\vdots} \\ \gamma^T \end{bmatrix} p = \begin{bmatrix} 0 \\ 0 \vphantom{\coboundary[\L]} \\ \vdots  \\ 1 \end{bmatrix}
$$
We first remove linearly-dependent columns from this linear system until this system has full column rank. Columns of the matrix are indexed by $(d{-}1)$ simplices of $\L$, and rows are indexed by $d$-simplices of $\L$.
Removing columns from $\coboundary[\L]$ changes it to the relative coboundary matrix $\delta[\L, \L_0]$ where $\L_0$ is the $(d-1)$-subcomplex corresponding to the columns that were removed. Removing linearly-dependent columns does not change the image of the system of equation, so there is still a solution $r$, i.e.
$$
  \begin{bmatrix} \phantom{0} \\ \coboundary[\L, \L_0] \\ \phantom{\vdots} \\ c^{T} \end{bmatrix} r = \begin{bmatrix} 0 \\ 0 \vphantom{\coboundary[\L]} \\ \vdots  \\ 1 \end{bmatrix}
$$
where $c$ is the subvector of $\gamma$ after removing the columns. The vector $r$ is not a $\gamma$-potential as it is a vector in $C_{d-1}(\L, \L_0)$, not $C_{d-1}(\L)$. However, we can extend $r$ to be a $\gamma$-potential by adding zeros in the entries indexed by $\L_0$. Adding zero-valued entries preserves the length of $r$.
\par
 We now want to remove rows from this matrix so that it has full row rank. Topologically, removing rows corresponds to removing $d$-simplices from the complex $\L$ to create a new complex $\L_1$. Note that we must always include the row $c$ to have full row rank; otherwise, $r$ would be a non-zero vector in the kernel of this system, meaning the system does not have full rank.  Removing these rows gives the linear system
 $$
  \begin{bmatrix} \phantom{0} \\ \coboundary[\L_1, \L_0] \\ \phantom{\vdots} \\ c \end{bmatrix} r = \begin{bmatrix} 0 \\ 0 \vphantom{\coboundary[\L]} \\ \vdots  \\ 1 \end{bmatrix}.
 $$
 Let $C=\begin{bmatrix}  \coboundary[\L_1, \L_0]^T & c^T \end{bmatrix}^T$ and $b=\begin{bmatrix} 0 & 0 & \cdots & 1 \end{bmatrix}^T$. Note that $C$ is an square matrix of size (say) $m\times m$, where $m\leq n+1$.
 \par
  We use Cramer's rule to bound the size of $\| r \|$. By Cramer's rule, $r_i$, the $i$th entry of $r$, is
 $$
  r_i = \frac{\det(C_{i,b})}{\det(C)}.
 $$
 where $C_{i,b}$ is the matrix obtained by replacing the $i$th column with $b$.
 \par 
 We first lower bound $|\det{(C)}|$. As $C$ is a full-rank integral matrix, then $|\det{(C)}|\geq 1$. We now upper bound $|\det(C_{i,b})|$. We calculate $\det(C_{i,b})$ with the cofactor expansion on the column replaced by $b$. As $b$ has 1 in its last entry and 0s elsewhere, the cofactor expansion gives $\det(C_{i,b}) = \pm 1\cdot \det(C_{i,b}^{i,c})$ where $C_{i,b}^{i,c}$ is the matrix where we dropped the $i$th column and the row $c$ from $C_{i,b}$. The matrix $C_{i,b}^{i,c}$ is a square submatrix of $\coboundary[\K]$, so we can bound $|\det(C_{i,b})|\leq \mathcal{T}_{\max}(\K)$. Thus, $r_i = \det(C_{i,b}) / \det(C) \leq \mathcal{T}_{\max}(\K)$ and
 \begin{align*}
 \| r \| &= \sqrt{\sum_{i=1}^m r_i^2} \\
 &\leq \sqrt{m\cdot \mathcal{T}_{\max}(\K)^2} \\
 &\leq \sqrt{n\cdot \mathcal{T}_{\max}(\K)^2} \\
 &= \sqrt{n} \cdot \mathcal{T}_{\max}(\K)
 \end{align*}
 \par
The potential energy of $r$ is $\| \coboundary[\K]r \|^2$. We can use Lemma \ref{lem:upper_bound_largest_singular_value_coboundary} to obtain the bound $\|\delta[\K]r\|^2 =O \left( n\cdot n_0\cdot \mathcal{T}_{\max}(\K)^2 \right)$.
\end{proof}

\subsubsection{Upper Bound on Relative Torsion.}
\label{sec:upper_bound_relative_torsion}

\noindent
To conclude this section, we provide an upper bound on $\mathcal{T}_{\max}(\K)$. 

\begin{lemma}
\label{lem:bound_relative_torsion}
    Let $\K$ be a simplicial complex. Let $n = \min\{n_{d-1},n_d\}$. Then the maximum rank of any $(d-1)$-dimensional relative torsion group of $K$ is $\mathcal{T}_{\max}(\K)\in O((\sqrt{d+1})^{n})$.
\end{lemma}
\begin{proof}
    By the proof of Lemma \ref{lem:detbound}, the quantity $\mathcal{T}_{\max}(\K)$ is the absolute value of the determinant of some submatrix of $\boundary_d$. We therefore bound the determinant of such a submatrix. We prove this bound using \textit{\textbf{Hadamard's Inequality}}: the determinant of an $m\times m$ matrix $B$ is upper-bounded by the product of the norms of its columns.
    \par
    Consider a square, $m\times m$ submatrix $B$ of $\boundary_{d}$. We know that $m\leq n$. Moreover, any column of $B$ has norm bounded above by $\sqrt{d+1}$. This bound follows from the fact that each column of $\boundary_d$ has norm exactly $\sqrt{d+1}$; each column has $d+1$ nonzero entries, each of which is $\pm 1$. The bound of the lemma follows by Hadamard's Inequality.
\end{proof}

\noindent
\Cref{lem:bound_relative_torsion} immediately implies the corollaries to \Cref{thm:resistance_upper_bound_torsion} and \Cref{thm:cap_bound}.

\begin{corollary}
\label{cor:resistance_upper_bound}
    Let $\K$ be a $d$-dimensional simplicial complex and $\gamma$ a unit-length null-homologous $(d{-}1)$-cycle in $\K$. Let $n=\min\{n_{d-1},n_d\}$. The effective resistance of $\gamma$ is bounded above by $\RE_\gamma(\K) \in O \left( n^2 (d+1)^{n} \right)$.
\end{corollary}

\begin{corollary}
    \label{cor:cap_bound}
    Let $\L\subset\K$ be $d$-dimensional simplicial complexes. Let $\gamma\in C_{d-1}(\L)$ be a $(d-1)$-cycle that is null-homologous in $\K$ but not in $\L$. Let $n=\min\{n_{d-1},n_d\}$. Assume that $\gamma=\boundary\sigma$ for a $d$-simplex $\sigma\notin \L_d$.
    The effective capacitance of $\gamma$ in $\K$ is bounded above by $\C_\gamma(\L,\K) \in O \left( n \cdot n_0 \cdot (d+1)^{n} \right)$.
\end{corollary}

\subsection{Lower Bounds.} 
\label{sec:lower_bounds}

\subsubsection{Lower Bounds on Effective Resistance} At the end of the previous section, we gave an exponential upper bound on the effective resistance of a $(d-1)$-cycle in a simplicial complex (\Cref{cor:resistance_upper_bound}). In this section, we describe a $d$-dimensional simplicial complex $\mathcal{B}_d^{n}$ with a $(d-1)$-cycle $\gamma$ with exponentially-large effective resistance with respect to the size of the complex. 

\begin{restatable}{theorem}{thmlowerboundeffectiveresistance}
    \label{thm:resistance_lower_bound}
    Let $d$, $n$ be positive integers. There is a constant $c_d\geq 1$ that depends only on $d$ and a $d$-dimensional simplicial complex $\mathcal{B}_d^n$ with $n_d \in \Theta((d+1)^{3}n)$ $d$-simplices and a unit-length null-homologous cycle $\gamma\in C_{d-1}(\mathcal{B}_d^n)$ such that $\mathcal{R}_{\gamma}(\mathcal{B}_d^n)\in \Theta(c_d^{n_d})$. 
\end{restatable} 

\paragraph{The building block.}

\begin{figure}
    \centering
    \begin{subfigure}{0.3\textwidth}
        \centering
        \vspace*{0.25in}
        \includegraphics[height=1in]{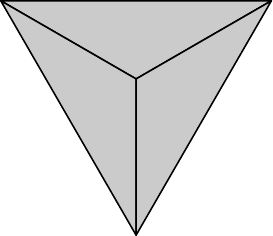}
        \vspace*{0.25in}
    \end{subfigure}
    \begin{subfigure}{0.3\textwidth}
        \centering
        \includegraphics[height=1.5in]{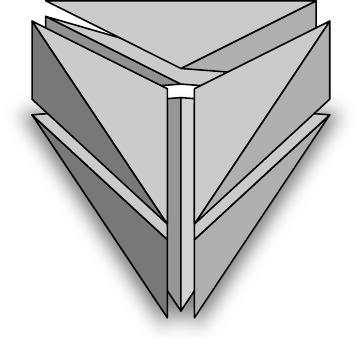}
    \end{subfigure}
    \begin{subfigure}{0.3\textwidth}
        \centering
        \includegraphics[height=1.5in]{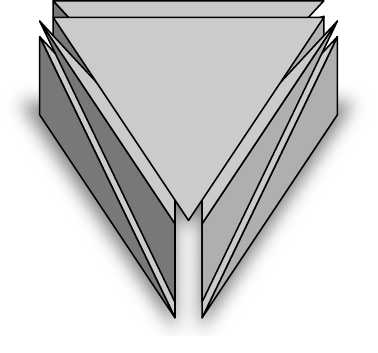}
    \end{subfigure}
    \caption{Left: Stellar subdivision of a triangle. Middle and Right: Top and bottom view of the stellar prism of a triangle with the tetrahedra pushed apart.}
    \label{fig:stellar_figure}
\end{figure}

Our simplicial complex will be obtained by gluing together multiple instances of the same ``building block'' $B_d$. A formal description of $B_d$ is given in Appendix \ref{sec:construction_building_block}; here, we give an intuitive description of the complex. Let $\Delta^{d}$ denote the closure of the $d$-simplex $\sigma=\{v_0,\ldots,v_d\}$, and let $\boundary\Delta^{d}$ denote the $(d-1)$-dimensional simplicial complex $\Delta^{d}\setminus\{\sigma\}$. Our construction starts with a triangulation of the space $\boundary\Delta^{d}\times[0,1]$ that we call the \textit{\textbf{stellar prism}}. The ``bottom copy'' $\boundary\Delta^{d}\times\{0\}$ is triangulated like the original complex $\boundary\Delta^{d}$, and the ``top copy'' $\boundary\Delta^{d}\times\{1\}$ is triangulated with the \textit{\textbf{stellar subdivision}}, the subdivision that adds a vertex to the center of each $(d-1)$-simplex. See Figure \ref{fig:stellar_figure}. The relevant property of this triangulation is that the bottom copy has $d+1$ $(d-1)$-simplices, and the top copy has $d\cdot(d+1)$ $(d-1)$-simplices. The \textit{\textbf{building block}} $B_d$ is obtained from this triangulation by identifying the vertex in the center of each $(d-1)$-simplex $\tau\times\{1\}$ with the unique vertex in $\sigma\times\{1\}\setminus\tau\times\{1\}$, the unique vertex in the \emph{set} $\sigma\times\{1\}$ that is not a vertex of the simplex $\tau\times\{1\}$. See \Cref{fig:complex_B_3}.
\par 
When we identify the vertices, each $(d-1)$-simplex in the top of the stellar prism is replaced by one of the $(d-1)$ faces of $\sigma\times\{1\}$. Moreover, this is replacement is $d$-to-1, meaning each face of $\sigma\times\{1\}$ replaces a $(d-1)$-simplex exactly $d$ times, or informally, each $(d-1)$-simplex ``appears $d$ times'' in the top copy of $B_{d}$. Of course, this is not literally true, as a simplicial complex can only contain a single copy of each simplex. However, something to this effect is true. Namely, there is a $d$-chain $f\in C_d(B_d)$ whose boundary assigns value $\pm 1$ to each $(d-1)$-simplex in the bottom copy and value $\pm d$ to each $(d-1)$-simplex in the top copy. The key properties of the building block $B_d$ are summarized in the following lemma. 

\begin{restatable}{lemma}{Bddchain}
\label{lem:Bd_d_chain}
    Let $\sigma=\{v_0,\ldots, v_d\}$ be a set. There is a $d$-dimensional simplicial complex $B_d$ with vertices $\sigma\times\{0,1\}$ such that 
    \begin{enumerate}
        \item $B_d$ has $\Theta((d+1)^{3})$ $d$-simplices.
        \item there is a $d$-chain $f\in C_{d}(B_{d})$ such that 
        \begin{enumerate}
            \item[(i)] $\boundary f = \boundary(\sigma\times\{0\}) + d\cdot\boundary(\sigma\times\{1\})$
            \item[(ii)] $\|f\|^{2} \in \Omega((d+1)^{3})$
        \end{enumerate}
    \end{enumerate}
\end{restatable}

\noindent Note that neither of the simplices $\sigma\times\{0\}$ or $\sigma\times\{1\}$ are in $B_d$; however, all of their faces are in the complex, so the boundary of these simplices are well-defined.

\paragraph{The total complex.}

\begin{figure}
    \centering
    \includegraphics[width=\linewidth]{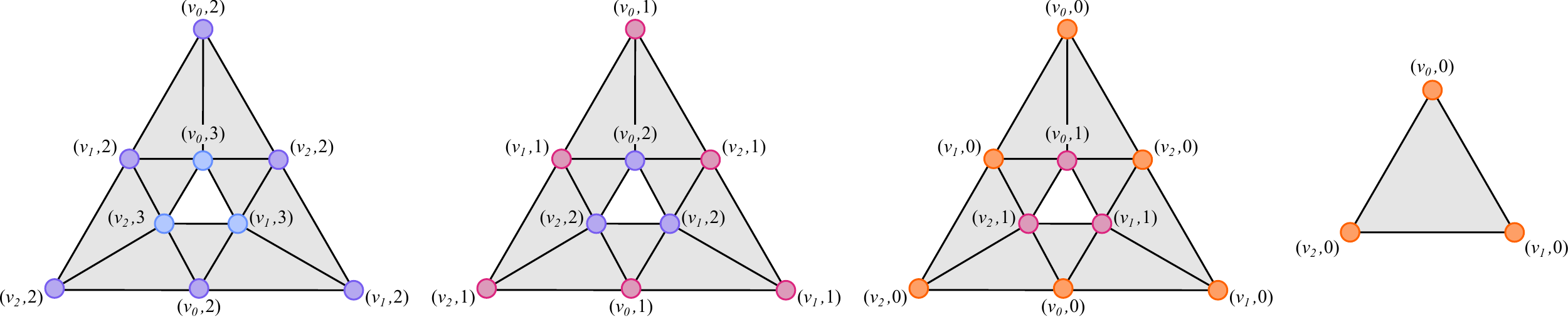}
    \caption{The complex $\mathcal{B}_{2}^{3}$. Recall that the vertices of $\mathcal{B}_{2}^{3}$ are of the form $(v_i,j)$ for $0\leq i\leq 2$ and $0\leq j\leq 3$; colors in the figure denote the second coordinate of the vertices. Vertices with the same label are identified.}
    \label{fig:complex_B_3}
\end{figure}

The complex $\mathcal{B}_d^{n}$ is obtained by gluing together $n$ copies of $B_d$. We describe this gluing inductively on $n$. The vertices of $\mathcal{B}_d^{n}$ are $\sigma\times\{0,\ldots,n\}$. The base case $\mathcal{B}_d^{0}$ is the complete complex on the vertices $\sigma\times\{0\}$. Inductively, the complex $\mathcal{B}_d^{n}$ is obtained from $\mathcal{B}_{d}^{n-1}$ by identifying the vertices $\sigma\times\{n-1\}$ of $\mathcal{B}_{d}^{n-1}$ with the vertices $\sigma\times\{1\}$ of a copy of $B_d$. We will denote this copy of $B_d$ as $B_d^{n}$ and the vertices $\sigma\times\{0\}$ of $B^{n}_d$ as $\sigma\times\{n\}$. See \Cref{fig:complex_B_3}.
\par 
The key property of $\mathcal{B}_d^{n}$ is that is has an exponentially-large chain with constant-sized boundary.

\begin{lemma}
\label{lem:Bdn_d_chain}
    Let $d\geq 1$ and $n\geq 1$. There is a $d$-chain $y_n\in C_d(\mathcal{B}_d^{n})$ such that 
    \begin{enumerate}
        \item $\|y_n\|^{2} \in \Omega(d^{2n}(d+1))$
        \item $\|\boundary y_n\|^{2} = d+1$
    \end{enumerate}
\end{lemma}
\begin{proof}
     We construct this chain by induction on $n$. The chain $y_n$ will have $\boundary y_n = \boundary(\sigma\times\{n\})$. For the base case $n=0$, the chain $y_0 = \sigma\times\{0\}$ clearly has this property.
    \par 
    For the inductive case, recall from Lemma \ref{lem:Bd_d_chain} that there is a $d$-chain $f\in C_d(B_d)$ such that $\boundary f = \boundary(\sigma\times\{0\}) + d\cdot\boundary(\sigma\times\{1\})$. Let $f_n$ denote this chain in $B_d^{n}$. We define the chain $y_n := f_n - d\cdot y_{n-1}$. We now verify that $y_n$ has boundary $\boundary(\sigma\times\{n\})$.
    \begin{align*}
        \boundary y_n =& \boundary f_n - d\cdot\boundary y_{n-1} \\
        =& \boundary(\sigma\times\{n\}) + d\cdot\boundary(\sigma\times\{n-1\}) - d\cdot \boundary(\sigma\times\{n-1\} \\
        =& \boundary(\sigma\times\{n\})
    \end{align*}
    It is clear that $\|\boundary y_n\|^2 = d+1$, so we just need to lower bound $\|y_n\|^2$. We prove that $\|y_n\|^2\in\Omega(d^{2n}(d+1))$ by induction. For the base case of $n=1$, we know that $\|y_1\|^2 \in \Omega((d+1)^3)$ by Lemma \ref{lem:Bd_d_chain}. For the inductive case, we can see that $f_n\perp y_{n-1}$ as these chains are supported on different sets of simplices. Therefore, $\|y_n\|^2 = \|f_n\|^2 + d^{2}\cdot\|y_{n-1}\| \in \Omega(d^{2n}(d+1))$. 
\end{proof}

Lemma \ref{lem:Bdn_d_chain} shows that $\mathcal{B}_d^{n}$ has an exponentially-large $d$-chain $f$ with constant-sized boundary. If we can show that $\ker\boundary_d=0$, this will prove that $\boundary f$ has exponentially-large effective resistance, as $f$ will be the \textit{only} $d$-chain with boundary $\boundary f$.

\begin{corollary}
\label{lem:Bdn_ker_boundary}
    The kernel of the boundary matrix $\boundary_d[\mathcal{B}_d^{n}]$  is trivial, i.e. $\ker\boundary_d[\mathcal{B}_d^{n}]=0$.
\end{corollary}
\begin{proof}
    Lemma \ref{lem:Bdn_collapses} in Appendix \ref{apx:collapsibility} shows that $\mathcal{B}_d^{n}$ collapses to a $(d-1)$-dimensional subcomplex; call this subcomplex $L$. This implies that $\mathcal{B}_d^{n}$ is homotopy equivalent to $L$, so in particular, the $d$-dimensional homology $H_d(\mathcal{B}_d^{n})=0$. As $\mathcal{B}_d^{n}$ is $d$-dimensional, then $\im\boundary_{d+1}[\mathcal{B}_d^{n}]=0$, so the only way that $H_d(\mathcal{B}_d^{n})$ can equal 0 is if $\ker\boundary_d[\mathcal{B}_d^{n}]=0$.
\end{proof}

\begin{proof}[Proof of Theorem \ref{thm:resistance_lower_bound}]
     Our cycle is the normalized cycle $\gamma = \boundary y_n/\|\boundary y_n\|$ where $y_n$ is the $d$-chain from Lemma \ref{lem:Bdn_d_chain}. We know that $\|y_n\|^2 / \|\boundary y_n\|^2 \in\Omega(d^{2n})$. Moreover, we know that $\ker\boundary_d=0$ by Corollary \ref{lem:Bdn_ker_boundary}. Therefore, we conclude that effective resistance $\mathcal{R}_{\gamma}(\mathcal{B}_d^{n})\in \Theta(d^{2n})$. 
    \par 
    Now we need to restate this bound in terms of the number of $d$-simplices of $\mathcal{B}_d^{n}$. Each copy of $B_d$ has $\Theta((d+1)^{3})$ $d$-simplices. Therefore, the entire complex $\mathcal{B}_d^{n}$ has $n_d = \Theta(n\cdot (d+1)^{3})$ $d$-simplices. If we substitute $n_d$ into our bound, we find that the effective resistance is bound below by $O((d+1)^{2n_d/c\cdot(d+1)^{3}})$ for some constant $c$. Therefore, the constant in the theorem statement is $c_d=(d+1)^{2/c\cdot(d+1)^{3}}$.
\end{proof}

\paragraph{Related Work.}

Variants of this construction are sometimes called the \textit{Iterated Mapping Cylinder} and have been used as a worst-case construction for other topological properties like torsion~\cite{Newman2018}, homotopy~\cite{Filakovský2018}, or embeddability~\cite{brehm1992linear, Freedman2014}. However, ours is the first work showing these complexes have a cycle with exponentially-large effective resistance (or exponentially-small spectral gap, as we will see later.) We were specifically inspired by the work of Newman~\cite{Newman2018}.
\par 
However, our construction is more efficient than previous constructions of the Iterated Mapping Cylinder in terms of the number of $d$-simplices; this is why we dedicate several pages in the appendix to this construction. As an example, the building block in Newman's construction is the iterated suspension of the M\"obius band, and the number of $d$-simplices in the suspension grows exponentially with the dimension $d$. Comparatively, the number of $d$-simplices in our building block only grows polynomially with the dimension. Additionally, our construction has a cycle with $\pm 1$ coefficients that is homologous to a cycle with $\pm d$ coefficients; in contrast, previous works have a cycle with $\pm 1$ coefficients that is homologous to a cycle with $\pm 2$ coefficients. Both of these properties result in a larger constant $c_d$ in \Cref{thm:resistance_lower_bound}. It is an open question if there is a simplicial complex with a cycle whose effective resistance exactly matches the constant of the lower bound, i.e. $c_d\in\Theta(d+1)$. 

\subsubsection{Lower Bounds on Capacitance.}

In this section, we describe a pair of nested simplicial complexes that have a cycle with exponentially large effective capacitance. This complex will be built from the same building block as the complex for lower bounding effective resistance, but the simplicial complex will built by gluing together the building blocks in slightly different ways.
\par 
Recall that the building block is a simplicial complex $B_d$ with vertices $\sigma\times\{0,1\}$, where $\sigma=\{v_0,\ldots,v_d\}$ is a $d$-simplex. For a natural number $n$, we recursively construct the simplicial complex $\mathcal{Q}_d^{n}$. The simplicial complex $\mathcal{Q}_d^{n}$ is obtained by gluing together $n$ copies of $B_d$. The vertices of $\mathcal{Q}_d^{n}$ will be denoted $\sigma\times\{0,\ldots,n\}$. The base case $\mathcal{Q}_d^{n}$ is the complete complex on the vertices $\sigma\times\{0\}$. Inductively, the complex $\mathcal{Q}_d^{n}$ is obtained from $\mathcal{Q}_d^{n-1}$ by identifying the vertices $\sigma\times\{n-1\}$ of $\mathcal{Q}_d^{n-1}$ with the vertices $\sigma\times\{0\}$ of a copy of $B_d$.\footnotemark We will denote this copy of $B_d$ as $B_d^{n}$ and the vertices $\sigma\times\{1\}$ as $\sigma\times\{n\}$. The simplicial complex $\mathcal{P}_d^{n}$ is defined $Q_d^{n-1}$ minus the simplex $\sigma\times\{0\}$.
\footnotetext
{ 
    It is worth comparing $\mathcal{Q}_d^{n}$ to the simplicial complex $\mathcal{B}_d^{n}$ from the bounds on effective resistance. Both complexes are obtained by identifying vertices of copies of the building block $B_d$, but the identifications are made in different ways. When a building block is added to $\mathcal{B}_d^{i}$ to form $\mathcal{B}_d^{i+1}$, its ``top'' $\sigma\times\{1\}$ is identified with $\sigma\times\{i\}$. When a building block is added to $\mathcal{Q}_d^{i}$ to form $\mathcal{Q}_d^{i+1}$, its ``bottom'' $\sigma\times\{0\}$ is identified with $\sigma\times\{i\}$. It is easiest to see the difference between $\mathcal{B}_d^{n}$ and $\mathcal{Q}_d^{n}$ by considering the cycle $\boundary(\sigma\times\{n\})$ in both complexes. The unique $d$-chain $f\in C_d(\mathcal{B}_d^{n})$ with boundary $\boundary(\sigma\times\{n\})$ assigns exponentially-large coefficients to some $d$-simplices. In contrast, the unique $d$-chain $f\in C_d(\mathcal{Q}_d^{n})$ assigns exponentially-small coefficients to some $d$-simplices. (This is a corollary to the proof of~\Cref{lem:Pdn_Bdn_d_chain}.)  
}
\par 
We will prove that the cycle $\gamma = \boundary(\sigma\times\{n\}$ has exponentially large capacitance in $\mathcal{P}_d^{n}\subset \mathcal{Q}_d^{n}$. In order for $\gamma$ to have finite capacitance, the cycle $\gamma$ must be null-homologous in $\mathcal{Q}_d^{n}$ but not null-homologous in $\mathcal{P}_d^{n}$. We prove this in the following lemma.

\begin{lemma}
\label{lem:Pdn_Bdn_d_chain}
    Let $\mathcal{P}_d^{n}$ and $\mathcal{Q}_d^{n}$ as described above. Consider the chain $\boundary(\sigma\times\{n\})\in C_d(\mathcal{P}_d^{n})$. 
    \begin{enumerate}
        \item The $(d-1)$-chain $\boundary(\sigma\times\{n\}) + (-1)^{n-1}\frac{1}{d^{n}}\boundary(\sigma\times\{0\})$ is null-homologous in $\mathcal{P}_d^{n}$.
        \item The $(d-1)$-chain $\boundary(\sigma\times\{n\})$ is null-homologous in $Q_d^{n}$.
        \item The $(d-1)$-chain $\boundary(\sigma\times\{n\})$ is not null-homologous in $P_d^{n}$.
    \end{enumerate}
\end{lemma}
\begin{proof}
    \textit{Proof of Part (1)} We will prove that the cycle $\boundary(\sigma\times\{i\}) + \frac{1}{d^{i}}\boundary(\sigma\times\{0\})$ is null-homologous in $Q_d^{i}$ by induction on $i$. Specifically, we will find a $d$-chain $y_i$ such that $\boundary y_i = \boundary(\sigma\times\{i\}) + \frac{1}{d^{i}}\boundary(\sigma\times\{0\})$. For each copy of the building block $B_d^{i}$, let $f_i\in C_d(B_d^{i})$ be the $d$-chain guaranteed by Lemma~\ref{lem:Bd_d_chain} such that $\boundary f_i = \boundary(\sigma\times\{i-1\}) + d\cdot\boundary(\sigma\times\{i\})$. For the base case of $i=1$, the $d$-chain $y_1=\frac{1}{d} f_1$. Inductively, we define the chain $y_{i} = \frac{1}{d} f_i - \frac{1}{d} y_{i-1}$. We can verify that $y_i$ has the claimed boundary as 
    \begin{align*}
        \boundary y_i =& \frac{1}{d} \boundary f_i - \frac{1}{d} \boundary y_{i-1}  \\
        =& \boundary(\sigma\times\{i\}) + \frac{1}{d}\boundary(\sigma\times\{i-1\}) - \frac{1}{d}\boundary(\sigma\times\{i-1\}) + (-1)^{i-1}\frac{1}{d^{i}}\boundary(\sigma\times\{0\}) \\ 
        =& \boundary(\sigma\times\{i\}) + (-1)^{i-1}\frac{1}{d^{i}}\boundary(\sigma\times\{0\}) 
    \end{align*}
    \par 
    \textit{Proof of Part (2)} As the simplex $(\sigma\times\{0\})\in\mathcal{Q}_d^{n}$, then the chain $y_n + (-1)^{n}\frac{1}{d^{n}}\cdot(\sigma\times\{0\})$ obviously has boundary $\boundary(\sigma\times\{n\})$.
    \par
    \textit{Proof of Part (3)} To prove that $\boundary(\sigma\times\{n\})$ is not null-homologous in $\mathcal{P}_d^{n}$, we will use two lemmas we prove in the appendix. \Cref{lem:Pdn_collapses} shows that $\mathcal{P}_d^{n}$ collapses to a $(d-1)$-dimensional subcomplex---call it $L$---containing the support of $\boundary(\sigma\times\{n\})$. \Cref{lem:collapsibility_and_null_homology} shows that for nested complexes $L\subset K$ such that $K$ collapses to $L$, a cycle is null-homologous in $L$ if and only if it is null-homologous in $K$. As $\boundary(\sigma\times\{n\})$ is not null-homologous in $L$ ($L$ is $(d-1)$-dimensional, so $\im\boundary_d[L]=0$), then $\boundary(\sigma\times\{n\})$ is not null-homologous in $\mathcal{P}_d^{n}$ either.
\end{proof}

\begin{restatable}{theorem}{thmlowerboundeffectivecapacitance}
    \label{thm:capacitance_lower_bound}
    Let $d$, $n$ be positive integers. There is a pair of nested $d$-dimensional simplicial complexes $\mathcal{P}_d^n\subset\mathcal{Q}_d^{n}$ with $n_d \in \Theta((d+1)^{3}n)$ $d$-simplices, a unit-length null-homologous cycle $\gamma\in C_{d-1}(\mathcal{Q}_d^n)$, and a constant $c_d\geq 1$ that depends only on $d$ such that $\C_{\gamma}(\mathcal{P}_d^n, \mathcal{Q}_d^n)\in \Theta(c_d^{n_d})$. 
\end{restatable}
\begin{proof}
    The complexes $\mathcal{P}_{d}^{n}$ and $\mathcal{Q}_d^{n}$ are the complexes described in the preceding paragraphs. The cycle $\gamma=\boundary(\sigma\times\{n\})/\sqrt{d+1}$, where $\sqrt{d+1}$ is a normalization factor. We must show that any unit $\gamma$-potential $p$ has exponentially-large potential energy. 
    \par 
    Let $p$ be any unit $\gamma$-potential in $\mathcal{P}_d^n$. We know that $\gamma^{T}p=1$ and $\coboundary_{d-1}[\mathcal{P}_n^{d}]p=0$. As $\coboundary_{d-1}[\mathcal{P}_n^{d}] = \boundary_{d}^{T}[\mathcal{P}_n^{d}]$, the second condition is equivalent to saying that $p^{T}b=0$ for any vector $b\in\im\boundary_{d}[\mathcal{P}_d^{n}]$. \Cref{lem:Pdn_Bdn_d_chain} Part 1 proves that $\boundary(\sigma\times\{n\})+(-1)^{n-1}d^{-n}\boundary(\sigma\times\{0\})\in\im\boundary_d[\mathcal{P}_d^{n}]$, so the previous two facts imply $p^{T}\big(\boundary(\sigma\times\{n\})+(-1)^{n-1}d^{-n}(\boundary(\sigma\times\{0\}))\big)=0$. As $p^{T}\boundary(\sigma\times\{n\})=\sqrt{d+1}$ and $p^{T}\big(\boundary(\sigma\times\{n\})+(-1)^{n-1}d^{-n}(\boundary(\sigma\times\{0\}))\big)=0$, then we conclude that $p^{T}\boundary(\sigma\times\{0\}) =(-1)^{n-1}d^{n}/\sqrt{d+1}$. Moreover, as the only $d$-simplex in $\mathcal{Q}_d^{n}$ that is not in $\mathcal{P}_d^{n}$ is $\sigma\times\{0\}$, then the fact that $\coboundary_{d-1}[\mathcal{P}_d^{n}]p=0$ implies that the potential energy $\|\coboundary_{d-1}[\mathcal{Q}_d^{n}]p\|^{2} = \big(p^{T}\boundary(\sigma\times\{0\})\big)^{2} = \Omega(d^{2n})$. 
    \par 
    This shows that $p$ has exponentially-large potential energy with respect to $n$. By the same argument as in the proof of~\Cref{thm:resistance_lower_bound}, we can show that $p$ also has exponentially-large potential energy with respect to $n_d$, but for a different (but constant) base of the exponent $c_d$. 
\end{proof}

\section{Bounds on the Spectral Gap}
\label{sec:bounds_spectral_gap}

Our lower and upper bounds on effective resistance imply lower and upper bounds on the spectral gap of the combinatorial Laplacian. This is because \textit{the spectral gap is the inverse of the maximum effective resistance of all unit-length, null-homologous cycles}, a fact we proved in Lemma~\ref{lem:spectral_gap_effective_resistance}. Therefore, a corollary of Theorem \ref{thm:resistance_lower_bound} is that the spectral gap of the combinatorial Laplacian can be exponentially-small in the worst-case. This resolves one of the most important open questions in the field of Quantum Topological Data Analysis and shows that Betti Number Estimation algorithms must run for an exponentially-long time to exactly compute Betti numbers.

\subsection{Exponentially-small spectral gap.}

Based on this connection between the spectral gap and effective resistance of Lemmas~\ref{lem:spectral_gap_effective_resistance} and \ref{lem:spectral_gap_combinatorial_up_down}, we can derive lower and upper bounds on the spectral gap of the $d$-combinatorial Laplacian as corollaries of the upper and lower bounds on the effective resistance (\Cref{cor:resistance_upper_bound} and \Cref{thm:resistance_lower_bound}). While lower bounds on the spectral gap were previously known (see for example \cite[Proof of Theorem 1.2]{Freedman2014}), one advantage of our proof is that it provides a necessary condition for large spectral gap, namely the existence of a subcomplex with exponentially-large relative torsion.

\spectralgaplowerbound*

\begin{restatable}{theorem}{coronesmalleigenvalue}
    \label{thm:spectral_gap_upper_bound}
    Let $d$, $n\geq 1$. There is a $d$-dimensional simplicial complex $\mathcal{B}_d^n$ with $n_d \in \Theta((d+1)^{3}\cdot n)$ $d$-simplices and a constant $c_d\geq 1$ that depends only on $d$ such that the spectral gaps of $L_{d-1}[\mathcal{B}_d^{n}]$ and $L_d[\mathcal{B}_d^{n}]$ are $O(\frac{1}{c_d^{n_d}})$. 
\end{restatable}

We remark that we can also derive lower bound of the spectral gap in terms of relative torsion as a corollary to Theorem~\ref{thm:resistance_upper_bound_torsion}. This implies bounded spectral gap for simplicial complexes with no relative torsion, as remarked upon by Friedman~\cite[Theorem 7.2]{Friedman1998BettiNumbers} in the case of orientable $d$-manifolds.

\subsection{Many Small Eigenvalues.}

Corollary~\ref{thm:spectral_gap_upper_bound} shows there is a simplicial complex with a single small eigenvalue. It is natural to ask whether there is a bound on the number of very small eigenvalues a simplicial complex can have. This is relevant to QTDA algorithms that work by counting the number of eigenvalues smaller than a given threshold. Here, we provide a complex ${\mathcal M}_{d}^{n}$ with a polynomial number of exponentially-small eigenvalues. 

\begin{restatable}{corollary}{thmmanysmalleigenvalues}
\label{thm:many_small_eigenvalues}
    Let $n, d\geq 1$. There exists a simplicial complex ${\cal M}_{d}^{n}$ with $n_d\in \Theta((d+1)^{3}\cdot n)$ $d$-simplices and a constant $c_d>1$ that depends only on $d$ such that both $L_{d-1}[{\cal M}_{d}^{n}]$ and $L_{d}[{\cal M}_{d}^{n}]$ have $\Omega(\sqrt{n_d})$ eigenvalues of size $O(\frac{1}{c_d^{\sqrt{n_d}}})$.
\end{restatable}

\begin{proof}[Proof of Theorem \ref{thm:many_small_eigenvalues}]
    This follows from Part 2 of \Cref{lem:properties_spectrum_laplacian} relating the spectrum of a complex to the spectra of its connected components. We define the complex ${\cal M}_{d}^{n}$ to be $n_d$ disjoint copies of the ${\cal B}_{d}^{n}$ of Corollary \ref{thm:spectral_gap_upper_bound}. The complex ${\cal M}_{d}^{n}$ has $n_d^2$ $d$-simplices and $n_d$ eigenvalues of size $O(\frac{1}{c_d^{n_d}})$.
\end{proof}

\subsection{Clique-Dense Complexes.}

Existing QTDA algorithms perform best when the simplicial complex is clique-dense, meaning that the simplicial complex has close to the maximal number of $d$-simplices, i.e $n_{d}\sim \binom{n_0}{d+1}$; However, the simplicial complex we constructed in \Cref{thm:spectral_gap_upper_bound} is sparse: it only has $n_d\in\Theta((d+1)^{2} n_0)$ $d$-simplices. Therefore, \Cref{thm:spectral_gap_upper_bound} does not rule out the possibility that clique-dense complexes avoid worst-case spectral gap.
\par 
However, in this section, we show that we can extend the construction of \Cref{thm:spectral_gap_upper_bound} to clique-dense complexes (at the expense of making the constant $c_d$ of the exponent smaller.) To do this, we use a probabilistic coloring argument of Newman \cite{Newman2018} that reduces the number of vertices of a simplicial complex while preserving the number of $d$-simplices and the Laplacian.
\par 
We begin with definitions. A \textit{\textbf{coloring}} of a simplicial complex is a map on its vertices $c:\K_{0}\to\mathbb{N}$. The \textit{\textbf{pattern complex}} of a simplicial complex $K$ with coloring $c$ is the simplicial complex $\K_c=\{\{c(v) : v\in\sigma\}:\sigma\in \K \}$; intuitively, the pattern complex of $\K$ is the simplicial complex obtained by identify all vertices of $\K$ of the same color and identifying all simplices whose vertices have the same set of colors. While a simplex in $\K$ may be mapped to a lower-dimensional simplex in the pattern complex $\K_c$ if two of its vertices are the same color, we are only considered with colorings where this does not happen. A \textit{\textbf{proper coloring}} of a $d$-dimensional simplicial complex $\K$ is a coloring $c:\K_0\to \mathbb{N}$ such that (1) the endpoints of each edge in $\K$ are different colors and (2) $\{c(v):v\in\sigma_{1}\}\neq \{c(v):v\in\sigma_{2}\}$ for any distinct $(d-1)$-simplices $\sigma_{1},\sigma_{2}\in K$. Note that for a proper coloring, condition (1) guarantees that each simplex in $\K$ corresponds to a simplex of the same dimension in $\K_c$. Additionally, $\K$ and $\K_c$ have the same set of $(d-1)$ and $d$-simplices up to recoloring. Proper colorings are relevant to our paper as they preserve the spectral gap. 

\begin{lemma}
\label{lem:proper_coloring_preserves_laplacian}
    Let $\K$ be a $d$-dimensional simplicial complex and let $c$ be a proper coloring of $\K$. Then $\lambda_{\min}(L_{d-1}^{up}[\K])=\lambda_{\min}(L_{d-1}^{up}[\K_c])$
\end{lemma}
\begin{proof}
    This follows as $\K$ and $\K_c$ have the same set of $(d-1)$ and $d$-simplices up to recoloring, so the boundary maps $\boundary_d[\K]$ and $\boundary_d[\K_c]$ are the same up to the signs on the simplices; however, the spectrum of the up Laplacians $L_{d-1}^{up}[\K]$ and $L_{d-1}^{up}[\K_c]$ are unaffected by different orientations of the simplices (\cite[Theorem 4.1.1]{Goldberg2002}), so the lemma follows.
\end{proof}

Newman's method also requires bounds on a generalized notion of degree. For a $d$-dimensional simplicial complex $\K$ and natural numbers $0\leq i<j\leq d$, define $\Delta_{i,j}(\K)= \max_{\sigma\in\K_i}|\{\tau\in K_{j} : \sigma\subset\tau\}|$ and $\Delta(\K) = \max_{1\leq i < j\leq d}\Delta_{i,j}(\K)$. Newman used a probabilistic argument to show that for a $d$-dimensional simplicial complex $\K$ there was always a proper coloring with a bounded number of colors. 

\begin{lemma}[Lemma 3, Newman~\cite{Newman2018}]
\label{lem:existence_proper_coloring}
    Let $\K$ be a $d$-dimensional simplicial complex such that $\Delta(K)\geq 4$. Then there is a proper coloring $c$ of $\K$ with at most $18(\Delta(K)+1)^{6}d^{6}\sqrt[d]{n_0}$ colors.
\end{lemma}

This implies the following corollary of our bound on the spectral gap.

\thmspectralgapupperbounddense*

\begin{proof}
    This is a corollary to \Cref{thm:spectral_gap_upper_bound}, \Cref{lem:proper_coloring_preserves_laplacian}, and \Cref{lem:existence_proper_coloring}. The simplicial complex $\mathcal{C}_d^{n}$ is a pattern complex of a coloring $c$ of the complex $\mathcal{B}_d^{n}$ from \Cref{thm:spectral_gap_upper_bound}. This coloring $c$ is the coloring guaranteed by \Cref{lem:existence_proper_coloring}. This coloring has $f(d)\sqrt[d]{n_0}$ colors for some function $f$; we can see this by bounding $\Delta(\mathcal{B}_d^{n})$ by a function of $d$. Examining its construction in \Cref{sec:lower_bounds}, each simplex in $\mathcal{B}_d^{n}$ is in at most two different building blocks. Each building block has $2(d+1)$ vertices, so for any $1\leq i<j\leq d$, an $i$-simplex in $\mathcal{B}_d^{n}$ is incident to at most $\binom{4d-i-1}{j-i}$ $j$-simplices. Thus, $\Delta(\mathcal{B}_d^{n})$ is at most some function of $d$. Therefore, there is a simplicial complex with $\tilde{n_0} = f(d)\sqrt[d]{n_0}$ vertices and $\theta\left(\frac{(d+1)^{3}}{f(d)^{d}}\tilde{n_0}^{d}\right)$ $d$-simplices. As $\binom{n_0}{d} \leq \left(\frac{e\tilde{n_0}}{d}\right)^{d}$, then there are $\Omega(\kappa_d\binom{\tilde{n_0}}{d})$ $d$-simplices in $\mathcal{C}_d^{n}$ for some appropriate function $f(d)=\Omega\left(\frac{(d+1)^{3}}{f(d)^{d}e^{d}}\right)$. 
\end{proof}

\subsection{Variants of the Laplacian.}

We finish this section by showing how our results imply upper and lower bounds on the spectral gap of several variants of the Laplacian.

\subsubsection{Boundary Matrix.}

The QTDA algorithm of McArdle, Gily\'{e}n, and Berta~\cite{mccardle2023streamlined} is not parameterized by the spectral gap of the combinatorial Laplacian; rather, it is parameterized by the spectral gap of the boundary matrices. However, the non-zero singular values of the $d$\textsuperscript{th} boundary matrix $\boundary_d$ are the square roots of the eigenvalues of the $(d-1)$\textsuperscript{st} up Laplacian $L_{d-1}^{up}$ as $L_{d-1}^{up} = \boundary_d\boundary_d^{T}$. Therefore, \Cref{thm:spectral_gap_lower_bound} and \Cref{thm:spectral_gap_upper_bound} imply exponential upper and lower bounds on the spectral gap of the boundary matrix.

\subsubsection{Normalized Laplacian.}

We now show that the normalized up Laplacian $\Tilde{L}_d^{up}$ can also have exponentially-small spectral gap. While the eigenvalues of the unnormalized $d$\textsuperscript{th} up Laplacian $L_d^{up}$ are in the range $[0,\,n_0]$, the eigenvalues of the normalized $d$\textsuperscript{th} up Laplacian are in the range $[0,\,d+2]$ \cite[Theorem 3.2.i]{Horak2013}. As the normalized up Laplacian has a constant upper bound on its eigenvalues, it is reasonable to suspect the normalized up Laplacian also has a constant lower bound on its eigenvalues. Theorem~\ref{thm:normalized_spectral_gap} shows this is not the case. 

\begin{restatable}{corollary}{thmNormalizedSpectralGap}
\label{thm:normalized_spectral_gap}
    Let $d$, $n\geq 1$. There is a $d$-dimensional simplicial complex $\mathcal{B}_d^n$ with $n_d \in \Theta(\poly(d)\cdot n)$ $d$-simplices and a constant $c_d\geq 1$ that depends only on $d$ such that the spectral gap the normalized up Laplacian is $\lambda_{\min}(\Tilde{L}_d^{up})\in O(\frac{1}{c_d^{n_d}})$.
\end{restatable}
\begin{proof}
    This follows from \Cref{thm:spectral_gap_lower_bound} and \Cref{lem:normalized_vs_unnormalized_spectral_gap}. The statement follows as $d_{\min}[\mathcal{B}_d^{n}]= 1$.  
\end{proof}

\subsubsection{Persistent Laplacian.} 

Recently, Wang, Nguyen, and Wei~\cite{wang2020persistent} introduced the \textit{\textbf{persistent Laplacian}} of simplicial filtrations as a generalization of the combinatorial Laplacian. The spectral gap of the persistent Laplacian has since appeared as a parameter of quantum algorithms for computing persistent Betti number~\cite{Hayakawa2022quantumalgorithm}, so lower bounding it is also of interest to QTDA. M\'{e}moli, Wan, and Wang~\cite[Theorem SM5.8]{memoli2022persistent} prove that persistent Laplacians preserve effective resistance of cycles; therefore, the bounds on the spectral gap of the combinatorial Laplacian also apply to the persistent Laplacian by \Cref{lem:spectral_gap_effective_resistance}.

\section{Conclusion and Open Questions.}

In this paper, we propose a new span-program-based quantum algorithm for computing Betti numbers. This algorithm is a novel approach to QTDA that is more similar to classical incremental algorithms for computing Betti numbers than previous QTDA algorithms. Unfortunately, we show that, in the worse case, the span-program based algorithm takes exponential time due to cycles with exponentially-large effective resistance or effective capacitance. However, as a corollary to exponentially-large effective resistance, we prove that the spectral gap of the combinatorial Laplacian can be exponentially small. This proves that all known QTDA algorithms also require exponential time in the worst case. Below we discuss some of the questions left open by our work.

\paragraph{Incremental Quantum Algorithm for Persistent Betti numbers.}

Our algorithm incrementally computes the Betti number of a simplicial complex. While the classical algorithm for computing \textit{persistent} Betti numbers is incremental~\cite{zomorodian2004computing}, our algorithm is unable to perform persistent pairing. In other words, our algorithm can identify when a homology class dies, but it cannot identify when that homology class was born. It is an open question whether our algorithm can be adapted to compute the persistent Betti numbers of a simplicial complex. There are quantum algorithms for computing persistent Betti numbers~\cite{Hayakawa2022quantumalgorithm, mccardle2023streamlined}, but these algorithms are not incremental.

\paragraph{Lower Bounds or Expectation of the Spectral Gap.}

\Cref{thm:spectral_gap_upper_bound} shows that the spectral gap of the combinatorial Laplacian can be exponentially small. However, it is an open question how common these sorts of worst-case complexes are. While there are exact or expected lower bounds for certain families of simplicial complexes~\cite{beit2020spectral, Friedman1998BettiNumbers, gundert2016eigenvalues,  lew2020spectralb, lew2020spectral, lew2023garland, STEENBERGEN201456, yamada2019spectrum}, it is still unknown what the expected spectral gap is, or if there are lower bounds on the spectral gap, for all simplicial complexes, or for families of simplicial complexes of interest like Vietoris-Rips complexes.   

\paragraph{Cheeger Inequalities and Implications of Exponentially-Small Spectral Gaps.} 

The existence of simplicial complexes with exponentially-small spectral gap implies that existing QTDA algorithms cannot exactly compute Betti numbers without running for an exponentially long time; however, they can solve the related problem of \textit{\textbf{Approximate Betti Number Estimation}}~\cite{gyurik2022advantage} of counting the number of eigenvalues of the Laplacian smaller than a given threshold. However, it remains an open question how useful approximate Betti number estimation is in practice.
\par 
A potential interpretation for approximate Betti number estimation could come in the form of a higher-dimensional Cheeger inequality. The \textit{\textbf{Cheeger inequality}} in graphs relates the smallest eigenvalue(s) of the graph Laplacian to a value called the \textit{\textbf{Cheeger constant}} that measures the existence of (multi-way) sparse graph cuts~\cite{chung1996laplacians, lee2014multiway}. Intuitively, if the Hodge Theorem~(\Cref{thm:hodge}) says that the graph Laplacian has more than one zero eigenvalue if and only if the graph is disconnected, then the Cheeger inequality says that it has small non-zero eigenvalues if and only if it is ``almost'' disconnected.
\par 
However, higher-dimensional generalizations of the Cheeger inequality remain elusive. Ideally, a higher-dimensional Cheeger inequality would say something similar: a simplicial complex ``almost has non-trivial $d$-homology'' or ``has a sparse cut'' if and only the $d$\textsuperscript{th} combinatorial Laplacian has small non-zero eigenvalues. One hurdle is that it is not clear how to generalize the notion of ``sparse cut'' to higher dimensions. While there have been several definitions proposed for a Cheeger constant for higher-dimensional Laplacians~\cite{Gromov2010Expanders, LinialMeshulam2006HomConn, MeshulamWallach2009HomConn, ParzanchevskiEtal2016IsoperSimpComp}, one or both sides of a Cheeger inequality have failed for these constants~\cite{Gunder2015HighDimCheeger, gundert2016eigenvalues, ParzanchevskiEtal2016IsoperSimpComp, SteenBergen2014CheegerType}. Our work provides another counterexample to these Cheeger inequalities; the spectral gap of our worst-case complexes is exponentially small, but the proposed notions of Cheeger constant cannot be.
\par
A recent paper presents a two-sided Cheeger inequality~\cite{jost2023cheeger} that connects the spectral gap of the combinatorial Laplacian to a Cheeger constant based on the 1-norm of chains. This Cheeger inequality does not have the same interpretation as the graph Cheeger inequality of implying that simplicial complexes with small eigenvalues ``almost have non-trivial homology'' though. It remains an open question if such a higher-dimensional Cheeger inequality exists.

\section*{Acknowledgements.}
\noindent
This work was supported by NSF grants CCF-1816442 and CCF-1617951

\bibliographystyle{plain}
\bibliography{quantum}

\appendix

\section{Proof of \Cref{lem:normalized_vs_unnormalized_spectral_gap}}
\label{apx:normalized_vs_unnoramlized}

\lemnormalizedvsunnormalizedspectralgap*

\begin{proof}
    Following the same steps as in the proof of Lemma~\ref{lem:spectral_gap_effective_resistance}, we can see that $\lambda_{\min}(\tilde{L}_d^{up}) = \lambda_{\max}^{-1}((\tilde{L}_d^{up})^{+})$. Therefore, we will prove the follow equivalent statement:
    $$
        d_{\min} \lambda_{\max}((L_d^{up})^{+}) \leq \lambda_{\max}((\tilde{L}_d^{up})^{+}) \leq d_{\max} \lambda_{\max}((L_d^{up})^{+})
    $$
    Also from the proof of \Cref{lem:spectral_gap_effective_resistance}, we see that $\lambda_{\max}((\tilde{L}_d^{up})^{+}) = \max\{x^{T}(\tilde{L}_d^{up})^{+}x : \|x\|=1,\,x\in\im D^{-1/2}\boundary_{d+1} \}$. For a vector $x\in\im\tilde{L}_d^{up}$, following the steps of Lemma~\ref{lem:effective_resistance_flows}, we conclude that 
    $$
        x^{T}(\tilde{L}_d^{up})^{+}x = \min\{ \|y\|^{2} : D^{-1/2}\boundary_{d+1}y = x \}
    $$ 
    As $D^{1/2}$ is a bijection, then we further conclude that 
    \begin{align*}
        x^{T}(\tilde{L}_d^{up})^{+}x =& \min\{ \|y\|^{2} : D^{-1/2}\boundary_{d+1}y = x \} \\
        =& \min\{ \|y\|^{2} : \boundary_{d+1}y = D^{1/2}x \} \\
        =& R_{D^{1/2}x}(\K) 
    \end{align*}
     We have the bound $\RE_{D^{1/2}x}(\K) = \|D^{1/2}x\|^{2}\RE_{\overline{x}}(\K) \geq d_{\min}\RE_{\overline{x}}(\K)$. As the chain $D^{1/2}x\in\im\boundary_{d+1}$ if and only if $x\in\im D^{-1/2}\boundary_{d+1}$, we have the general bound:
    \begin{align*}
        \lambda_{\max}((\tilde{L}_d^{up})^{+}) =& \max\{x^{T}(\tilde{L}_d^{up})^{+}x : \|x\|=1,\,x\in\im D^{-1/2}\boundary_{d+1} \} \\
        =& \max\{ R_{D^{1/2}x}(\K): \|x\|=1,\,x\in\im D^{-1/2}\boundary_{d+1} \} \\
        \geq & d_{\min} \max\{ R_{x}(\K): \|\overline{x}\|=1,\,\overline{x}\in\im\boundary_{d+1} \} \\
        =& d_{\min} \lambda_{\max}((L_d^{up})^{+}) \tag{\Cref{lem:spectral_gap_effective_resistance}}
    \end{align*}
    To prove $\lambda_{\max}((\tilde{L}_d^{up})^{+}) \leq d_{\max} \lambda_{\max}(L_d^{up})^{+})$, we can similarly prove that 
    $$
    \lambda_{\max}((L_d^{up})^{+}) = \{(D^{-1/2}x)(\tilde{L}_d^{up})^{+}(D^{-1/2}x) : \|x\|=1,\, x\in\im\boundary_{d+1}\},$$
    and the rest of the proof follows similarly.
\end{proof}

\section{Properties of Effective Resistance: Parallel, Series, and Monotonicity Formulas.}
\label{sec:formulas}

 We now prove there are formulas for effective resistance in simplicial complexes analogous to the series and parallel formulas for effective resistance in graphs. These formulas not only are useful for calculating effective resistance, but they also provide intuition for effective resistance. In particular, they provide justification for the claim that effective resistance measures ``how null-homologous'' a cycle is in a complex.
\begin{theorem}[Series Formula]
\label{thm:resistance_complex_series}
  Let $\K_1$ and $\K_2$ be simplicial complexes with $\gamma\in C_{d-1}(\K_1)\cap C_{d-1}(\K_2)$, $C_d(\K_1)\cap C_d(\K_2)=\emptyset$, and $\gamma$ null-homologous in $\K_1$ and $\K_2$. Let $\K=\K_1\cup\K_2$. Then
  $$
    \mathcal{R}_\gamma(\K) \leq \mathcal{R}_{\gamma_1}(\K_1) + \mathcal{R}_{\gamma_2}(\K_2).
  $$
  Equality is achieved when $\gamma_1$ and $\gamma_2$ are the unique chains in $C_{d}(\K_1)$ and $C_d(\K_2)$ that sum to $\gamma$.
\end{theorem}
\begin{proof}
  Let $\gamma_1$ and $\gamma_2$ be null-homologous cycles in $\K_1$ and $\K_2$ respectively that sum to $\gamma$, and let $f_1$ and $f_2$ be the minimum-energy unit $\gamma_1$- and $\gamma_2$-flows, respectively. Then $f=f_1+f_2$ is a unit $\gamma$-flow, and we can bound $\mathcal{R}_\gamma(\K) \leq \mathsf{J}(f) = \mathsf{J}(f_1)+\mathsf{J}(f_2) = \mathcal{R}_{\gamma_1}(\K_1) + \mathcal{R}_{\gamma_2}(\K_2)$; the equality $\mathsf{J}(f) = \mathsf{J}(f_1)+\mathsf{J}(f_2)$ follows from the fact that $\K_1$ and $\K_2$ have disjoint sets of $d$-simplices.
  \par
  To prove the other direction, observe that $\gamma$ can always be written as the sum of two null-homologous chains $\gamma_1 \in C_{d-1}(\K_1)$ and $\gamma_2 \in C_{d-1}(\K_2)$. Any unit $\gamma$-flow $g$ defines null-homologous $(d-1)$-cycles $\gamma_1$ and $\gamma_2$ that sum to $\gamma$; namely, if $g_1$ and $g_2$ are the restriction of $g$ to $\K_1$ and $\K_2$ respectively, then $\gamma_1 = \boundary g_1$ and $\gamma_2 = \boundary g_2$.
  \par
   If $\gamma$ can be uniquely decomposed as $\gamma = \gamma_1 + \gamma_2$, then any unit $\gamma$-flow $f$ can be decomposed as a unit $\gamma_1$-flow $f_1$ and a unit $\gamma_2$-flow $f_2$. It follows that the energy of $f$ is minimized when the energy of $f_1$ and $f_2$ are both minimized. Hence, $\mathcal{R}_\gamma(\K) = \mathcal{R}_\gamma(\K_1) + \mathcal{R}_\gamma(\K_2)$.
\end{proof}

\begin{theorem}[Parallel Formula]
\label{thm:resistance_complex_parallel}
  Let $\K_1$ and $\K_2$ be simplicial complexes with $\gamma\in C_{d-1}(\K_1)\cap C_{d-1}(\K_2)$, $C_d(\K_1)\cap C_d(\K_2)=\emptyset$, and $\gamma$ null-homologous in $\K_1$ and $\K_2$. Let $\K=\K_1\cup\K_2$. Then
  $$
    \mathcal{R}_{\gamma}(\K) \leq \left( \frac{1}{\mathcal{R}_{\gamma}(\K_1)} + \frac{1}{\mathcal{R}_{\gamma}(\K_2)} \right)^{-1}
  $$
  Equality is achieved when $\im\boundary_{d}[\K_1]\cap\im\boundary_{d}[\K_2]=\spn\{\gamma\}$.
\end{theorem}
\begin{proof}
  Let $f_1$ and $f_2$ be the minimum energy unit $\gamma$-flows in $\K_1$ and $\K_2$ resp. For any $t\in\R$, the chain $g_t=tf_1+(1-t)f_2$ is a unit $\gamma$-flow in $\K$. We can therefore bound the effective resistance over the minimum of these combinations as $\mathcal{R}_\gamma(\K)\leq \min_{t\in\R} \mathsf{J}(g_t)$.
  \par
   To get the tighest bound of $\RE_\gamma(\K)$, we now derive $t_\opt := \arg\min_{t\in\R} \mathsf{J}(g_t)$. Observe that $\mathsf{J}(g_t)=t^2 \mathsf{J}(f_1)+(1-t)^2 \mathsf{J}(f_2)=t^2 \mathcal{R}_\gamma(\K_1) + (1-t)^2 \mathcal{R}_\gamma(\K_2)$; this follows from the fact that $\K_1$ and $\K_2$ have disjoint sets of $d$-simplices. The quantity $\mathsf{J}(g_t)$ is a positive quadratic with respect to $t$, so $t_\opt$ is the value of $t$ where the derivative of $\mathsf{J}(g_t)$ is 0. Taking the derivative, we find that $t_\opt = \mathcal{R}_\gamma(\K_2)/(\mathcal{R}_\gamma(\K_1)+\mathcal{R}_\gamma(\K_2))$. Plugging $t_\opt$ into $\mathsf{J}(g_{t_\opt})$, we find that
  \begin{align*}
    \mathsf{J}(g_{t_\opt}) = & \left( \frac{\mathcal{R}_\gamma(\K_2)}{\mathcal{R}_\gamma(\K_1) + \mathcal{R}_\gamma(\K_2)} \right)^2 \mathcal{R}_\gamma(\K_1) + \left( 1 - \frac{\mathcal{R}_\gamma(\K_2)}{\mathcal{R}_\gamma(\K_1) + \mathcal{R}_\gamma(\K_2)} \right)^2 \mathcal{R}_\gamma(\K_2) \\
    = & \left( \frac{\mathcal{R}_\gamma(\K_2)}{\mathcal{R}_\gamma(\K_1) + \mathcal{R}_\gamma(\K_2)} \right)^2 \mathcal{R}_\gamma(\K_1) + \left( \frac{\mathcal{R}_\gamma(\K_1)}{\mathcal{R}_\gamma(\K_1) + \mathcal{R}_\gamma(\K_2)} \right)^2 \mathcal{R}_\gamma(\K_2) \\
    = & \left( \frac{1}{\mathcal{R}_\gamma(\K_1) + \mathcal{R}_\gamma(\K_2)} \right)^2 \Big(\mathcal{R}_\gamma(\K_1) + \mathcal{R}_\gamma(\K_2) \Big) \mathcal{R}_\gamma(\K_1) \mathcal{R}_\gamma(\K_2) \\
    = & \phantom{\Bigg(} \frac{\mathcal{R}_\gamma(\K_1) \mathcal{R}_\gamma(\K_2)}{\mathcal{R}_\gamma(\K_1) + \mathcal{R}_\gamma(\K_2)} \\
    = & \left( \frac{1}{\mathcal{R}_{\gamma}(\K_1)} + \frac{1}{\mathcal{R}_{\gamma}(\K_2)} \right)^{-1}
  \end{align*}
  This implies the upper bound on $\mathcal{R}_\gamma(\K)$ in the theorem statement. To get the lower bound, observe that any unit $\gamma$-flow $g$ in $\K$ can be orthogonally decomposed into chains $g_1\in C_d(\K_1)$ and $g_2\in C_d(\K_2)$, so $\boundary g_1+\boundary g_2=\gamma$. We claim that $\boundary g_1=t\gamma$ and $\boundary g_2=(1-t)\gamma$ for some value of $t$; if not, then $\boundary g_1=t\gamma+\eta$ and $\boundary g_2=(1-t)\gamma-\eta$ for some non-zero chain $\eta\not\in\spn\{\gamma\}$, which cannot be the case as $\im\boundary_{\K_1}\cap\im\boundary_{\K_2}=\spn\{\gamma\}$. This proves the chain $g$ is a linear combination of a unit $\gamma$-flow in $\K_1$ and a unit $\gamma$-flow in $\K_2$. The chain $g_t$ is the lowest energy such linear combination.
\end{proof}

\begin{figure}
    \centering
    \begin{subfigure}{0.49\textwidth}
        \centering
        \vspace{0.25in}
        \includegraphics[height=1in]{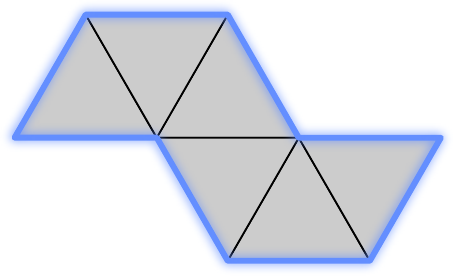}
        \vspace{0.25in}
    \end{subfigure}
    \begin{subfigure}{0.49\textwidth}
        \centering
        \includegraphics[height=1.5in]{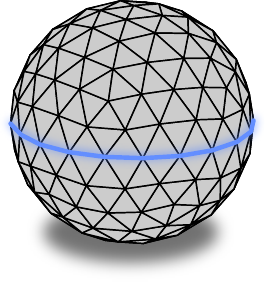}
    \end{subfigure}
    \caption{Left: The cycle $\gamma$ is the blue boundary with $\pm 1$ coefficients, and the unique unit $\gamma$-flow is the 6 triangles in series. If the complex is unweighted, then the effective resistance of $\gamma$ is 6. Right: The cycle $\gamma$ is the equator of the sphere, and the two hemispheres are two unit $\gamma$-flows in parallel. If each hemisphere has flow energy $c$, then the effective resistance of $\gamma$ is $\frac{c}{2}$.}
    \label{fig:series_parallel}
\end{figure}

Figure \ref{fig:series_parallel} shows examples of unit $\gamma$-flows in  series and parallel. These formulas justify the claim that the effective resistance of a null-homologous cycle $\gamma$ is a measure of \textit{how} null-homologous $\gamma$ is. The more chains with boundary $\gamma$, the smaller the effective resistance of $\gamma$ by the parallel formula. The smaller the chains bounding $\gamma$, the lower the effective resistance by the series formula.
\par
Another important property of effective resistance in graphs is \textit{\textbf{Rayleigh monotonicity}}. Rayleigh monotonicity says that adding edges to the graph can only decrease the effective resistance between any pair of vertices; this reinforces the notion that effective resistance measures how well-connected a pair of vertices are, as adding an edge can only make a pair of vertices better connected. We prove a similar result for simplicial complexes.

\begin{theorem}[Rayleigh Monotonicity]
\label{thm:resistance_monoticity}
  Let $\L\subset \K$ be simplicial complexes. Let $\gamma\in C_{d-1}(\K) \cap C_{d-1}(\L)$ be a cycle that is null-homologous in both complexes. Then $\RE_\gamma(\K) \leq \RE_\gamma(\L)$.
\end{theorem}
\begin{proof}
  As $C_d(\L) \subset C_d(\K)$, then any unit $\gamma$-flow in $\L$ is also a unit $\gamma$-flow in $\K$. As the effective resistance is the minimum energy of a unit $\gamma$-flow, then clearly $\RE_\gamma(\K) \leq \RE_\gamma(\L)$.
\end{proof}

\section{Duality of Resistance and Capacitance in Embedded Complexes.}
\label{sec:duality}

In this section, we consider the effective capacitance of a cycle in the special case when $\K$ is a $d$-dimensional simplicial complex with a given embedding into $\R^{d+1}$. In this case, $\K$ is called an \EMPH{embedded complex}.
 Embedded complexes serve as a high-dimensional generalization of planar graphs and naturally admit a dual graph.
 We will show that the effective capacitance of certain $(d-1)$-cycles $\gamma$ in $\K$ are equal to the effective resistance between a pair of vertices in the dual graph that are \say{dual} to $\gamma$. This theorem generalizes the analysis of capacitance in planar graphs given by Jeffery and Kimmel~\cite{Jeffery2017}.
Hence, we can parameterize the quantum algorithm deciding if $\gamma$ is null-homologous (\Cref{thm:querycomplexity}) in terms of the effective resistance of $\gamma$ in $\K$ and the effective resistance between the pair of vertices in the dual graph.
 Specifically, we will generalize the special case of planar graphs for which the vertices $s$ and $t$ appear on the boundary of the same face.
Throughout this section we assume we are given the embedding as input. Computing the dual graph from an embedding can be done in polynomial time~\cite{Dey2020}. 

\subsection{Duality in Embedded Complexes.} The \textit{\textbf{Alexander Duality theorem}}~\cite[Corollary 3.45]{Hatcher} states that for a $d$-dimensional simplicial complex $\K$ with an embedding into $\R^{d+1}$, the complement $\R^{d+1} \setminus \K$ consists of $\beta_d[\K] + 1$ connected components.
We call these connected components \EMPH{voids}. Exactly one of these voids is unbounded.
We denote the bounded voids as $V_i$ for $1 \leq i \leq \beta_d$ and the unbounded void as $V_\infty$. Moreover, the boundaries of the bounded voids generate the homology group $H_d(\K)$.
The embedding implies that each $d$-simplex is contained on the boundary of at most two voids, and we make the assumption that the $d$-simplices are oriented consistently with respect to the voids.
That is, if a $d$-simplex is on the boundary of two voids it is oriented positively on one void, and negatively on the other.
We have a boundary matrix $\partial_{d+1}$ whose columns are the voids and whose rows are the $d$-simplices.
From the embedding and the consistent orientation we see that $\partial_{d+1}$ is the edge-vertex incident matrix of the \EMPH{directed dual graph}: the directed graph whose vertices are in bijection with the voids and whose edges are in bijection with the $d$-simplices of $\K$. The direction of the edges are inherited from the orientations of the $d$-simplices.
For a $d$-simplex $\sigma$ on the boundary of voids $V_1$ and $V_2$ we denote the dual edge by $\sigma^* = \{v_1^*, v_2^*\}$ and we define the dual weight function by $w^*(\sigma^*) = 1/w(\sigma)$.
\par 
We construct an additional chain group $C_{d+1}(\K)$ with the bounded voids as basis elements. This is a purely algebraic construction, meaning basis elements of $C_{d+1}(\K)$ do \textit{not} correspond to $(d+1)$-simplices. We can also define a boundary map $\boundary_{d+1}:C_{d+1}(\K)\to C_{d}(\K)$, where $\boundary V_i$ is the boundary of the void $V_i$ as described above. This gives rise to a new chain complex 
\[\cdots 
    0 \rightarrow C_{d+1}(\K) \xrightarrow{\partial_{d+1}} C_d(\K) \xrightarrow{\partial_{d}} \cdots \xrightarrow{\partial_1} C_0(\K).
\]
Since the boundaries of the voids generate the $d^{\text{th}}$ homology group of $\K$ and $C_{d+1}(\K)$ is generated by these voids we obtain a valid chain complex. Moreover, we have that $\dim H_d(\K) = 0$ in our new chain complex.
\par 
In addition to a dual graph, we define the \EMPH{dual complex} of $\K$, denoted $\K^*$, by defining the dual chain groups $C_k(\K^{*})$ via the isomorphism $C_{d - k + 1}(\K) \cong C_k(\K)$. Importantly, note that only $C_0(\K^{*})$ and $C_1(\K^{*})$ correspond to the chain groups of a simplicial complex (the dual graph), while the higher chain groups $C_{i}(\K)$ for $i\geq 2$ are a purely algebraic construction. Moreover, we define the dual boundary operator $\partial^*_k \colon C_k(\K^*) \rightarrow C_{k-1}(\K^*)$ to be the coboundary operator $\delta_{d - k + 1} \colon C_{d-k + 1}(\K) \rightarrow C_{d-k}(\K)$ of $\K$, and the dual coboundary operator $\delta^*_k \colon C_{k-1}(\K^*) \rightarrow C_{k}(\K^*)$ to be the boundary operator $\partial_{d - k + 1} \colon C_{d-k+1}(\K) \rightarrow C_{d-k}(\K)$ of $\K$. In other words the (co)boundary operators commute with the duality isomorphism. We summarize the construction in the following commutative diagram.

\begin{equation*}
    \begin{tikzcd}
    C_{d+1}(\K)
        \arrow[d, leftrightarrow, "\cong"]
        \arrow[r, "\partial_{d+1}"]
  & C_{d}(\K)
        \arrow[d, leftrightarrow, "\cong"]
        \arrow[r, "\partial_{d}"]
        \arrow[l, "\delta_{d}"]
  & \cdots
        \arrow[r, "\partial_1"]
        \arrow[l, "\delta_{d-1}"]
  & C_0(\K)
        \arrow[l, "\delta_1"]
        \arrow[d, leftrightarrow, "\cong"] \\
    C_0(\K^*)
        \arrow[r, "\delta^*_0"]
  & C_1(\K^*)
        \arrow[l, "\partial^*_1"]
        \arrow[r, "\delta^*_1"]
  & \cdots
        \arrow[l, "\partial^*_2"]
        \arrow[r, "\delta^*_{d}"]
  & C_{d+1}(\K^*)
        \arrow[l, "\partial_{d+1}^*"]
    \end{tikzcd}
\end{equation*}

We need to make one additional assumption on the location of the input $(d-1)$-dimensional cycle $\gamma$ which makes our setup a generalization of a planar graph with two vertices $s$ and $t$ appearing on the same face.
We assume that there exists a void $V_i$ with two unit $\gamma$-flows $\Gamma_1$ and $\Gamma_2$ such that $\supp(\Gamma_1) \cap \supp(\Gamma_2) = \emptyset$ and $\supp(\Gamma_1) \cup \supp(\Gamma_2) = \supp(\partial_{d+1} V_i)$.
That is, there exist two unit $\gamma$-flows whose supports partition the boundary of the void $V_i$. For example, the support of the cycle $\gamma$ could be the equator of a sphere, and the support of the cycles $\Gamma_1$ and $\Gamma_2$ could be the north and south hemispheres.
This generalizes the fact in planar graphs that when $s$ and $t$ are on the same face we can find two $st$-paths which partition the boundary of the face.
In planar graphs we are guaranteed to find two such paths, however for an arbitrary $(d-1)$-cycle $\gamma$ we are not guaranteed to find two unit $\gamma$-flows partitioning the boundary of some void.
More specifically, we take $\Gamma_2$ to be a unit $(-\gamma)$-flow so that $\partial_d \Gamma_2 = - \gamma$. In the planar graph analogy this is equivalent as viewing $\Gamma_1$ as a path from $s$ to $t$ and viewing $\Gamma_2$ as a path from $t$ to $s$.
We add an additional basis element $\Sigma$ to $C_d(\K)$ such that $\partial_d \Sigma = - \gamma$.
In planar graphs this is equivalent to adding an edge directed from $t$ to $s$.
The addition of this edge splits the face containing $s$ and $t$ into two.
In higher dimensions, the geometry of adding another $d$-simplex to fill $\gamma$ is more complicated, but the addition of $\Sigma$ to $C_d(\K)$ allows us to perform a purely algebraic operation to our chain complex that behaves as if $V_i$ has been split into two; namely, we remove $V_i$ from $C_{d+1}(\K)$ and replace it with two new basis elements $V_s$ and $V_t$.
Next, we extend the boundary operator to $V_s$ and $V_t$ in the following way: $\partial_{d+1} V_s = \Gamma_1 - \Sigma$ and $\partial_{d+1} V_t = \Gamma_2 + \Sigma$. In the dual complex, the vertices dual to $V_s$ and $V_t$ are denoted $s^*$ and $t^*$, and the edge dual to $\Sigma$ is denoted $\Sigma^* = \{t^*, s^*\}$.

\subsection{Effective capacitance is Dual to Effective Resistance.}

In this next section, we will show that the effective capacitance of $\gamma$ in a subcomplex $\L\subset\K$ is equal to the effective resistance between $s^*$ and $t^*$ in the dual complex $\L^{*}\subset\K^{*}$, where $\L^{*}$ is defined as being the subgraph of $\K^{*}$ with all of the vertices of $\K^*$ but only the edges dual to the $d$-simplices \textbf{not} in $\L$.
\par 
The effective resistance between $s^*$ and $t^*$ in $\L^{*}$ is determined by the unit $s^*t^*$-flows in $\L^{*}$. However, it will be convenient to work with \EMPH{circulations} instead of flows. A unit $s^*t^*$-circulation $f$ is a 1-cycle such that $f(\Sigma^*) = 1$. Recall that $\Sigma^*$ is the edge directed from $t^*$ to $s^*$, so a unit $s^*t^*$-circulation is just a unit $s^*t^*$-flow with additional flow on the edge $\Sigma^*$ to the cycle.
Clearly, there is a bijection between unit $s^*t^*$-flows and unit $s^*t^*$-circulations. We define the flow energy of a circulation to be equal to the flow energy of its corresponding flow.

\begin{theorem}
Let $\K$ be a $d$-dimensional simplicial complex embedded into $\R^{d+1}$, and let $\L\subset\K$ be a subcomplex. Let $\gamma\in C_{d-1}(\K)$ be a $(d-1)$-cycle such that there exist two unit $\gamma$-flows $\Gamma_1,\,\Gamma_2\in C_{d}(\K)$ with supports that partition the boundary of some void $V_i$. Then the effective capacitance $\C_\gamma(\L,\K)$ is equal to the effective resistance $\mathcal{R}_{s^*t^*}(\L^{*})$.
\end{theorem}
\begin{proof}

Let $p$ be a unit $\gamma$-potential in $\L$, and define $f\in C_1(\L^{*})$ to be the image of $\delta_{d-1} p$ under the duality isomorphism; that is, $f = \partial_2^*p^*$. The 1-chain $f$ is a circulation in the 1-skeleton of $\L^{*}$. Further, since $\Sigma^* = \{t^*,s^*\}$ the circulation $ f$ corresponds to a unit $s^*t^*$-flow by the following calculation: 
\[
    f^{T}\Sigma^* = p^{T}\partial_d \Sigma = p^{T}\gamma = 1.
\] 
Next, we calculate the flow energy of $f$ and show it is equal to the potential energy of $p$.
\begin{align*}
\mathsf{J}(f) &= \sum_{e^* \in \L^*_1} \frac{f(e^*)^2}{w^*(e^*)} \\
&= \sum_{e^* \in \L^*_1} f(\sigma^*)^2 w(e)\\
&= \sum_{\sigma \in \K_d\setminus \L_d} \delta p (\sigma)^2 w(\sigma)\\
&= \mathcal{J}(p)
\end{align*}

Conversely, let $f^*$ be the minimal-energy unit $s^*t^*$-circulation in $\L^{*}$.
By the assumptions outlined in the previous section, we have $\dim H_d(\K) = 0$, which in turn gives us $\dim H_1(\K^*) = 0$.
Hence, $f^* \in \im \partial_2^*$.
Let $p^*\in C_2(\L^{*})$ with $\partial^*_2 p^* = f^*$; we will show that $p$ is the unit $\gamma$-potential in $\L$ in bijection with $f^*$.
To see that $p$ is a unit $\gamma$-potential we compute its value on $\gamma$:
\[ 
p^{T}\gamma = p^{T}\partial_d \Sigma = (\delta_{d-1}p)^{T}\Sigma = (\partial^*_2p^*)^{T}\Sigma^* = (f^*)^{T}\Sigma^* = 1.
\]
Moreover, as $\boundary_2^{*}p^{*}=f^{*}$ and $f^{*}$ has its support on $\L_1^{*}$ (which is defined to only have edges dual to the $d$-simplices in $\K\setminus\L$), then $\boundary_{d-1}[\L]p=0$, so $p$ is a unit $\gamma$-potential in $\L$.
\par 
It remains to show that the potential energy of $p$ is equal to the flow energy of $f^*$. We have the following calculation:
\begin{align*}
\mathcal{J}(p) &= \sum_{\sigma \in \K_d\setminus\L_d} \delta_{d-1}p(\sigma)^2 w(\sigma)\\
&= \sum_{\sigma \in \K_d\setminus\L_d} \frac{\delta_{d-1}p(\sigma)^2}{w^*(\sigma^*)}\\
&= \sum_{\sigma^{*} \in \L_1^{*}} \frac{\partial^*_2 p^*(\sigma^*)^2}{w^*(\sigma^*)}\\
&= \sum_{\sigma^{*} \in \L_1^{*}} \frac{f^*(\sigma^{*})^2}{w^*(\sigma^*)}\\
&= \mathsf{J}(f^*).
\end{align*}
\end{proof}

\section{Evaluating the span program for null-homology testing.}
\label{sec:evaluate}
\newcommand{\RCperp}{R_{C^\perp}}
\newcommand{\RU}{R_{U^{-}}}
\newcommand{\Rker}{R_{\ker\boundary}}

In this section, we give a quantum algorithm for evaluating the null-homology span program. Our algorithm is inspired by and generalizes the quantum algorithm for evaluating $st$-connectivity span program in graphs. The first quantum algorithm for evaluating the $st$-connectivity span program was given by Belovsz and Reichardt~\cite{Belovs2012}; however, we follow the slightly different algorithm introduced by Ito and Jeffery~\cite{Ito2018}. We are also greatly indebted to the presentation of this algorithm given by Jeffery and Kimmel~\cite{Jeffery2017}, from which our algorithm is adapted.
\par
 The algorithm for evaluating a general span program $\P=(\H,\U,\ket{\tau},A)$ is to perform phase estimation of the vector $\ket{w_0}:=A^+\ket{\tau}$ on the unitary operator $U=R_{\H(x)}R_{\ker A}$ where the notation $R_S$ denotes the reflection about the subspace $S$.
(The unitary $R_S = 2\Pi_S - I$, where $\Pi_S$ is the projection onto $S$.)
Intuitively, if $x$ is a positive instance, then $\ket{w_0}$ will be close to an eigenvector of $U$ with phase $0$. If $x$ is a negative instance, then $\ket{w_0}$ will be far from any eigenvector of $U$ of phase 0. If we want to evaluate the function $f:D\to\{0,1\}$, we need to perform phase estimation to precision $O\left(1/\sqrt{W_-(f,\P)W_+(f,\P)}\right)$. The algorithm for phase estimation of a unitary $U$ to precision $O(\delta)$ performs $O(1/\delta)$ implementations of the unitary $U$ \cite{Kitaev1995}, so the algorithm for evaluating the span program $\P=(\H, \U, \ket{\tau}, A)$ requires $O\left(\sqrt{W_-(f,\P)W_+(f,\P)}\right)$ implementations of $U$. We now analyze the time complexity of implementing the unitary $U$.
\par
The reflection $R_{\H(X)}$ can be implemented with two queries to $\oracle_x$. 
This reflection is the same as the reflection across the good states in Grover's Algorithm. The rest of this section is devoted to an implementation of $R_{\ker\boundary}$.
\par
 Recall that $\ker\boundary_d\subset C_d(\K)$. The idea behind the implementation of $R_{\ker\boundary}$ is that instead of reflecting across $\ker\boundary_d$ directly, we can embed $C_d(\K)$ into $C_{d-1}(\K)\tensor C_{d}(\K)$ by sending $\ket{\tau}\to c\ket{\boundary\tau}\ket{\tau}$ (where $c$ is a normalization constant). We can then implement the reflection $R_{\ker\boundary}$ by implementing a series of ``local reflections" on the basis $\ket{\boundary\tau}\ket{\tau}$.
\par
We consider two subspaces $B$ and $C$ of $C_{d-1}(K)\tensor C_{d}(K)$. The spaces $B$ and $C$ are defined:
$$
  B = \spn \left\{ \ket{b_\tau}:=\frac{1}{\sqrt{d+1}}\ket{\boundary\tau}\ket{\tau} : \tau\in \K_{d} \right\}
$$
and
$$
  C = \spn\left\{ \ket{c_\sigma}:=\sum_{\sigma\subset\tau}\sqrt{\frac{w(\sigma)}{\deg(\sigma)}}\ket{\sigma}\ket{\tau} : \sigma\in \K_{d-1} \right\}.
$$

The space $C_{d-1}(\K)\tensor C_d(\K)$ has basis $\{\ket\sigma\ket\tau \mid \sigma\in K_{d-1},\, \tau\in\K_d\}$. The vector $\ket{b_\tau}$ is non-zero on a basis element $\ket\sigma\ket\tau$ if and only if $\sigma$ is in the support of the boundary of $\tau$. Similarly, a component of $\ket{c_\sigma}$ is non-zero on $\ket\sigma\ket\tau$ if and only if $\tau$ is in the support of the coboundary of $\sigma$. The vector $\ket{b_\tau}$ can be thought of as being \textit{like} the boundary of $\tau$, with the additional property that the set $\{\ket{b_\tau}\mid\tau\in C_d(\K)\}$ is orthonormal. Similarly, the vector $\ket{c_\sigma}$ is \textit{like} the coboundary of $\sigma$ but orthonormal.
\par
We also define operators that embed $C_{d}(\K)$ and $C_{d-1}(\K)$ into $B$ and $C$ respectively. We define linear operators $M_B \colon C_{d}(\K) \rightarrow B$ and $M_C \colon C_{d-1}(\K) \rightarrow C$ as follows:
  $$
    M_B:=\sum_{\tau\in \K_d}\ket{b_\tau}\bra{\tau},
  $$
  and
  $$
    M_C:=\sum_{\sigma\in \K_{d-1}}\ket{c_\sigma}\bra{\sigma}.
  $$
 As the columns of $M_B$ and $M_C$ are orthonormal, both operators are isometries.
 \par
 We introduce the matrices $M_C$ and $M_B$ as they have the property that $\ker M_C^\dagger M_B=\ker\boundary$, which we prove in the follow lemma. This fact will give us a way to implement $R_{\ker\boundary}$.
 \begin{lemma}
  \label{lem:ker_mcmb}
  $\ker M_C^\dagger M_B = \ker\boundary$.
 \end{lemma}
 \begin{proof}
   We first calculate the matrix $M_C^{\dagger}M_B$. We then argue that $\ker M_C^\dagger M_B=\ker\boundary$. For a $(d{-}1)$-simplex $\sigma$ and a $d$-simplex $\tau$, we have that
     $$
       \braket{c_\sigma}{b_\tau}=\sum_{\tau'\in\K_d : \sigma\subset\tau'} \frac{w(\tau)}{\sqrt{\deg(\sigma)}}\braket{\sigma}{\boundary\tau}\braket{\tau'}{\tau}=
       \begin{cases}
         \sqrt{\frac{w(\tau)}{(d+1)\deg(\sigma)}} \braket{\sigma}{\boundary\tau} & \text{if $\sigma\subset\tau$} \\
         0 & \text{otherwise}
       \end{cases}.
     $$
     So $\braket{c_\sigma}{b_\tau}$ is non-zero if and only if $\sigma$ is in the boundary of $\tau.$ We use this to calculate the product $M_C^\dagger M_B:$
     \begin{align*}
       M_C^\dagger M_B =& \sum_{\sigma\in \K_{d-1}} \sum_{\tau\in \K_d} \ket{\sigma}\braket{c_\sigma}{b_\tau}\bra{\tau}\\
       =&\frac{1}{\sqrt{(d+1)}} \sum_{\sigma\subset\tau} \frac{\sqrt{w(\tau)}\braket{\sigma}{\boundary\tau}}{\sqrt{\deg(\sigma)}} \ket{\sigma}\bra{\tau}\\
       =&\frac{1}{\sqrt{(d+1)}} \left(\sum_{\sigma\in \K_{d-1}(\K)} \frac{\ket{\sigma}\bra{\sigma}}{\sqrt{\deg(\sigma)}}\right) \sum_{\tau\in \K_d}\sqrt{w(\tau)}\ket{\boundary\tau}\bra{\tau} \\
       =&\frac{1}{\sqrt{(d+1)}} \left(\sum_{\sigma\in \K_{d-1}(\K)} \frac{\ket{\sigma}\bra{\sigma}}{\sqrt{\deg(\sigma)}}\right) \boundary\sqrt{W} =:\hat{\boundary}.
     \end{align*}
   The term $\frac{\ket{\sigma}\bra{\sigma}}{\sqrt{\deg(\sigma)}}$ is all-zeros matrix except for the $(\sigma, \sigma)$-entry, which is $\frac{1}{\sqrt{\deg(\sigma)}}$. The sum $\sum_{\sigma\in \K_{d-1}} \frac{\ket{\sigma}\bra{\sigma}}{\sqrt{\deg(\sigma)}}$ is a diagonal matrix.
   Accordingly, the matrix $\hat{\boundary}$ is $\boundary\sqrt{W}$ with each row scaled.
   Scaling the rows of a matrix does not change its row space or kernel, so $\ker M_C^{\dagger}M_B=\ker\boundary$.
 \end{proof}
 The spaces $B$ and $C$ and the matrices $M_B$ and $M_C$ are inspired by the follow lemma of Szegedy which is necessary for implementing $R_{\ker \partial}$.
 \begin{lemma}[Szegedy \cite{SzegedyMarkov}, Theorem 1]
    \label{lem:spectral_gap}
    Let $M_B$ and $M_C$ be matrices with the same number of rows and orthonormal columns, and let $B=\spn M_B$ and $C=\spn M_C$. The matrix $M_C^\dagger M_B$ has singular values at most 1. Let $\cos\theta_1,\ldots,\cos\theta_k$ be the singular values of $M_C^\dagger M_B$ in the range $(0,1)$. Let $U=R_CR_B$. We can decompose the eigenspaces of $U$ as
    \begin{itemize}
      \item The (+1)-eigenspace of $U$ is $(B\cap C)\oplus(B^\perp \cap C^\perp)$.
      \item The (-1)-eigenspace of $U$ is $(B\cap C^\perp)\oplus(B^\perp\cap C)$.
      \item The remaining eigenvalues of $U$ are $e^{\pm 2 i\theta_j}$ for $1\leq j\leq k$.
    \end{itemize}
  \end{lemma}
The following lemma gives us a way to implment the $R_{\ker\boundary}$.
 Let $\RU$ be the rotation about $(-1)$-eigenspace of $U$, and let $V=M_B^\dagger \RU M_B$. The matrix $V$ embeds $C_d(\K)$ into $B$ with $M_B$, performs a reflection on $B$ about the $(-1)$-eigenspace of $U$, and unembeds with $M_B^\dagger$. The following lemma proves that $V=R_{\ker\boundary}$.
 \begin{lemma}
 \label{lem:V_Rker}
  The matrix $V= M_B^\dagger R_{U^-} M_B$ satisfies the equality $V = R_{\ker\boundary}$.
 \end{lemma}
 \begin{proof}
   We first verify that $V$ is a reflection; that is, we show the eigenvalues  of $V$ are 1 and -1. The matrices $M_B$ and $M_C$ have orthonormal columns, so we can use Lemma \ref{lem:spectral_gap} to characterize the eigenspaces of $U$. The $(-1)$-eigenspace of $U$ is $(B\cap C^{\perp})\oplus(B^{\perp}\cap C)$ and the $(+1)$-eigenspace of $U$ is $(B\cap C)\cap (B^\perp\cap C^\perp)$. As the spaces $(B\cap C)$ and $(B\cap C^\perp)$ span $B$, then $R_{U^{-}}$ restricted to $B$ has eigenvalues 1 and -1. As $B=\im M_B$ and $V=M_B^\dagger \RU M_B$, then we conclude that $V$ has eigenvalues 1 and -1 as well.
   \par
   Now that we have determined that $V$ is a reflection, we need to determine which subspace $V$ reflects across. A corollary of the previous paragraph is that a vector $\ket{\psi}\in C_d(\K)$ is in the $(+1)$-eigenspace of $V$ if and only if $M_B\ket{\psi}$ is in the $(-1)$-eigenspace of $U$. Specifically, a vector $\ket{\psi}$ is in the $(+1)$-eigenspace of $V$ if and only if $M_B\ket{\psi}\in C^{\perp}$. As $C^{\perp}=\ker M_C^{\dagger}$, the vector $\ket{\psi}$ is in the $(+1)$-eigenspace of $V$ if and only if $\ket{\psi}\in\ker{M_C^{\dagger}M_B}$. We proved in Lemma \ref{lem:ker_mcmb} that $\ker{M_C^{\dagger}M_B}=\ker\boundary$, so we conclude that $V=\Rker$
 \end{proof}

  We have a matrix $V$ that implements $\Rker$; next, we analyze the complexity of implementing $V$. We start by analyzing the complexity of implementing $\RU$, the reflection across the $(-1)$-eigenspace of $U$.
  \par
  We implement the reflection around the $(-1)$-eigenspace of $U$ using phase estimation, an algorithm introduced by Magniez et. al.~\cite{Magniez2011}. The algorithm is as follows. We first estimate the phase of $U$ to some degree of accuracy to be specified shortly. Intuitively, we need to estimate the phase of $U$ to high enough accuracy to distinguish between $-1$ eigenvalues of $U$ and eigenvalues of $U$ close to -1. We then perform a reflection controlled on the estimated phase.
  \par
   The \textit{\textbf{phase gap}} of a unitary $U$ with eigenvalues $\{e^{i\theta_1},\ldots, e^{i\theta_k}\}$ is $\min\{|\theta_i|:\theta_i\neq 0\}$. The following lemma shows that the phase gap determines the complexity of reflecting across the 1-eigenspace of $U$.

   \begin{lemma}[Magniez et. al.~\cite{Magniez2011}, Paraphrase of Theorem 6]
   \label{lem:Magniez_reflection}
    Let $U$ be a unitary with phase gap $\theta$. A reflection around the 1-eigenspace of $U$ can be performed to constant precision with $O\left( \frac{1}{\theta}\right)$ applications of $U$.
   \end{lemma}
   The phase gap measures gap between the 1-eigenspace of a unitary and all other eigenvalues. We are interested in the gap in phase between the $(-1)$-eigenspace of $U$ and the other eigenvalues of $U$. This is precisely the phase gap of $-U$. The following lemma analyzes the phase gap of $-U$ and gives the complexity of reflecting about the $(-1)$-eigenspace of $U$.

  \begin{lemma}
  \label{lem:phase_gap_U}
    We can implement $\RU$ with $O\left( \sqrt{\frac{d+1}{\tilde{\lambda}_{\min}}} \right)$ calls to $U$, where $\tilde{\lambda}_{\min}$ is the smallest non-zero eigenvalue of the normalized up-Laplacian of $\K$.
  \end{lemma}
  \begin{proof}
    We need to calculate the phase gap of $-U$ to determine the precision to which we need to estimate the phase of $U$. Observe that if $\theta$ is the phase of an eigenvalue of $U$, then $\theta+\pi$ is the phase of an eigenvalue of $-U$. We can bound the phase gap of $-U$ using Lemma \ref{lem:spectral_gap}. The non-zero eigenvalues of $U$ are $\{e^{\pm i2 \theta_j}\}_j$, where $\{\cos\theta_j\}_j$ were the singular values of $M_C^\dagger M_B$. Therefore, the phases of $-U$ are
    $\{\pm |\pi - 2\theta_j|\}_j$. Using the inequality that $\pi/2 - \theta_j\geq \cos\theta_j$ for $\theta_j\in[0,\pi/2]$, then the phase gap of $-U$ is bounded below by
    $$
      |\pi - 2\theta_j| \geq 2\cos\theta_j \geq 2\cdot\sigma_{\min}(M_C^\dagger M_B)
    $$
    where $\sigma_{\min}(M_C^\dagger M_B)$ is the smallest singular value of $M_C^\dagger M_B$.
    \par
    We can actually relate the smallest singular value of $M_C^\dagger M_B$ to something more meaningful. By the proof of Lemma \ref{lem:ker_mcmb}, the matrix $M_C^\dagger M_B=\frac{1}{\sqrt{d+1}} D^{-1/2}\boundary\sqrt{W}$, where $D$ is the diagonal matrix with the degrees of the $(d{-}1)$-simplices on the diagonal. Thus, $(M_C^\dagger M_B)(M_C^\dagger M_B)^\dagger = \frac{1}{d+1} D^{-1/2}\boundary W \coboundary D^{-1/2} = \frac{1}{d+1} D^{-1/2} L D^{-1/2}.$ Recall from Section \ref{sec:prelim} that the matrix $ D^{-1/2} L D^{-1/2}$ is the normalied up-Laplacian. The singular values of a matrix $A$ are the square roots of the eigenvalues of $AA^T$. Thus, the smallest singular value of $M_C^\dagger M_B$, and the phase gap of $-U$, is $\Omega\left( \sqrt{\frac{\tilde{\lambda}_{\min}}{d+1}} \right)$, where $\tilde{\lambda}_{\min}$ is the smallest eigenvalue of $D^{-1/2}LD^{-1/2}$. Therefore, by Lemma \ref{lem:Magniez_reflection}, we can implement $\RU$ with $O\left( \sqrt{\frac{d+1}{\tilde{\lambda}_{\min}}} \right)$ calls to $U$.
  \end{proof}

  We are almost ready to give the running time for $V=\Rker$, but first, we need to make a delicate distinction. The matrices $M_B$ and $M_C$ have orthonormal columns, but they are \textit{not} unitary. We can see this as $\ker M_B^\dagger\neq 0$ and $\ker M_C^\dagger\neq 0$. As $M_B$ and $M_C$ are not unitary, they cannot be implemented on a quantum computer. Fortunately, it suffices to implement unitaries $U_B$ and $U_C$ such that $U_B|_{C_d(\K)}=M_B$ and $U_C|_{C_{d-1}(\K)}=M_C$. With this in mind, we can give the running time for $V=\Rker$.

 \begin{lemma}
  There is an algorithm to perform $R_{\ker\boundary}$ in time $\tilde{O}\left( \sqrt{\frac{d+1}{\tilde{\lambda}_{\min}}} (T_B+T_C) \right)$, where $T_B$ and $T_C$ are the times to perform the unitaries $U_B$ and $U_C$ respectively.
 \end{lemma}
 \begin{proof}
    Lemma \ref{lem:V_Rker} shows that $R_{\ker\boundary}=V=M_B^\dagger \RU M_B$. We can equivalently run $U_B^\dagger\RU U_B$. As $U_B$ takes $T_B$ by definition, we only need to show we can implement $\RU$ in $\Tilde{O}( \sqrt{(d+1)/\tilde{\lambda}_{\min}}\, (T_B+T_C) )$ time. Lemma \ref{lem:phase_gap_U} shows we can implement $\RU$ with $O(\sqrt{(d+1)/\tilde{\lambda}_{\min}})$ calls to $U$, so we need to show we can implment $U$ in $\Tilde{O}(T_B+T_C)$. The unitary $U=R_C R_B$, and we claim we can implement $R_C$ and $R_B$ in $\Tilde{O}(T_B)$ and $\Tilde{O}(T_C)$ respectively. We can implement $R_B$ as $U_B R_{\K_d} U_B^\dagger$, where $R_{\K_d}$ reflects across the basis states $\{\ket{0}\ket{\sigma}\mid\sigma\in\K_d\}$. We can check if a quantum state is of the form $\ket{0}\ket{\sigma}$ in $O(\log n_d)$ gates (specifically, by checking if the basis state is within a certain range), so the unitary $R_{\K_d}$ takes $O(\log n_d)$ gates, and $R_B$ takes $\Tilde{O}(T_B)$ time. The unitary $R_C$ takes $\Tilde{O}(T_C)$ time by the same argument.
 \end{proof}

The running time $T_B$ is dependent on how the boundary maps are loaded into the quantum algorithm. We propose a method of storing the boundary maps in a quantum computer called the \textit{\textbf{incidence array}}. The incidence array is adapted from the \textit{adjacency array} introduced by Durr et al.~\cite{Drr2006} to store the adjacency between pairs of vertices in a graph.
\par
For a $d$-simplex $\tau=\{v_0,\ldots,v_{d}\}$, the \textit{\textbf{down-incidence array}} is the function $g:\ket{\tau}\ket{j}\ket{0}\to\ket{\tau}\ket{j}\ket{\tau\setminus\{v_j\}}$ for $0\leq j\leq d$. The simplices in the boundary of $\tau$ have alternating sign. To address this, we also perform a negation conditioned on the parity of $\ket{j}$ to compute $(-1)^j\ket{\tau}\ket{j}\ket{\tau\setminus\{v_j\}}$.
\par
Durr et al.~\cite{Drr2006} claim that queries to the incidence array can be performed in logarithmic time. As the down-incidence array is identical to the adjacency array\footnotemark, queries to the down-incidence also take logarithmic time. We can compute the state $\ket{\boundary\tau}\ket{\tau}$ with the down-incidence array and the following lemma.
\footnotetext{The down-incidence array is actually an adjacency array of a graph related to simplicial complexes, namely, the incidence graph between the $(d{-}1)$- and $d$-simplices.}

\begin{lemma}[Cade, Montanaro, Belovs \cite{cade}, Implicit in the proof of Lemma 2]
\label{lem:cade_lemma}
  Let $f:[m]\to[k]$ be a function, and let $\oracle_f$ be an oracle that computes $\oracle_f:\ket{i}\ket{0}\to\ket{i}\ket{f(i)}$. The state $\frac{1}{\sqrt{m}}\sum_{i=1}^{m}\ket{f(i)}$ can be computed with $O(\sqrt{m})$ queries to $\oracle_f$ and $O(\polylog(m))$ additional gates.
\end{lemma}

\begin{corollary}
  The unitary $U_B$ can be implemented in $O(\sqrt{d})$ queries to the up-incidence array and $\Tilde{O}(\sqrt{d})$ time.
\end{corollary}

It is harder to produce a generic implementation of $U_C$ than $U_B$. The $d$-simplices can have arbitrary weights, so constructing the states $\ket{c_\sigma}$ in general requires constructing arbitrary quantum states with real coefficients. However, the weights on the simplices do not affect whether or not a cycle is null-homologous. Therefore, we can always run our null-homology test on the unweighted complex; the trade-off is that the effective resistance or effective capacitance might be higher in the unweighted complex. We analyze the running time of the unweighted case in Section \ref{sec:special_cases}.

We now analyze the complexity of constructing the initial state $\ket{w_0}/\|\ket{w_0}\|$.  To construct the dummy state, we start by adding an additional ``$d$-cell'' $\ket{\emptyset}$ to the complex with boundary $\ket{\gamma}$ (really, we just add $\ket{\gamma}$ as a column to $\boundary$.) The new cell will have non-trivial overlap with $\ket{w_0}$, so we can construct $\ket{w_0}$ by amplifying this component of $\ket{\emptyset}$. We outline this method in the proof of Theorem \ref{thm:init}, but first, we state Lemma~\ref{lem:general_parallel} which is a generalization of the parallel formula for effective resistance; its proof is nearly identical to the proof of Theorem \ref{thm:resistance_complex_parallel}.

\begin{lemma}
\label{lem:general_parallel}
    Let $V=V_1 \oplus V_2$ be a vector space. Let $A:V\to U$ be a linear map, and let $A_1:U_1\to V$ and $A_2:U_2\to V$ be the restriction of $A$ to $U_1$ and $U_2$. Let $\ket{t}\in \im A_1\cap \im A_2\subset U$. If $\ket{s}=A^+\ket{t}$, $\ket{s_1}=A_1^+\ket{t}$, and $\ket{s_2}=A_2^+\ket{t}$, then
    $$
      \|s\|^2 \leq \left( \frac{1}{\|s_1\|^2} + \frac{1}{\|s_2\|^2} \right)^{-1}
    $$
    Equality is achieved when $\im A_1 \cap \im A_2 = \spn\{\ket{t}\}$. In this case, $\ket{s} = t\ket{s_1} + (1-t)\ket{s_2}$ where $t=\|s_2\|^2/(\|s_1\|^2+\|s_2\|^2)$.
\end{lemma}

\begin{theorem}
\label{thm:init}
  Let $\oracle_\gamma$ be the oracle that takes $\oracle_\gamma:\ket{0}\to\ket{\gamma}$. Let $T_\gamma$ be the time it takes to implement $\oracle_\gamma$. The state $\ket{w_0}=\boundary^{+}\ket{\gamma}$ can be created in $\Tilde{O}((\sqrt{1/\mathcal{R}_\gamma(\K)} + \sqrt{\mathcal{R}_\gamma(\K)})(T_B+T_C+T_\gamma))$ time.
\end{theorem}
\begin{proof}
   We append $\ket{\gamma}$ as a column to $\boundary$ to create a new matrix $\hat{\boundary}$. Let $\ket{\emptyset}$ be index of the new column, so $\hat{\boundary} = \boundary + \ket{\gamma}\bra{\emptyset}$. Let $\ket{w_0'}=\hat{\boundary}^+\ket{\gamma}$. We conclude that $\ket{\emptyset}=\ket{w_0'}+\ket{w_0'^\perp}$ where $\ket{w_0'^\perp}\in\ker\hat{\boundary}$, as the projection $\Pi_{\ker\hat{\boundary}^\perp}\ket{\emptyset}=\hat{\boundary}^+ \hat{\boundary}\ket{\emptyset}=\hat{\boundary}^+ \ket{\gamma}=\ket{\hat{w_0}}$.
   \par
   We construct $\ket{w_0}/\|\ket{w_0}\|$ in two steps. First, we use amplitude amplification to amplify the $\ket{w_0'}$ component of $\ket{\emptyset}$. We then use a second amplitude amplification to amplify the $\ket{w_0}$ component of $\ket{w_0'}$. These amplitude amplifications are nested, as we need to perform the first to create the initial state for the second.
   \par
    If we perform constant time phase estimation of $\ket{\emptyset}$ on the unitary $R_{\ker\hat{\boundary}}$, then we can map $\ket{\emptyset}$ to $\ket{0}\ket{w_0'}+\ket{1}\ket{w_0^\perp}$.
    We can then amplify the amplitude of $\ket{0}\ket{w_0}$ part arbitrarily close to $\ket{w_0}/\|\ket{w_0}\|$ using $O(\|\ket{w_0}\|^{-1})$ calls to $R_{\ker\hat{\boundary}}$.
   \par
   We calculate $\|\ket{w_0}\|$ using the formula from the lemma. The vector $\ket{\emptyset}$ has length 1, so Lemma \ref{lem:general_parallel} shows that
   $$
    \|\ket{w_0'}\|^2 = \left( \frac{1}{\mathcal{R}_\gamma(\K)} + 1 \right)^{-1} = \frac{\mathcal{R}_\gamma(\K)}{\mathcal{R}_\gamma(\K)+1}.
   $$
   Thus, we need to perform the reflection $R_{\ker\hat{\boundary}}$ a total of $O(\|\ket{w_0}\|^{-1})=O(\sqrt{(\mathcal{R}_\gamma(\K)+1)/\mathcal{R}_\gamma(\K)})$ times to create $\ket{\hat{w_0}}/\|\ket{\hat{w_0}}\|$.
   \par
   The next step in our algorithm is to amplify the $\ket{w_0}/\|\ket{w_0}\|$ component of $\ket{w_0'}/\|\ket{w_0'}\|$. By Lemma \ref{lem:general_parallel}, the state $\|\ket{w_0'}\|=t\|\ket{w_0}\|+(1-t)\|\ket{\emptyset}\|$ for $t=1/(\mathcal{R}_\gamma(\K)+1)$. Therefore, the $\ket{w_0}/\|\ket{w_0}\|$ component of $\ket{w_0'}/\|\ket{w_0'}\|$ has norm
   \begin{align*}
   t\|\ket{w_0}\|/\|\ket{w_0'}\| &= \frac{1}{\mathcal{R}_\gamma(\K)+1}\sqrt{\mathcal{R}_\gamma(\K)}\sqrt{\frac{\mathcal{R}_\gamma(\K)+1}{\mathcal{R}_\gamma(\K)}} \\
   &= \sqrt{\frac{1}{\mathcal{R}_\gamma(\K)+1}}
   \end{align*}
    To return the state $\ket{w_0}/\|\ket{w_0}\|$, we need to perform amplitude amplification again. We can create the state $\ket{w_0'}$ using the amplitude amplification from the previous two paragraphs with $O(\sqrt{(\mathcal{R}_\gamma(\K)+1)/\mathcal{R}_\gamma(\K)})$ applications of $R_{\ker\boundary'}$, and we can reflect across $\ket{\emptyset}$  in constant time as it is a basis state. To create $\ket{w_0}/\|\ket{w_0}\|$, we need
    $$
      O\left(\sqrt{(\mathcal{R}_\gamma(\K)+1)/\mathcal{R}_\gamma(\K)}\sqrt{(\mathcal{R}_\gamma(\K)+1)}\right) = O\left(\sqrt{\mathcal{R}_\gamma(\K)}+\sqrt{\frac{1}{\mathcal{R}_\gamma(\K)}}\right)
    $$
    applications of $R_{\ker\hat\boundary}$.
  \par
  We now argue that we can compute $R_{\ker\hat{\boundary}}$ in $O(T_C+T_B+T_\gamma)$ time. As was the case with $R_{\ker\boundary}$, we decompose $R_{\ker\hat{\boundary}}=M_{\hat{B}}^\dagger R_{\hat{U}^-} M_{\hat{B}}$ for space $\hat{B}$ and $\hat{C}$ defined
  $$
    \hat{B} = B \cup \{\ket{b_\emptyset}:=\ket{\gamma}\}
  $$

  $$
    \hat{C} = \spn\left\{ \ket{c_\sigma}:=\frac{1}{\sqrt{\deg(\sigma)+1}}\ket{\sigma}\ket{\emptyset}+\sum_{\sigma\subset\tau}\frac{w(\sigma)}{\sqrt{\deg(\sigma)+1}}\ket{\sigma}\ket{\tau} : \sigma\in \K_{d-1} \right\}.
  $$
  The unitaries $M_{\hat{B}}$ $M_{\hat{C}}$, and $\R_{\hat{U}^-}$ are defined analogously to $M_B$ and $M_C$. We can implement the unitary version of these matrices $U_{\hat{B}}$ in $O(T_B+T_\gamma)$ and $U_{\hat{C}}$ in $O(T_C)$.
\end{proof}

We now summarize this section in the following theorem.

\begin{theorem}
\label{thm:running_time_w_oracles}
  Let $\K$ be a simplicial complex, $\gamma\in C_{d-1}(\K)$ a null-homologous cycle, and $\K(x) \subset \K$ be a simplicial complex. There is a quantum algorithm for deciding if $\gamma$ is null-homologous in $\K(x)$ that runs in time
  $$
    \Tilde{O}\left(\sqrt{\frac{(d+1)\RE_{\max}(\gamma)\C_{\max}(\gamma)}{\tilde{\lambda}_{\min}}}(\sqrt{d} + T_C) + \left(\sqrt{\frac{1}{\RE_{\gamma}(\K)}} + \sqrt{\RE_{\gamma}(\K)}\right)(\sqrt{d}+T_C+T_\gamma) \right)
  $$
  where $\RE_{\max}$ is the maximum effective resistance of $\gamma$ in any subcomplex $\K(y)$ and $\C_{\max}(\gamma)$ is the maximum effective capacitance $\gamma$ in any subcomplex $\K(y)$ and $\tilde{\lambda}_{\min}$ is the smallest eigenvalue of the normalized $d$ up-Laplacian.
\end{theorem}

\subsection{Special cases.}
\label{sec:special_cases}

We now consider a few special cases of the null-homology span program. These special cases will allow us to replace the terms $T_B$ and $T_\gamma$ in Theorem \ref{thm:running_time_w_oracles} with concrete running times.

\paragraph{Unweighted simplicial complexes}
We now consider the case where there are no weights on the $d$-simplices, or equivalently, when $w(\tau)=1$ for each $d$-simplex $\tau$. While computing $M_C$ is hard in general, in the unweighted case, we can implement the unitary $M_C$ using a straightforward oracle. For a $(d{-}1)$-simplex $\sigma$ that is incident to the $d$-simplices $\{ \tau_1,\ldots, \tau_m \}$, the \EMPH{up-incidence array} is the oracle is the function that maps $h:\ket{\tau}\ket{j}\ket{0} \to \ket{\tau}\ket{j}\ket{0}$. By Lemma \ref{lem:cade_lemma}, the up-incidence array can be used to compute $\ket{c_\tau} = \frac{1}{\sqrt{m}}\ket{\tau}\ket{\coboundary\tau}$ in $O(\sqrt{m})$ time. The unitary $M_C$ computes the state $\ket{c_\sigma}$ in parallel, so computing $M_C$ will take $\sqrt{d_{\max}}$ queries, where $d_{\max} = \max_{\tau\in C_{d-1}(\K)}\deg(\tau)$ time. This is summarized in the following lemma.

\begin{lemma}
\label{lem:coboundary_gate_unweighted}
  If $\K$ is an unweighted simplicial complex, the unitary $U_C$ can be implemented in $O(\sqrt{d_{\max}} )$ queries to the up incidence array and $T_C=\Tilde{O}(\sqrt{d_{\max}}) = \Tilde{O}(\sqrt{n_0})$ time.
\end{lemma}

Additionally, if $\K$ is an unweighted complex, we can upper bound the quantity $\frac{1}{\mathcal{R}_\gamma(\mathcal{K})}$. 

\begin{lemma}
\label{lem:effective_resistance_lower_bound}
    Let $\K$ be an unweighted simplicial complex with $n_{0}$ vertices, and let $\gamma$ be a unit-length null-homologous $(d-1)$-cycle in $\K$. Then $\frac{1}{\RE_{\gamma}(\K)}\leq n_{0}$.
\end{lemma}
\begin{proof}
    Recall that in an unweighted simplicial complex that $\RE_{\gamma}(\K) = \gamma^{T}(L^{up}_{d-1})^{+}\gamma$. As the eigenvalues of $L^{up}_{d-1}$ are bounded above by $n_0$~(\Cref{thm:upper_bound_maximum_eigenvalue_laplacian}), then the non-zero eigenvalues of $(L^{up}_{d-1})^{+}$ are bounded below by $\frac{1}{n_0}$. As $\gamma\in\im L_{d-1}^{up}$, then $\gamma^{T}(L^{up}_{d-1})^{+}\gamma$ is bounded below by the smallest non-zero eigenvalue of $(L^{up}_{d-1})^{+}$, i.e $\frac{1}{n_0}$. Therefore, $\frac{1}{\RE_{\gamma}(\K)}\leq n_{0}$
\end{proof}

\paragraph{Cycle is the boundary of a $d$-simplex.}

We now consider the case that the input cycle $\gamma$ is the boundary of a $d$-simplex. In this case, we can implement the oracle $\oracle_\gamma$ with the down incidence array used to implement $M_B$. We get the same running time for $T_\gamma$ as $T_B$.

\begin{lemma}
\label{lem:cycle_gate_boundary_d_simplex}
  If $\gamma$ is the boundary of a $d$-simplex, there is a quantum algorithm implementing $\oracle_\gamma$ in $O(\sqrt{d})$ queries to the down incidence array and $T_{\gamma} \in \Tilde{O}(\sqrt{d}) \subset \Tilde{O}(\sqrt{n_0})$ time.
\end{lemma}

\paragraph{Summing Up.}

If we combine our bounds for the case where our complex is unweighted (\Cref{lem:coboundary_gate_unweighted} and \Cref{lem:effective_resistance_lower_bound}) and the cycle is the boundary of $d$-simplex (\Cref{lem:cycle_gate_boundary_d_simplex}) with the bound of the time complexity (\Cref{thm:running_time_w_oracles}), we get the following bound on the time complexity.

\thmtimecomplexity*

\subsection{Space Complexity.}

We now comment on the space complexity of our algorithm. The inner product space $C_{d-1}(\K)\oplus C_{d}(\K)$ has dimension $n_{d-1}n_d\in O(\binom{n_0}{d}\binom{n_0}{d+1})$, so vectors in this space can be represented with $O(\log(\binom{n_0}{d}\binom{n_0}{d+1}))=O(dn)$ qubits. Additionally, the phase estimation step from \Cref{lem:phase_gap_U} takes $O( \log(1/\tilde{\lambda}_{\min}))$ ancillary qubits, which is $O(dn_{d})$ qubits by the bounds on the normalized spectral gap~(\Cref{lem:normalized_vs_unnormalized_spectral_gap} and \Cref{thm:spectral_gap_lower_bound}). Finally, the outer phase estimation of $R_{\H(x)}R_{\ker A}$ takes $O(\log\left(\RE_{\max}(\gamma)\C_{\max}(\gamma)\right))$ qubits. Other gates require polylogarithmic ancillary qubits. The space complexity of our algorithm is comparable to some other QTDA algorithms, as other QTDA algorithms require $O\left(\log(1/\lambda_{\min})\right)$ to perform phase estimation of the combinatorial Laplacian~\cite{gunn2019}.

\section{Construction of the Building Block.}
\label{sec:construction_building_block}

\begin{figure}[ht]
    \centering
    \begin{subfigure}{0.3\textwidth}
        \centering
        \vspace{0.25in}
        \includegraphics[height=1in]{Figures/stellar_triangle.pdf}
        \vspace{0.25in}
    \end{subfigure}
    \begin{subfigure}{0.3\textwidth}
        \centering
        \includegraphics[height=1.5in]{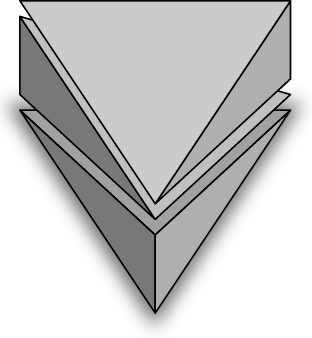}
    \end{subfigure}
    \begin{subfigure}{0.3\textwidth}
        \centering
        \includegraphics[height=1.5in]{Figures/stellar_prism_top_shadow.pdf}
    \end{subfigure}
    \caption{Left to Right: Stellar subdivision, prism, and stellar prism of a triangle.}
    \label{fig:different_constructions}
\end{figure}

In this section, we formally describe the ``building block'' $B_d$ that we use to construct the simplicial complexes $\mathcal{B}_d^{n}$, $\mathcal{P}_d^{n}$, and $\mathcal{Q}_d^{n}$ with exponentially-large effective resistance and capacitance in Section \ref{sec:lower_bounds}. A more intuitive, but informal, description of the building block can be found in Section \ref{sec:lower_bounds}
\par
The building block $B_d$ is constructed using two constructions from algebraic topology: the stellar subdivision and the prism. We will combine these two constructions to make a new construction we call the stellar prism. The building block $B_d$ is then a quotient of the stellar prism. This section will first describe the stellar subdivision, prism, and stellar prism. We then present the construction of $B_d$ and prove some of its relevant properties.

\subsection{Stellar Subdivision.} Let $K$ be a $d$-dimensional simplicial complex. The \textit{\textbf{stellar subdivision}} of $K$ is the $d$-dimensional simplicial complex $\star{K}$ that is the union of the $(d-1)$-skeleton of $K$ and, for each $d$-simplex $\sigma\in K_d$, the set of simplices $\{\tau\cup \{v_{\sigma}\} : \tau\subsetneq\sigma \}$, where $v_\sigma$ is a new vertex. See Figure \ref{fig:different_constructions}. The property of the stellar subdivision that is key to our construction is that it increases the number of $d$-simplices in the simplicial complex. 

\begin{observation}
    Let $K$ be a simplicial complex with $n_{d}$ $d$-simplices. Then $\star{K}$ has $d\cdot n_d$ $d$-simplices. 
\end{observation}

We can map chains from our original complex $K$ to chains in its stellar subdivision. For each $d$-simplex $\sigma\in K_d$ and any integer $k$, we define a map $b_{\sigma}:C_{k}(K)\to C_{k+1}(\star{K})$. Let $\tau\subsetneq\sigma$, and suppose that $v_{\sigma}$ is the $(i+1)$st element of $\tau\cup \{v_{\sigma}\}$ with respect to the ordering on the vertices\footnotemark; that is, if $\tau=\{v_{0},\ldots,v_k\}$, then $\tau\cup\{v_{\sigma}\}=\{v_{0},\ldots,v_{i-1},v_{\sigma},v_{i},\ldots,v_{k}\}$. We define $b_\sigma(\tau) = (-1)^{i} \cdot (\tau\cup\{v_{\sigma}\})$ for all $k$-simplices of the form $\tau\subsetneq\sigma$, and $b_\sigma(\varphi)=0$ for any other $k$-simplex $\varphi$.  

\footnotetext{It is worth noting that it is arbitrary where $v_\sigma$ is in the ordering of the vertices, so $i$ can be any number in the range $0\leq i\leq k$. Lemmas \ref{lem:chain_join} and \ref{lem:stellar_map} summarize the relevant properties of the stellar subdivision for our paper, and as we will see, these lemmas will hold wherever $v_\sigma$ is in the vertex order, so long as the map $b_\sigma$ is defined appropriately.}

\begin{lemma}
\label{lem:chain_join}
    Consider $b_{\sigma}:C_{k}(K)\to C_{k+1}(\star{K})$ as defined above. The map $b_{\sigma}$ satisfies $\boundary\circ b_{\sigma} = \id - b_{\sigma}\circ\boundary$.
\end{lemma}
\begin{proof}
    As $b_\sigma$ is a linear map, it suffices to prove this for the basis $k$-chains $\tau\in C_{k}(K)$. The boundary of $b_{\sigma}(\tau)$ is
    \begin{align*}
        \boundary(b_\sigma(\tau)) =& (-1)^{i}\cdot \boundary(\tau\cup\{v_{\sigma}\}) \\
        =& (-1)^{i}\left[ \sum_{j=0}^{i-1} (-1)^{j}\cdot(\tau\cup\{v_\sigma\}\setminus\{v_j\}) + (-1)^{i}\cdot\tau + \sum_{j=i}^{k} (-1)^{j+1}\cdot(\tau\cup\{v_\sigma\}\setminus\{v_{j}\}) \right]\\ 
        =& \tau + (-1)^{i}\left[ \sum_{j=0}^{i-1} (-1)^{j}\cdot(\tau\cup\{v_\sigma\}\setminus\{v_j\}) + \sum_{j=i}^{k} (-1)^{j+1}\cdot(\tau\cup\{v_\sigma\}\setminus\{v_{j}\}) \right]
    \end{align*}
    We claim that the second term in this sum equals $(-1)\cdot b_{\sigma}(\boundary\tau)$. Observe that for $j<i$, $v_\sigma$ will be the $i$th vertex of $\tau\cup\{v_{\sigma}\}\setminus\{v_j\}$. Likewise, for $i\geq j$, $v_\sigma$ will be the $(i+1)$st vertex of $\tau\cup\{v_{\sigma}\}\setminus\{v_j\}$. Therefore, we have that  
    \begin{align*}
        & (-1)^{i}\Big[ \sum_{j=0}^{i-1} (-1)^{j}\cdot(\tau\cup\{v_\sigma\}\setminus\{v_j\}) + \sum_{j=i}^{k} (-1)^{j+1}\cdot(\tau\cup\{v_\sigma\}\setminus\{v_{j}\}) \big] \\
        =& (-1)^{i}\Big[ \sum_{j=0}^{i-1} (-1)^{j}(-1)^{i-1} b_{\sigma}(\tau\setminus\{v_j\}) +\sum_{j=i}^{k} (-1)^{j+1}(-1)^{i} b_{\sigma}(\tau\setminus\{v_{j}\}) \big] \\
        =& (-1)\Big[ \sum_{j=0}^{k} (-1)^{j} b_{\sigma}(\tau\setminus\{v_j\}) \big] = (-1)\cdot b_{\sigma}(\boundary\tau)
    \end{align*}
\end{proof}

Using the maps $b_\sigma$, we now define another map $S_{*}:C_{k}(K)\to C_{k}(\star{K})$. For $k<d$, we define the map $S_{*}:C_{k}(K)\to C_{k}(\star{K})$ just to be the inclusion map. For $k=d$, define $S:C_{d}(K)\to C_{d}(\star{K})$ on each $d$-simplex $\sigma = \{v_0,\ldots,v_d\}$ as $S_{*}(\sigma) = b_{\sigma}(\boundary\sigma)$. 

\begin{lemma}
\label{lem:stellar_map}
    Consider $S_{*}:C_{k}(K)\to C_{k}(\star{K})$ as defined above. The map $S_{*}$ satisfies $\boundary\circ S_{*} = S_{*}\circ\boundary$. 
\end{lemma}
\begin{proof}
    For $i<d$, this is obvious as $S_{*}$ is the inclusion map. For $i=d$, this follows from Lemma \ref{lem:chain_join} as $\boundary S_{*}(\sigma) = \boundary b_{\sigma}(\boundary\sigma) = \boundary\sigma - b_{\sigma}\boundary\boundary(\sigma)= \boundary\sigma = S_{*}\boundary\sigma$.
\end{proof}

\subsection{Prisms.} The \textit{\textbf{prism}} of a simplicial complex $K$ is a triangulation of the space $K\times[0,1]$. To define the prism of $K$, we will define the prism of a single simplex $\sigma\in K$;  the prism of the entire complex $K$ is then the union of the prism of its simplices, i.e. $\prism{K} = \cup_{\sigma\in K} P{\sigma}$. See Figure \ref{fig:different_constructions}.
\par 
Assume that $\sigma = \{v_0,\ldots,v_d\}$. The \textit{\textbf{prism}} of $\sigma$ is a $(d+1)$-dimensional simplicial complex $\prism{\sigma}$ with vertices $\sigma\times\{0,1\}$. The prism contains all simplices of the form $\tau\times\{0\}$ and $\tau\times\{1\}$ for $\tau\subset\sigma$. The remaining simplices of $P{\sigma}$ are the closure of the $(d+1)$-simplices $\{(v_0,0),\ldots,(v_j,0),(v_j,1),\ldots,(v_d,1)\}$ for $0 \leq j\leq d$. We denote such a simplex 
$$
    \sigma^{j} := \{(v_0,0),\ldots,(v_j,0),(v_j,1),\ldots,(v_d,1)\}.
$$
Note that for a paired of nested simplices $\tau\subset\sigma$, their prisms are also nested, i.e. $\prism{\tau}\subset\prism{\sigma}$. 
\par 
As was the case for stellar subdivision, we will define several maps between chains in $K$ and chains in $\prism{K}$. We abuse notation and define $K\times\{1\} := \{\sigma\times\{1\} : \sigma\in K\}$. The first map $I_{1}:C_{k}(K)\to C_{k}(\prism{K})$ maps chains in $K$ to chains in $K\times\{1\}$. Specifically, for a $d$-simplex $\sigma\in K_d$, we define $I_{1}(\sigma) = \sigma\times\{1\}$. We define a map $I_{0}$ analogously. The following lemma is obvious.

\begin{lemma}
\label{lem:prism_inclusion_maps}
    Consider $I_{i}:C_{k}(K)\to C_{k}(\prism{K})$ as defined above for $i=0,1$. The map $I_{i}$ satisfy $\boundary\circ I_{i} = I_{i}\circ\boundary$ for $i=0,1$.
\end{lemma}

We now define a map $P_{*}:C_{k}(K)\to C_{k+1}(\prism{K})$. Specifically, for a $k$-simplex $\sigma\in K$ with $\sigma=\{v_0,\ldots,v_k\}$, the corresponding $(k+1)$-chain is defined $P_{*}(\sigma) = \sum_{i=0}^{k} (-1)^{i}\sigma^{i}$. 

\begin{lemma}
\label{lem:prism_map}
    Consider $P_{*}:C_{k}(K)\to C_{k+1}(\prism{K})$ as defined above. The map $P_{*}$ satisfies $P_{*}\circ\boundary + \boundary\circ P_{*} = I_{1} - I_{0}$.
\end{lemma}
\begin{proof}
    The key points of this proof are essentially identical to those provided by Hatcher \cite{Hatcher} in the proof of Theorem 2.10.
\end{proof}

\subsection{Stellar Prisms.} Now we propose a way of combining stellar subdivisions and prisms that we call the \textit{\textbf{stellar prism}}. Intuitively, the stellar prism is a triangulation of the space $K\times[0,1]$ where the bottom copy $K\times\{0\}$ is triangulated the same way as $K$ and the top copy $K\times\{1\}$ is triangulated using the stellar subdivision. See Figure \ref{fig:different_constructions}.
\par 
We first define the stellar prism of a $d$-simplex; the stellar prism of a $d$-dimensional simplicial complex $K$ is the union of the prism of the $(d-1)$-skeleton and the stellar prisms of the $d$-simplices, i.e. $\sp{K} = \prism{K^{d-1}}\cup_{\sigma\in K_d} \sp{\sigma}$. 
\par 
Let $\sigma$ be a $d$-simplex. The stellar prism of $\sigma$ is the $(d+1)$-dimensional simplicial complex $\sp{\sigma}$ described as follows. The vertices of $\sp{\sigma}$ are $\left(\sigma\times\{0,1\}\right)\cup\{v_\sigma\}$, where $v_\sigma$ is a new vertex. For any simplex $\tau\subsetneq\sigma$, $\sp{\sigma}$ contains both the complex $P\tau$ and all simplices of the form $\{\varphi\cup\{v_\sigma\} : \varphi\in P\tau\}$. Additionally, $\sp{\sigma}$ contains the $d$-simplex $\sigma\times\{0\}$ and the $(d+1)$ simplex $(\sigma\times\{0\})\cup\{v_\sigma\}$. If we let $v_{\sigma\times\{1\}} = v_{\sigma}$, note that $SP\sigma$ contains the subdivision $S(\sigma\times\{1\})$. 
\par 
We now define a linear map between $SP_{*}:C_i(K)\to C_{i+1}(\sp{K})$. For $i<d$, we simply define $SP_{*} = P_{*}$. For $i=d$, we define $SP_{*}(\sigma) = -b_{\sigma} I_{0}(\sigma)  - b_{\sigma} P_{*}\boundary(\sigma)$

\begin{lemma} 
\label{lem:stellar_prism_map}
    Consider $SP_{*}:C_i(K)\to C_{i+1}(\sp{K})$ as defined above. The map $SP_{*}$ satisfies $\boundary SP_{*} + SP_{*}\boundary = S_{*}I_{1} - I_{0}$
\end{lemma}
\begin{proof}
    For $i<d$, this follows from Lemma \ref{lem:prism_map} and the fact that $SP_{*} = P_{*}$ and $S_{*} = \id$. We now verify this for $i=d$.
    \par 
    We will analyse the boundary of $\boundary SP_{*}(\sigma)$. We find that
    \begin{equation}
    \label{eqn:boundary_stellar_prism}
        \boundary SP_{*}(\sigma) = -\boundary b_{\sigma} I_{0}(\sigma) - \boundary b_{\sigma} P_{*}\boundary(\sigma).
    \end{equation}
    We analyse the two terms in this sum. Using Lemma \ref{lem:chain_join} and \ref{lem:prism_inclusion_maps}, the first term of Equation (\ref{eqn:boundary_stellar_prism}) evaluates to 
    \begin{align*}
        \boundary b_{\sigma} I_{0}(\sigma) &= I_{0}(\sigma) - b_{\sigma}\boundary I_{0}(\sigma) \\
        &= I_{0}(\sigma) - b_{\sigma}I_{0}\boundary(\sigma) 
    \end{align*}
    Using Lemma \ref{lem:prism_map}, the second term of Equation (\ref{eqn:boundary_stellar_prism}) evaluates to 
    \begin{align*}
        \boundary b_{\sigma} P_{*}\boundary(\sigma) =& P_{*}\boundary(\sigma) - b_{\sigma}\boundary P_{*}\boundary(\sigma) \\
        =& P_{*}\boundary(\sigma) - b_{\sigma}(I_{1} - I_{0} -P_{*}\boundary)\boundary(\sigma)\\
       =& P_{*}\boundary(\sigma) - b_{\sigma}I_{1}\boundary(\sigma) + b_{\sigma}I_{0}\boundary(\sigma) + b_{\sigma}P_{*}\boundary\boundary(\sigma)\\
       =& P_{*}\boundary(\sigma) - b_{\sigma}I_{1}\boundary(\sigma) + b_{\sigma}I_{0}\boundary(\sigma)
    \end{align*}
    The first term in this sum is $P_{*}\boundary(\sigma) = SP_{*}\boundary(\sigma)$ as $P_{*} = SP_{*}$ for all dimensions less than $d$. The term $b_{\sigma}I_{1}\boundary(\sigma) = b_{\sigma}\boundary I_{1}(\sigma) = S_{*} I_{1}(\sigma)$. Combining the two terms of Equation (\ref{eqn:boundary_stellar_prism}), we find that 
    \begin{align*}
        \boundary SP_{*}(\sigma) =& -\boundary b_{\sigma} I_{0}(\sigma) - \boundary b_{\sigma} P_{*}\boundary(\sigma) \\
        =& -(I_{0}(\sigma) - b_{\sigma}I_{0}\boundary(\sigma)) - (P_{*}\boundary(\sigma) - S_{*} I_{1}(\sigma) + b_{\sigma}I_{0}\boundary(\sigma)) \\
        =& -I_{0}(\sigma) - P_{*}\boundary(\sigma) + S_{*} I_{1}(\sigma) 
    \end{align*}
\end{proof}

\subsection{Building Block.}

We now describe the building block $B_d$. Let $\Delta^{d}$ be the closure of the $d$-simplex $\sigma=\{v_0,\ldots,v_d\}$, and let $\boundary\Delta^{d}$ denote the $(d-1)$-dimensional simplicial complex $\Delta^{d}\setminus\{\sigma\}$. The building block $B_d$ is derived from the stellar prism $\sp{(\boundary\Delta^{d})}$. Denote the $(d-1)$-simplex $\sigma\times\{1\}\setminus\{(v_i,1)\}$ as $\sigma_{i}$. Observe that the vertices of $\sp{(\boundary\Delta^{d})}$ are $(\sigma\times\{0,1\})\cup\{v_{\sigma_{i}} : 0\leq i\leq d\}$. The building block $B_{d}$ is the simplicial complex obtained by replacing each vertex $v_{\sigma_{i}}$ in $\sp{(\boundary\Delta^{d})}$ with the vertex $(v_i,1)$. See the main body of the text for a more intuitive description (but less formal) of $B_d$.
\par 
We now prove the relevant properties of $B_d$. 

\Bddchain*

\begin{proof}[Proof of Lemma \ref{lem:Bd_d_chain}, Part 1]
    To keep track of the simplices, we introduce some notation. Recall that $\sigma_{i} = \sigma\setminus\{v_i\}$. Additionally, we will denote $\sigma_{ij}=\sigma\setminus\{v_i,v_j\}$. For $0\leq k\leq d$, denote the set 
    $$
    \sigma^{k} := \{(v_0,0),\ldots,(v_k,0),(v_k,1),\ldots,(v_d,1).
    $$
    Note that $\sigma^{k}$ is not a simplex in any of the complexes we consider, but some of its subsets will be. For distinct values $i,j\neq k$, denote the simplex $\sigma_{ij}^{k} := \sigma_{k}\setminus\{(v_i,0),(v_i,1),(v_j,0),(v_j,1)\}$. Note that the subtraction in this definition is redundant, as only one of the vertices $(v_i,0)$ or $(v_i,1)$ will be contained in $\sigma^{k}$. 
    \par 
    First, we count the number of $d$-simplices in $SP(\boundary\Delta^{d})$. Fix a $(d-1)$-simplex $\sigma_i$, and consider a $(d-2)$-simplex $\sigma_{ij}$ in its boundary. The $(d-1)$-simplices in $P\sigma_{ij}$ are those of the form $\sigma_{ij}^{k}$ for $0\leq k\leq d$, $k\neq i,j$. The $d$-simplices in $SP\sigma_{i}$ are those simplices of the form $\sigma^{k}_{ij}\cup\{v_{\sigma_i}\}$. Additionally, $SP(\boundary\Delta^{d})$ contains those simplices $(\sigma_{i}\times\{0\})\cup\{v_{\sigma_i}\}$. Therefore, the number of $d$-simplices in $SP(\boundary\Delta^{d})$ equals the number of choices of $i$, $j$, and $k$, plus $d+1$ for the $d$-simplices of the form $(\sigma_i\times\{0\})\cup\{v_{\sigma_i}\}$. In total, there are $(d+1)\cdot d\cdot (d-1) + (d+1) \in \Theta((d+1)^{3})$ $d$-simplices in $SP(\boundary\Delta^{d})$
    \par 
    Now we count the number of $d$-simplices in $B_d$. As we will see, replacing the vertex $v_{\sigma_i}$ will reduce the number of $d$-simplices, but only by a constant factor. In the case that $i>k$, then $\sigma^{k}_{ij}\cup\{(v_i,1)\} = \sigma^{k}_{j}$, but in the case that $i<k$, no such simplification is possible. Therefore, the $(d+1)$-simplices in $B_d$ fall into two sets: $\{\sigma_{j}^{k} : 0\leq k< d, j\neq k\}$ and $\{\sigma_{ij}^{k}\cup\{(v_i,1)\} : 0\leq k\leq d, j\neq k, i<k\}$. Note that the simplices $\sigma^{k}_{j}$ may arise from different stellar prisms $SP\sigma_{i}$ for different values of $i$, whereas the simplices of the form $\sigma_{ij}^{k}\cup\{v_i\}$ only arise from the stellar prism $SP\sigma_{i}$. However, as there are $O((d+1)^{2})$ simplices of the form $\sigma_{j}^{k}$, $O((d+1)^{3})$ simplices of the form $\sigma_{ij}^{k}\cup\{(v_i,1)\}$, and $d+1$ simplices of the form $(\sigma_i\times\{0\})\cup\{(v_i,1)\}$, we conclude that there are $\Theta((d+1)^{3})$ $d$-simplices in $B_d$.
\end{proof}

\begin{proof}[Proof of Lemma \ref{lem:Bd_d_chain}, Part 2]
    The simplicial complex $B_d$ is the one described in the paragraphs preceding this proof, which is derived from the simplicial complex $SP(\boundary\Delta^{d})$ by replacing some vertices. We first describe the chain $f$ in the complex $SP(\boundary\Delta^{d})$. We then analyse this chain after we replace the vertices.
    \par 
    Let $\sigma=\{v_0,\ldots,v_d\}$. The chain $f$ is defined  $f:=SP_{*}(\boundary\sigma).$ It is important to note that while the simplex $\sigma$ is not contained in $\boundary\Delta^{d}$, its boundary $\boundary\sigma$ is still a well-defined $(d-1)$-chain in $C_{d-1}(\boundary\Delta^{d})$. By Lemma \ref{lem:stellar_prism_map}, we know that 
    \begin{align*}
        \boundary f = \boundary SP_{*}(\boundary\sigma) =& S_{*}I_{1}\boundary(\sigma) - I_{0}\boundary(\sigma) - SP_{*}\boundary\boundary(\sigma) \\
        =& S_{*}I_{1}\boundary(\sigma) - I_{0}\boundary(\sigma)
    \end{align*}
    We know $I_{0}(\boundary\sigma) = \boundary(\sigma\times\{0\})$ by the definition of $I_0$, so to finish the proof, we only need to verify that $S_{*}I_{1}\boundary(\sigma) = - d\cdot\boundary(\sigma\times\{1\})$ after we replace the vertices.
    \par 
    For this, we separately consider the value of $f$ on each $d$-simplex in its support.  Recall the notation $\sigma_i:=\sigma\setminus\{v_i\}$. We can expand this chain using the definition of the boundary map as
    \begin{align*}
        S_{*}I_{1}(\boundary\sigma) = \sum_{i=0}^{d} (-1)^{i} S_{*}I_{1}(\sigma_i).
    \end{align*}
     Let us now investigate the terms $S_{*}I_{1}(\sigma_i)$.
    \par 
    Recall that the operator $b_{\sigma_{i}}$ assigns different signs to a simplex depending on where $v_{\sigma_{i}}$ is in the order of the simplex's vertices, so assume the vertex $v_{\sigma_{i}}$ is ordered between $(v_{i-1},1)$ and $(v_{i+1},1)$ in the ordering of the simplices. In this case, $v_{\sigma_{i}}$ is the $i$th vertex in the simplex $I_1(\sigma_{ij})\cup\{v_{\sigma_i}\}$ for $j<i$ and in the $(i+1)$st position for $j>i$. Therefore 
    \begin{align*}
        S_{*}I_1(\sigma_{i}) =& b_{\sigma_{i}}\boundary I_1(\sigma_{i}) \\
        =& b_{\sigma_{i}}\left(\sum_{j=0}^{i-1} (-1)^{j} I_{1}(\sigma_{ij}) + \sum_{j=i+1}^{d} (-1)^{j-1} I_{1}(\sigma_{ij}) \right) \\
        =& \sum_{j=0}^{i-1} (-1)^{j}(-1)^{i-1} I_{1}(\sigma_{ij})\cup\{v_{\sigma_i}\} + \sum_{j=i+1}^{d} (-1)^{j-1}(-1)^{i} I_{1}(\sigma_{ij})\cup\{v_{\sigma_{i}}\} \\
        =& (-1)^{i-1}\left[\sum_{j=0}^{i-1} (-1)^{j} I_{1}(\sigma_{ij})\cup\{v_{\sigma_i}\} + \sum_{j=i+1}^{d} (-1)^{j} I_{1}(\sigma_{ij})\cup\{v_{\sigma_{i}}\} \right]
    \end{align*}
    Therefore, when we replace the vertex $v_{\sigma_{i}}$ with $(v_{i},1)$, we find that
    \begin{align*}
        S_{*}I_{1}(\boundary\{v_0,\ldots,v_d\}) = \sum_{i=0}^{d}(-1)\left[\sum_{j=0}^{i-1} (-1)^{j} I_{1}(\sigma_{j}) + \sum_{j=i+1}^{d} (-1)^{j} I_{1}(\sigma_{j}) \right]
    \end{align*}
    Note that each term $I_{1}(\sigma_{j})$ appears $d$ times in this sum, once for each $i\neq j$. Moreover, each time the term appears, it has sign $(-1)^{j}$. We conclude that $S_{*}I_{1}(\boundary\{v_0,\ldots,v_d\}) = -d\cdot\boundary(\sigma\times\{1\})$, as claimed
    \par 
    We now analyze the size of $f$. We know that $f = SP_{*}(\boundary\sigma)$. We can analyze this sum linearly as 
    \begin{align*}
        f = SP_{*}(\boundary\sigma) =& \sum_{i=0}^{d} (-1)^{i} SP_{*}(\sigma_{i}) \\
        =& \sum_{i=0}^{d} (-1)^{i} (-b_{\sigma_{i}}I_0(\sigma_{i}) - b_{\sigma_{i}}P_{*}\boundary(\sigma_{i})) 
    \end{align*}
    By the first term in the sum, the chain $f$ assigns value $\pm 1$ to each $d$-simplex $(\sigma_{i}\times\{0\}) \cup \{v_i\}$. Likewise, were we to expand the second sum, we would find that $f$ assigns value $\pm 1$ to all simplices of the form $\sigma^{k}_{ij}\cup\{(v_i,1)\}$. By Part $1$ of the lemma, we know there are $\Theta((d+1)^{3})$ such simplices. We conclude that the squared norm $\|f\|^{2}$ is $\Omega(({d+1})^{3})$. 
\end{proof}

\section{Collapsibility of the complex.}
\label{apx:collapsibility}

In this section, we will prove that the complexes $\mathcal{B}_d^{n}$ and $\mathcal{P}_{d}^{n}$ from Section~\ref{sec:lower_bounds} collapse to $(d-1)$-dimensional subcomplexes. This will imply that $\ker\boundary_d[\mathcal{B}_d^{n}]=0$ and $\ker\boundary_d[\mathcal{P}_d^{n}]=0$.

\subsection{Preliminaries.}
\noindent
In this section, we introduce the necessary background on collapsiblity.

A $(d-1)$-simplex $\tau$ is a \EMPH{face} of a $d$-simplex $\sigma$ iff $\tau\subset\sigma$.
The pair $(\sigma, \tau)$ is a \EMPH{collapse pair} in $K$ iff (i) $\tau$ is a face of $\sigma$, and (ii) $\tau$ is not the face of any other simplex in $K$.
The complex $K$ \EMPH{collapses} to the complex $K\backslash\{\sigma, \tau\}$ if $(\sigma, \tau)$ is a collapse pair in $K$. More generally, a complex $K$ collapses into a complex $K''$ if there exists a complex $K'$ such that $K$ collapses to $K'$ and $K'$ collapses to $K''$.  A complex $K$ is \EMPH{collapsible} if it collapses into a single vertex.

By the inductive definition of collapsibilty, whenever a complex $K$ collapse into a complex $L$, there exists a sequence of complexes $K = K_0\supset K_1\supset\ldots\supset K_k = L$ such that for any $0<i\leq k$, $K_{i} = K_{i-1}\backslash\{\sigma_i, \tau_i\}$ where $(\sigma_i, \tau_i)$ is a collapse pair in $K_{i-1}$.  This sequence is called a \EMPH{collapsing sequence}.

We now give two facts about consequences of collapsibility we use in this paper. \Cref{lem:collapsibility_implies_homotopy_equivalence} is a standard fact about collapsibility we state without proof. \Cref{lem:collapsibility_and_null_homology} is a non-standard fact about collapsibility, but one that will come as no surprise to those familiar with collapsibility.

\begin{lemma}
\label{lem:collapsibility_implies_homotopy_equivalence}
    Let $L\subset K$ be simplicial complexes such that $K$ collapses to $L$. Then $L$ is a deformation retract of $K$. Consequently, $L$ and $K$ have isomorphic homology groups in all dimensions.
\end{lemma}

\begin{lemma}
\label{lem:collapsibility_and_null_homology}
    Let $L\subset K$ be simplicial complexes such that $K$ collapses to $L$. Let $\gamma$ be a $(d-1)$-cycle in $L$. Then $\gamma$ is null-homologous in $L$ if and only if $\gamma$ is null-homologous in $K$.
\end{lemma}
\begin{proof}[Proof of \Cref{lem:collapsibility_and_null_homology}]
    Suppose that $\gamma$ is null-homologous in $L$, and let $f_L$ be a $d$-chain in $C_d(L)$ such that $\boundary f_L = \gamma$. We can extend $f$ to a $d$-chain $f_K$ in $K$ by setting $f_K(\sigma)=0$ for any $d$-simplex $\sigma\in K_d\setminus L_d$. It is straightforward to see that $\boundary f_K=\gamma$ by the linearity of the boundary map, so $\gamma$ is null-homologous in $K$ 
    \par 
    Conversely, suppose that $\gamma$ is null-homologous in $K$, and let $f_K$ be a $d$-chain in $C_d(K)$ such that $\boundary f_K = \gamma$. Consider a collapsing sequence $K = K_0\supset K_1\supset\ldots\supset K_k = L$ with collapses pairs $(\sigma_i,\tau_i)$. We will prove by induction on $i$ that there is a chain $f_{K_i}\in C_d(K_i)$ for $0\leq i\leq k$ such that $\boundary f_{K_i} = \gamma$. This will imply that $\gamma$ is null-homologous in $L$.
    \par 
    For the base case of $i=0$, the chain $f_{K_0} = f_{K}$ satisfies the claim by assumption. Now suppose there is a chain $f_{K_i}\in C_d(K_i)$ such that $\boundary f_{k_i} = \gamma$. Consider the collapse pair $(\sigma_i, \tau_i)$. In the case that neither $\sigma_i$ nor $\tau_i$ are $d$-simplices, then $K_i$ and $K_{i+1}$ have the same set of $d$-simplices, so $f_{K_i}$ is a valid $d$-chain in $C_d(K_i)$ and we set $f_{K_{i+1}} = f_{K_i}$. In the case that $\sigma_i$ is a $d$-simplex, then $\tau$ is a $(d-1)$-simplex. As $\tau\not\in L_{d-1}$, then $\gamma(\tau)=0$. This implies that it must be the case that $f_{K_i}(\sigma_i)=0$. If not, then $\boundary f_{K_i}(\sigma_i) = \pm f_{K_i}(\sigma_i)$ (a contradiction) as the only $d$-simplex incident to $\tau_i$ is $\sigma_i$. Therefore, we can set $f_{K_{i+1}}=f_{K_i}$. Finally, in the case that $\tau_i$ is a $d$-simplex, then $\sigma_i$ is a $(d+1)$-simplex. In this case, we set $f_{K_{i+1}} = f_{K_{i}} - (\boundary\sigma_i(\tau_i))\cdot(f_{K_{i}}(\tau_i))\cdot \boundary\sigma_i$. (Note that in this expression, $\boundary\sigma_i(\tau_i)=\pm 1$, $f_{K_{i}}(\sigma_i)$ is a scalar, and $\boundary\sigma_i$ is a $d$-chain.) We need to verify two things about $f_{K_{i+1}}$: (1) $\boundary f_{K_{i+1}} = \gamma$ and (2) $f_{K_{i+1}}(\tau_i)=0$. Condition (1) is easy to verify as 
    \begin{align*}
        \boundary f_{K_{i+1}} &= \boundary f_{K_{i}} - (\boundary\sigma_i(\tau_i))\cdot(f_{K_{i}}(\sigma_i))\cdot \boundary\boundary\sigma_i \\
         &= \boundary(f_{K_{i}}) \tag{as $\boundary\boundary=0$} \\
          &= \gamma \tag{Induction Hypothesis} 
    \end{align*}
    Condition (2) is also straightforward to verify.
    \begin{align*}
        f_{K_{i+1}}(\tau_i) &= f_{K_{i}}(\tau_i) - (\boundary\sigma_i(\tau_i))\cdot(f_{K_{i}}(\tau_i))\cdot \boundary\sigma_i(\tau_i) \\ 
         &= f_{K_{i}}(\tau_i) - f_{K_{i}}(\tau_i) \tag{as $(\boundary\sigma_i(\tau_i))^{2}=1$} \\ 
         &= 0 
    \end{align*}
    Therefore, $f_{K_{i+1}}$ is a $d$-chain in $C_d(K_{i+1})$ with boundary $\gamma$. 
    \par 
    In all three cases, we can find a $d$-chain $f_{K_{i+1}}$ such that $\boundary f_{K_{i+1}} \gamma$. This proves that $\gamma$ is null-homologous in $L$.
\end{proof}

\subsection{The main lemma.}

In this section, we will prove the following lemmas about the complexes $\mathcal{B}_d^{n}$ and $\mathcal{P}_d^{n}$ from Section \ref{sec:lower_bounds}.

\begin{restatable}{lemma}{Bdncollapses}
\label{lem:Bdn_collapses}
    The simplicial complex $\mathcal{B}_d^{n}$ collapses to a $(d-1)$-dimensional subcomplex.
\end{restatable}
\begin{restatable}{lemma}{Pdncollapses}
\label{lem:Pdn_collapses}
    The simplicial complex $\mathcal{P}_d^{n}$ collapses to a $(d-1)$-dimensional subcomplex. 
\end{restatable}

For this, we need to prove an auxiliary lemma about the building block from Section~\ref{sec:construction_building_block}.

\begin{lemma}
\label{lem:Bd_collapses}
    There is a collapsing sequence that collapses all $d$-simplices of $B_d$. Furthermore, no $(d-1)$-simplices that are subsets of $\sigma\times\{1\}$ are involved in the collapsing sequence.
\end{lemma}
\begin{proof}
    This proof relies on notation introduced in the proof of Lemma \ref{lem:Bd_d_chain} Part 1. Recall that the $d$-simplices of $B_d$ are of the form $\{\sigma_{j}^{k} : 0\leq k< d, j\neq k\}$,  $\{\sigma_{ij}^{k}\cup\{(v_i,1)\} : 0\leq k\leq d, j\neq k, i<k\}$, and $\{(\sigma_i\times\{0\})\cup\{(v_i,1)\} : 0\leq i \leq d\}$.
    We will collapse these $d$-simplices using the $(d-1)$-simplices of the form $\sigma^{k}_{j}\setminus\{(v_k,1)\}$, $\sigma_{ij}^{k}\cup\{(v_i,1)\}\setminus\{(v_k,1)\}$, and $\sigma_i\times\{0\}$. To define an appropriate collapsing sequence, we need to know which $d$-simplices are incident to which $(d-1)$-simplices.
    \par 
    The $(d-1)$-simplices $\sigma_i\times\{0\}$ are incident to exactly one $d$-simplex: $(\sigma_i\times\{0\})\cup\{(v_i,1)\}$. We can see this as any other $d$-simplex will have two vertices with second coordinate 1.
    \par 
    The simplices $\sigma^{k}_{ij}\cup\{(v_i,1)\}\setminus\{(v_k,1)\}$ are incident to exactly two $d$-simplices. It is straightforward to verify that $\sigma^{k}_{ij}\cup\{(v_i,1)\}\setminus\{(v_k,1)\}$ is incident to $\sigma^{k}_{ij}\cup\{(v_i,1)\}$ and $\sigma^{k+1}_{ij}\cup\{(v_i,1)\}$ for $j\neq k,k+1$ and $k<d$; we can see this as $\sigma^{k}_{ij}\cup\{(v_i,1)\}\setminus\{(v_k,1)\} = \sigma^{k+1}_{ij}\cup\{(v_i,1)\}\setminus\{(v_{k+1},0)\}$. For the case of special case of $j=k+1$, we can see that $\sigma^{k}_{i(k+1)}\cup\{(v_i,1)\}\setminus\{(v_k,1)\} = \sigma^{k+2}_{i(k+1)}\cup\{(v_i,1)\}\setminus\{(v_{k+2},0)\}$, so $\sigma^{k}_{i(k+1)}\cup\{(v_i,1)\}\setminus\{(v_k,1)\}$ is instead incident to $\sigma^{k}_{i(k+1)}\cup\{(v_i,1)\}$ and $\sigma^{k+2}_{i(k+1)}\cup\{(v_i,1)\}$. Finally, when $k=d$, the simplices $\sigma^{d}_{ij}\cup\{(v_i,1)\}\setminus\{(v_d,1)\}$ are incident to $\sigma^{d}_{ij}\cup\{(v_i,1)\}$ and $(\sigma_i\times\{0\})\cup\{(v_i,1)\}$.
    \par 
    Lastly, the simplex $\sigma^{k}_j\setminus\{(v_k,1)\}$ is incident to exactly three $d$-simplices. We see that $\sigma^{k}_{j}\setminus\{(v_{k},1)\}=\sigma^{k+1}_{j}\setminus\{(v_{k+1},0)\} = \sigma^{k+2}_{(k+1)j}\cup\{(v_{k+1},1)\}\setminus\{(v_{k+2},0)\}$ when $k\neq j+1$ and $k < d-1$, and $\sigma^{k}_{k+1}\setminus\{(v_{k},1)\}=\sigma^{k+2}_{k}\setminus\{(v_{k+2},0)\} = \sigma^{k+3}_{(k+2)j}\cup\{(v_{k+2},1)\}\setminus\{(v_{k+3},0)\}$ when $k = j+1$ and $k < d-1$. Finally, we consider the case of $k=d-1$. By the construction of $B_d$, we know that $\sigma^{d-1}_{j} = \sigma^{d-1}_{ij}$ for some $i$. We know that $i>d-1$, so we conclude that $i=d$. This implies $\sigma_{j}^{d-1}\setminus\{(v_{d-1},1)\} = \sigma^{d-1}_{dj}\cup\{(v_d,1)\}\setminus\{(v_{d-1},1)\}$ is incident to the two $d$-simplices: $\sigma^{d-1}_{j}$ and $(\sigma_{d}\times\{0\})\cup\{v_{d},1\}$ 
    \par 
    We now describe a collapsing sequence for $B$. We know that $\sigma_i\times\{0\}$ is only incident to $\sigma_i\times\{0\}\cup\{(v_i,1)\}$. We therefore collapse each simplex $\sigma_i\times\{0\}$ onto $\sigma_i\times\{0\}\cup\{(v_i,1)\}$.
    \par 
    Now fix $0\leq i\neq j\leq d$. For $k$ starting at $d$ and iterating backwards to $i$, we can collapse the simplices $\sigma^{k}_{ij}\cup\{(v_i,1)\}\setminus\{(v_k,1)\}$ into the simplex $\sigma^{k}_{ij}\cup\{(v_i,1)\}$; this collapse is valid as the other simplex incident to $\sigma^{k}_{ij}\cup\{(v_i,1)\}\setminus\{(v_k,1)\}$ was collapsed in a previous iteration. These collapses remove all simplices of the form $\sigma_{ij}^{k}\cup\{(v_i,1)\}$.
    \par
    Now fix an integer $j$. For $k$ starting at $d-1$ and iterating backwards to $0$ (and skipping $j$), we will collapse $\sigma^{k}_{j}\setminus\{(v_j,1)\}$ into $\sigma^{k}_j$. Initially, this is a valid collapse as the simplex $\sigma^{d-1}_{j}\setminus\{(v_{d-1},1)\}$ is only incident to $\sigma^{d-1}_{j}$ and $(\sigma_d\times\{0\})\cup \{(v_d,1)\}$, and the second simplex has already been collapsed. For the other values of $k$, the only simplices incident to $\sigma^{k}_j\setminus\{(v_k,1)\}$ were either removed in the previous iteration or in the series of collapses from the previous paragraph. 
    \par 
    Finally, no $(d-1)$-simplex in $K\times\{1\}$ was collapsed, as all of the simplices that were collapsed have at least one vertex with second coordinate 0.
\end{proof}

We can now prove Lemmas \ref{lem:Bdn_collapses} and \ref{lem:Pdn_collapses}.

\begin{proof}[Proof of Lemma \ref{lem:Bdn_collapses}]
    Recall that $\mathcal{B}_d^{n}$ is constructed by gluing together a $d$-simplex, denoted $B_d^{0}$, and $n$ copies of $B_d$, denoted $B_d^{i}$ for $1\leq i\leq n$. We will therefore prove that we can collapse each of the $d$-simplices of $B_d^{i}$ for $n\geq i\geq 1$ by induction on $i$ in reverse order.
    \par 
    For the base case of $i=n$, we know we can collapse each of the $d$-simplices of $B_d^{n}$. By Lemma \ref{lem:Bd_collapses}, we know we can collapse all $d$-simplices $B_d$, and the collapsing sequence does not collapse any $(d-1)$-simplex on the vertices $\sigma\times\{1\}$. As the only $(d-1)$-simplices in $B_d^{n}$ that are incident to other $d$-simplices are simplices in $\sigma\times\{n-1\}$, then we know we can collapse all $d$-simplices of $B_d^{n}$.
    \par 
    Inductively, we know that none of the $(d-1)$-simplices needed to collapse $B^{i}_{n}$ have previously been collapsed, as the only $(d-1)$-simplices of $B^{i}_n$ that are $(d-1)$-simplices of $B^{j}_{n}$ for $j>i$ are subsets of $\sigma\times\{1\}$ in $B^{i+1}_n$, which we know have not been collapsed. Finally, we can collapse the unique $d$-simplex in $B_{d}^{0}$, as none of the $(d-1)$-simplices of $B_d^{0}$ have previously been collapsed.
\end{proof}

\begin{proof}[Proof of Lemma \ref{lem:Pdn_collapses}]
    The proof of this theorem is nearly identical to the proof of~\Cref{lem:Bdn_collapses} above, so we exclude details and instead sketch the proof. Recall that $\mathcal{P}_d^{n}$ is constructed by gluing together n copies of $B_d$, denoted $B_d^{i}$ for $1\leq i\leq n$. As in the previous proof, we can prove that we collapse each of the $d$-simplices of $B_d^{i}$ for $1\leq i\leq n$ by induction. This works because, inductively, the only $(d-1)$-simplices in $B_d^{i}$ that are incident to $d$-simplices outside of $B_d^{i}$ are subsets of $\sigma\times\{1\}$. As in the previous proof, this is sufficient to prove that we can collapse the $d$-simplices of $B_d^{i}$. 
\end{proof}

\end{document}